\documentclass[moor,nonblindrev]{informs1arxiv} 

\usepackage{natbib}
 \NatBibNumeric
 \def\bibfont{\small}%
 \bibpunct[, ]{[}{]}{,}{n}{}{,}%

\usepackage[colorlinks=true,breaklinks=true,bookmarks=true,urlcolor=blue,
     citecolor=blue,linkcolor=blue,bookmarksopen=false,draft=false]{hyperref}

\def\EMAIL#1{\href{mailto:#1}{#1}}
\def\URL#1{\href{#1}{#1}}         

\TheoremsNumberedThrough     

\EquationsNumberedThrough    

\usepackage{graphicx}
\allowdisplaybreaks
\usepackage{amssymb}
\usepackage{bbm}
\usepackage{bm} 
\usepackage{enumitem}
\usepackage{tikz-cd}
\usepackage{cprotect}

\renewcommand{\Pr}{\mathbb{P}} 
\newcommand{\expect}{\mathbb{E}} 

\newcommand{\satlinks}{\mathcal{L}_s}
\newcommand{\steady}{\overline{\bm{N}}}
\newcommand{\steadyperp}{\steady_{\perp}}
\newcommand{\steadyparaspace}{\steady_{\shortparallel}^s}
\newcommand{\nperp}{\bm{n}_{\perp}}
\newcommand{\npara}{\bm{n}_{\shortparallel}}
\newcommand{\nperpspace}{\bm{n}_{\perp}^s}
\newcommand{\nparaspace}{\bm{n}_{\shortparallel}^s}
\newcommand{\weight}{\kappa}
\newcommand{\servrate}{\mu}

\newcommand{\Phase}{K}
\newcommand{\phase}{k}
\newcommand{\transition}{S}

\usepackage{etoolbox}
\newtoggle{complete}
\toggletrue{complete}

\begin{document}


\RUNAUTHOR{Wang et al.}

\RUNTITLE{Heavy-Traffic Insensitive Bounds}

\TITLE{Heavy-Traffic Insensitive Bounds for Weighted Proportionally Fair Bandwidth Sharing Policies}

\ARTICLEAUTHORS{%
\AUTHOR{Weina Wang}
\AFF{Carnegie Mellon University, \EMAIL{weinaw@cs.cmu.edu}, \URL{}}
\AUTHOR{Siva Theja Maguluri}
\AFF{Georgia Institute of Technology, \EMAIL{siva.theja@gatech.edu}, \URL{}}
\AUTHOR{R.\ Srikant}
\AFF{University of Illinois at Urbana-Champaign, \EMAIL{rsrikant@illinois.edu}, \URL{}}
\AUTHOR{Lei Ying}
\AFF{University of Michigan, \EMAIL{leiying@umich.edu}, \URL{}}
} 

\ABSTRACT{%
We consider a connection-level model proposed by Massouli\'{e} and Roberts for bandwidth sharing among file transfer flows in a communication network.  We study weighted proportionally fair sharing policies and establish explicit-form bounds on the weighted sum of the expected numbers of flows on different routes in heavy traffic.  The bounds are linear in the number of critically loaded links in the network, and they hold for a class of phase-type file-size distributions; i.e., the \emph{bounds} are \emph{heavy-traffic insensitive} to the distributions in this class.  Our approach is Lyapunov-drift based, which is different from the widely used diffusion approximation approach.  A key technique we develop is to construct a novel inner product in the state space, which then allows us to obtain a multiplicative type of state-space collapse in steady state.  Furthermore, this state-space collapse result implies the interchange of limits as a by-product for the diffusion approximation of the \emph{equal-weight} case under phase-type file-size distributions, demonstrating the heavy-traffic insensitivity of the \emph{stationary distribution}.
}%



\maketitle

\section{Introduction.}\label{sec-intro}
We study the following connection-level model proposed by \citet{MasRob_00} for bandwidth sharing among file transfer flows in a communication network, illustrated in Figure~\ref{fig-network} {(see details in Section~\ref{sec-model})}.  File transfer requests arrive to a network, and the transfer of each file, also referred to as a \emph{flow}, is through a predetermined route that consists of a set of consecutive links connecting the source and the destination. Each link in the network has a finite bandwidth capacity that needs to be allocated to the flows on the link by a bandwidth sharing policy. The bandwidth/rate received by a flow determines the speed at which its data can be transferred, thus determines the delay of the file transfer, namely the time from when the file arrives until the completion of the transfer.

In this paper, we consider an important class of bandwidth sharing policies called \emph{weighted proportionally fair sharing policies} \cite{Kel_97,MoWal_00}, where each route is associated with a weight that represents the importance of this route.  When the weights of all the routes are equal, the policy is simply called the \emph{proportionally fair sharing policy}.  We are interested in characterizing the delay performance, reflected by the numbers of files present on the routes in steady state.  We study the \emph{heavy-traffic} regime, where the loads on some links are close to their capacities.  Heavy-traffic analysis is an approach that has been widely adopted to study queueing systems.  It provides approximations on the performance of a system and gives insights into policy design since it examines a policy in the critical scenario of heavy load.

\subsection{Existing work.}
Bandwidth sharing policies have been extensively studied in the literature.  See a recent survey paper by \citet{Wil_16} for an overview of this topic.

The question of stability was first posed by \citet{deVLeeKon_01}, and they gave results for weighted proportionally fair policies and another family of policies called weighted max-min fair policies under the assumption of exponentially distributed file sizes.  \citet{BonMas_01} generalized the results to a more general family of policies called weighted $\alpha$-fair policies \cite{MoWal_00}, which includes weighted proportionally fair policies and max-min fair policies as special cases (corresponding to $\alpha=1$ and $\alpha=\infty$, respectively).  Stability results for more general file size distributions have also been derived \cite{LakBecSri_05,Mas_07,Bra_10,GroWil_08,PagTanFer_12}.
In particular, \citet{Mas_07} showed the stability of the (equal-weight) proportionally fair policy under phase-type file size distributions, and \citet{PagTanFer_12} showed the stability of a fluid limit of this system under general file size distributions (although the fluid approximation of the original stochastic system itself has not been rigorously established). 

Heavy-traffic analysis of weighted proportionally fair policies has mainly focused on the \emph{diffusion approximation} approach.  Assuming exponentially distributed file sizes and Poisson arrivals, \citet{KanKelLee_09} derived the diffusion approximations for weighted proportionally fair policies under a \emph{local traffic assumption}, which requires that each link has at least one route that uses this link only.  The local traffic assumption may not be appropriate in some application scenarios such as in data centers.  \citet{YeYao_12} replaced the local traffic assumption with a weaker assumption that requires the routing matrix to have a full row-rank.  We refer to this assumption as the \emph{full-rank assumption}. For the (equal-weight) proportionally fair policy, \cite{KanKelLee_09} and \cite{YeYao_12} obtained the stationary distribution of the diffusion process explicitly, using results on product-form invariant measures for diffusion processes in \cite{HarWil_87,Wil_87}.  For weighted proportionally fair policies, explicit stationary distributions of the limiting diffusion processes are unknown.

In the above work \cite{KanKelLee_09} and \cite{YeYao_12}, the convergence of the scaled systems to the diffusion process is shown over finite time intervals. To establish that the stationary distribution of the original system can be approximated by the stationary distribution of the diffusion process, a so-called \emph{interchange-of-limits} argument is required.  \citet{ShaTsiZho_14} and \citet{YeYao_16} proved the interchange-of-limits for the diffusion approximations in \cite{KanKelLee_09} under the local traffic assumption and in \cite{YeYao_12} under the full-rank assumption, respectively.

A question that is of importance both in theory and in practice is whether the performance metrics of a policy are \emph{insensitive}, where insensitivity means no dependence on the forms of the file size distributions except for their means. Insensitivity is a highly desirable property since file size distributions in practice may not be exponential, and they may change over time with the evolution of application scenarios.  
{\citet{BonPro_03} showed that a necessary and sufficient condition for a bandwidth sharing policy to have an insensitive stationary distribution is the so-called balance property, which is generally not satisfied by the proportionally fair policy except for very special network structures.  \citet{BonPro_03} further proposed an insensitive policy named balanced fairness policy that is maximally stable and Pareto-efficient.  However, although the proportionally fair policy does not have the balance property in general, \citet{Mas_07} uncovered a deeper connection between proportional fairness and balanced fairness.  Namely, \citet{Mas_07} noted that balanced fairness converges to proportional fairness when the number of flows are large if such a limit exists.  Additionally, \citet{Mas_07} proposed a modified proportional fairness policy that has the same large-deviation rate function as the balanced fairness, and also converges to proportional fairness when the number of flows is large.  \citet{Wal_11} further proved that, when the number of flows goes to infinity, the only possible limit of a maximum stable, insensitive policy is the proportionally fair policy.
}

A relaxed goal of interest is \emph{insensitivity in heavy-traffic}.  Notably, \citet{VlaZhaZwa_14} recently derived the diffusion approximation for the proportionally fair policy under a class of phase-type file size distributions, and showed that the stationary distribution of the diffusion process is insensitive to the distributions in this class.  However, an interchange-of-limits argument was not provided, i.e., the result in \cite{VlaZhaZwa_14} only holds over finite time intervals.
{ In an earlier version \cite{WanMagSri_18} of this paper, we also considered the proportionally fair policy under phase-type file size distributions.  We derived bounds on the expected total number of flows in steady state, where the dominant terms in the bounds in heavy traffic are insensitive to the file size distributions.  In this paper, we focus on the more general weighted proportionally fair sharing, under a much broader
heavy-traffic regime.
}
{A weighted proportionally fair policy is in general sensitive in heavy traffic.  It is known that in the setting of a single link, where a weighted proportionally fair policy is known as discriminatory processor-sharing, the limiting distribution in heavy traffic depends on the second moments of file size distributions \cite{RegSen_96,vanNunBor_04}.
}

\subsection{Our results.}
\begin{sloppypar}
We take an approach that is different from the diffusion approximation approach, where we directly analyze the steady state of the system.  We obtain explicit-form bounds on the weighted sum of the expected numbers of flows on different routes.  Specifically, let $\overline{N}_r$ denote the number of flows on route $r$ in steady state, and $\weight_r$ denote the weight associated with route $r$.  We show the following bounds, assuming a class of phase-type distributions for the file sizes, Poisson arrivals, and the full-rank assumption:
\begin{equation}\label{eq-bounds}
\frac{L_s\cdot \min_r \weight_r}{\epsilon}+o\biggl(\frac{1}{\epsilon}\biggr)\le \expect\Biggl[\sum_r \weight_r \overline{N}_r\Biggr]
\le\frac{L_s\cdot \max_r \weight_r}{\epsilon}+o\biggl(\frac{1}{\epsilon}\biggr),
\end{equation}
where $L_s$ is the number of critically loaded links in the network in the heavy-traffic regime, and $\epsilon>0$ is the heavy-traffic parameter, depending only on the \emph{mean} file sizes and representing how far away the traffic load is from the boundary of the system capacity.  We note that in the weighted proportionally fair policy, the rates allocated to the flows remain the same if we multiply all the weights by the same number.  So without loss of generality we may assume that the weights are normalized such that $\max_r\weight_r=1$.
\end{sloppypar}

The weighted sum in \eqref{eq-bounds}, i.e., $\expect\bigl[\sum_r \weight_r \overline{N}_r\bigr]$, 
can be viewed as a cost for delay incurred by the system when each flow on route $r$ has a cost of $\weight_r$ for each unit of its delay.  A notable feature of our bounds is that these unit costs are the same as the weights in the weighted proportionally fair policy.  Suppose the performance goal of the system is to provide a guarantee on this total cost, with the unit costs representing the relative importance of the flows on different routes.  Then our bounds imply that the system is able to achieve this goal with a weighted proportionally fair policy by choosing the weights in the policy to be the same as the unit costs.

Furthermore, our bounds hold for a class of phase-type file size distributions that can approximate any file size distribution arbitrarily closely.  The dominant terms in the upper and lower bounds, i.e., $L_s\cdot\max_r\weight_r/\epsilon$ and $L_s\cdot\min_r\weight_r/\epsilon$ respectively, do not depend on the specific forms of the file size distributions in this class except for their means.  Therefore, we say that the upper and lower \emph{bounds} are \emph{heavy-traffic insensitive}.  For the proportionally fair policy where the weights are equal, the dominant terms in the upper and lower bounds coincide, and thus the total expected number of flows is heavy-traffic insensitive.  
{Interestingly, for weighted proportionally fair policies, when we consider the special case of a single link, the dominant terms in our upper and lower bounds match the upper and lower bounds derived for the heavy-traffic limiting distribution in discriminatory processor-sharing \cite{AalAyeBor_07}.  But to the best of our knowledge, results of explicit forms for weighted proportionally fair policies in general networks are scarce.
}
Note that while these results yield heavy-traffic insensitive bounds, they do \emph{not} imply that the weighted sum of expected numbers of flows under a general weighted proportionally fair policy is heavy-traffic insensitive.

Our bounds for the (equal-weight) proportionally fair policy complements the diffusion approximation result in \cite{VlaZhaZwa_14} since it justifies the validity of the approximation on the total expected number of flows given by the diffusion process in steady state. Furthermore, our state-space collapse result, which will be explained in the next section, implies the interchange of limits as a by-product for the diffusion approximation in \cite{VlaZhaZwa_14}, thus demonstrating the heavy-traffic insensitivity of the \emph{distribution} of the numbers of flows on different routes.

We remark that our upper bound in \eqref{eq-bounds} scales linearly with the \emph{number of links}, while static planning for bandwidth sharing would yield a result that scales linearly with the \emph{number of routes}. This scaling behavior of weighted proportionally fair policies is very appealing in applications such as data centers and the Internet, since the number of links is typically several orders of magnitude smaller than the number of routes.

\subsection{Our techniques.}
\begin{sloppypar}
Our approach is under a framework called the \emph{drift method} for studying the steady state of a queueing system in heavy traffic, developed in \cite{ErySri_12} and \cite{MagSri_16}. The basic idea is to obtain bounds on expected steady-state queue lengths by setting the drift of an appropriately chosen Lyapunov function equal to zero in steady state. In this approach, a critical step is to establish a \emph{state-space collapse} result. In prior work of the drift method, the state-space collapse is in the following sense. Consider the steady state of the system, which usually lies in a multi-dimensional vector space, and a lower-dimensional subspace of the state space. The state space is said to \emph{collapse} to this lower-dimensional subspace if the moments of the distance between the steady state and the lower-dimensional subspace are upper bounded by \emph{constants} as the heavy-traffic parameter $\epsilon$ goes to $0$. This intuitively means that the steady state concentrates around the lower-dimensional subspace in heavy traffic, hence the term collapse. In \cite{ErySri_12,MagSriYin_14,WanZhuYin_16,XieLu_15}, the state spaces collapse to single-dimensional subspaces. Papers \cite{MagSri_16} and \cite{MagBurSri_16} generalized the drift method to the case where the state space collapses to a multi-dimensional subspace, and resolved the open problem on the scaling behavior of backlog in a switch under the MaxWeight algorithm.
\end{sloppypar}

In this paper, our state-space collapse result is of a slightly different type from the prior work above of the drift method.  In contrast to the constant moment bounds above for the distance to the lower-dimensional subspace, our state-space collapse result shows that these moments may grow to infinity as the heavy-traffic parameter $\epsilon$ goes to $0$, but at a speed slower than the corresponding moments of the norm of the state vector.  Specifically, the $m$-th moment of this distance grows as $O\bigl((1/\sqrt{\epsilon})^m\bigr)$, while the $m$-th moment of the norm of the state vector grows as $\Theta\bigl((1/\epsilon)^m\bigr)$.  Therefore, their ratio still goes to $0$.  In this sense, the state-space collapse in this paper is of a \emph{multiplicative} type, which has a similar flavor to the multiplicative state-space collapse in the diffusion approximation literature (see, e.g., \cite{Bra_98,Wil_98,KanKelLee_09}).
We remark that a recent work \cite{WanMagJav_17} that studies switches with reconfiguration delay also deals with multiplicative type of state-space collapse, but the technique there cannot quantify how fast the moments grow in terms of $\epsilon$.

A key technique that allows us to establish the state-space collapse for phase-type file size distributions is an inner product we construct, which is different from the usual dot product.  The construction is inspired by the Lyapunov function in \cite{PagTanFer_12}, {where it is used to establish fluid stability.}  But the induced norm under our constructed inner product is slightly different from the Lyapunov function in \cite{PagTanFer_12} since we need the inner product to satisfy some sufficient conditions for the heavy-traffic analysis.  Our constructed inner product rotates the space in a way such that the utilization of resources under a weighted proportionally fair policy is reflected by quantities with clear geometric interpretations. This enables us to study the dynamics of geometric quantities such as the aforementioned distance between the state vector and a lower-dimensional subspace, which are needed in the drift-based approach.

\subsection{Organization.}
The rest of this paper is organized as follows. We introduce the model and notation in details in Section~\ref{sec-model}. For ease of exposition, we start with results and proofs for exponential file size distributions in Sections~\ref{sec-main}--\ref{sec-proofs-main-exp}. Our main results for exponential file size distributions are summarized in Section~\ref{sec-main}, the proof outlines are given in Section~\ref{sec-proof-outlines} and the detailed proofs are given in Section~\ref{sec-proofs-main-exp}.  Here we intend to use the proof outlines in Section~\ref{sec-proof-outlines} as a tutorial for illustrating the steps in the drift method.
In Section~\ref{sec-phase-type}, we generalize our results to a general class of phase-type distributions and present the implication on interchange of limits for the diffusion approximation in \cite{VlaZhaZwa_14}. We conclude our paper in Section~\ref{sec-conclusions}.

\section{System model.}\label{sec-model}
\paragraph{Basic notation.}
Let $\mathbb{R}$, $\mathbb{R}_+$ and $\mathbb{Z}_+$ denote the sets of real numbers, nonnegative real numbers and nonnegative integers, respectively. Let $[K]$ denote the set $\{1,2,\dots,K\}$ for a positive integer $K$. We use $\bm{1}_{K\times 1}$ to denote a $K\times 1$ vector whose entries all equal $1$, and omit the subscript when the dimension is clear from the context. We use $\mathbbm{1}_{\{\cdot\}}$ to denote an indicator function that equals $1$ when the event in the subscript is true and equals $0$ otherwise.  Vectors are column vectors unless otherwise stated. We use ``$\Rightarrow$'' to denote weak convergence (convergence in distribution) of random elements.

\paragraph{Bandwidth sharing.}
\begin{figure}
\centering
\includegraphics[scale=0.6]{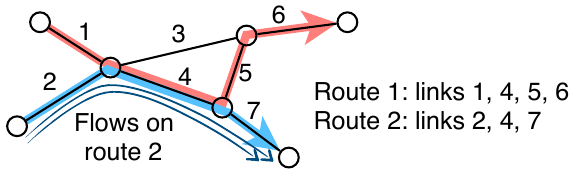}
\caption{A Bandwidth Sharing Network.}
\label{fig-network}
\end{figure}
We consider a network where nodes are connected by a set of \emph{links} $\mathcal{L}=\{1,2,\dots,L\}$, illustrated in Figure~\ref{fig-network}.  Data file transfer requests arrive to the system.  Each file transfer request, also referred to as a \emph{flow}, transfers a data file from a source node to a destination node, through a predetermined \emph{route} that consists of a set of consecutive links connecting the source and destination nodes. We consider a fixed set of routes indexed as $\mathcal{R}=\{1,2,\dots,R\}$. We write $\ell\in r$ if link~$\ell$ is on route~$r$. The relation between links and routes can be represented by the \emph{routing matrix} $H=(h_{\ell r})_{\ell\in\mathcal{L},r\in\mathcal{R}}$ with $h_{\ell r}=1$ if $\ell\in r$, and $h_{\ell r}=0$ otherwise. We assume that the routing matrix has full row-rank, referred to as the \emph{full-rank assumption}.

The system is operated in continuous time. Each link $\ell$ in the network has a bandwidth capacity $C_{\ell}$, which needs to be allocated to the flows on the link by a \emph{bandwidth sharing policy}. A bandwidth sharing policy specifies how much bandwidth/rate each flow receives according to the number of flows present on all the routes, subject to bandwidth capacity constraints. The rate a flow receives determines the speed at which its data can be transferred. We are interested in the delay of a file transfer, namely the time from when the file arrives until the completion of the transfer. Specifically, if we allocate a rate of $x(t)$ at time $t$ to a flow that arrives at time $A$ and has a file size $F$, then its delay $D$ is given by the equation $\int_{A}^{A+D}x(t)dt = F$; I.e., the transfer completes when the accumulative rate equals to the file size. Therefore, the bandwidth allocation policy affects the delay by specifying the rates $x(t)$'s for the flows.

\paragraph{Weighted proportionally fair policy.}
We consider an important class of bandwidth sharing policies called \emph{weighted proportionally fair policies}, where each route $r$ is associated with a positive weight $\weight_r$ that represents the importance of this route. When all the weights are equal (in which case we may assume that $\weight_r=1$ for all route $r$ without loss of generality), the policy is simply called the \emph{proportionally fair policy}.
We denote by $N_r(t)$ the total number of flows (flow count) on route $r$ at time $t$. A weighted proportionally fair policy allocates a rate of $x_r(t)$ to each flow on route $r$, where $(x_r(t))_{r\in\mathcal{R}}$ is the optimal solution of the following optimization problem with each $n_r$ equal to $N_r(t)$:
\begin{align}
\max_{(x_1,\dots,x_R)}
&\mspace{18mu} \sum_r \weight_rn_r\log x_r\label{eq-prop-fair-obj}\\
\text{subject to}&\mspace{13mu}\sum_{r:\ell\in r}n_rx_r\le C_{\ell},\forall \ell,\label{eq-prop-fair-constr-capacity}\\
&\mspace{18mu}x_r\ge 0,\forall r,
\end{align}
and $x_r(t)=0$ when $n_r=0$. The constraints in \eqref{eq-prop-fair-constr-capacity} are the bandwidth capacity constraints of the links, which requires that the total rate allocated to the flows on the link cannot exceed the link's capacity $C_{\ell}$. Let $p_{\ell}$ denote the Lagrange multiplier for the capacity constraint of link $\ell$.  {Note that the $p_{\ell}$'s are nonnegative.} Then the rate allocation $(x_r(t))_{r\in\mathcal{R}}$ satisfies
\begin{equation}\label{eq-allocation}
x_r(t)=
\begin{cases}
\frac{\weight_r}{\sum_{\ell:\ell\in r}p_{\ell}} & \text{when }n_{r}>0,\\
0 & \text{otherwise}.
\end{cases}
\end{equation}
For simplicity, we will just write $x_r(t)=x_r$ in the remainder of this paper, keeping in mind that $x_r$ implicitly depends on the flow counts at time $t$.

\paragraph{Arrivals and service.}
Flows arrive at route $r$ as a Poisson process with rate $\lambda_r$, and the arrival processes for different routes are independent. We assume that the file sizes are i.i.d.\ for flows on the same route and are independent across different routes. 
We remark that if the flows on a route have multiple classes, we can view each class as a different route and then our model still applies. For ease of exposition, we assume for now that the file size distribution of route $r$ is an exponential distribution with rate $\mu_r$. We will discuss more general file size distributions in Section~\ref{sec-phase-type}. We define the \emph{load on route $r$} to be $\rho_r=\lambda_r/\servrate_r$. Let $\bm{\lambda}=
(\lambda_r)_{r\in\mathcal{R}}$, $\bm{\mu}=
(\mu_r)_{r\in\mathcal{R}}$ and $\bm{\rho}=(\rho_r)_{r\in\mathcal{R}}$ denote the arrival rate vector, service rate vector and load vector, respectively.

\paragraph{Flow count process.}
Let $\bm{N}(t)=[N_1(t),N_2(t),\dots,N_R(t)]^T$ be the flow count vector.  The flow count process $(\bm{N}(t)\colon t\ge 0)$ is a continuous-time Markov chain. The transition rate $q_{\bm{n}\bm{n}'}$ from a state $\bm{n}$ to a state $\bm{n}'\neq\bm{n}$ is as follows:
\begin{equation}\label{eq-transition-rates-exp}
q_{\bm{n}\bm{n}'}=
\begin{cases}
\lambda_r &\text{if }\bm{n}'=\bm{n}+\bm{e}^{(r)},\\
n_{r}x_r\servrate_r&\text{if }\bm{n}'=\bm{n}-\bm{e}^{(r)},n_{r}>0,\\
0&\text{otherwise},
\end{cases}
\end{equation}
where $\bm{e}^{(r)}$ is the $r$th standard basis vector in $\mathbb{R}^{R}$, i.e., its $r$th entry is $1$ and other entries are $0$.
{We say that the system is stable if the flow count process is irreducible and positive recurrent.}

\paragraph{Heavy-traffic regime.}
We are interested in the stationary distribution of the flow count process in the heavy-traffic regime. Specifically, we consider a sequence of systems with the arrival rate vectors approaching the boundary of the capacity region.  {Here the capacity region refers to the interior of the set that consists of all the arrival rate vectors such that the system is stable under some policy}. Let the systems be indexed by a nonnegative parameter $\epsilon$, which represents how far away the arrival rate vector is from the boundary of the system capacity. Smaller $\epsilon$ means closer to the boundary and $\epsilon=0$ means on the boundary. The precise definition of $\epsilon$ will be given later. For clarity, we append the superscript $^{(\epsilon)}$ to the quantities that depend on $\epsilon$ in the $\epsilon$-th system. We say a quantity is a \emph{constant} if it does not depend on either the system state or $\epsilon$. Our goal is to analyze the flow counts in steady state for each system and then look at how they scale in the heavy-traffic regime where $\epsilon\to 0^+$.

Now we specify how the systems are parameterized by $\epsilon$. Consider an arrival rate vector $\bm{\lambda}^{(0)}$ on the boundary of the capacity region. We assume that in the $\epsilon$-th system, the arrival rate vector is given by $\bm{\lambda}^{(\epsilon)}=(1-\epsilon)\bm{\lambda}^{(0)}$. Let $\satlinks\subseteq \mathcal{L}$ denote the set of critically loaded links under the arrival rate vector $\bm{\lambda}^{(0)}$, i.e.,
\begin{align}
\sum_{r:\ell\in r}\rho^{(0)}_r&=C_\ell,\quad\text{for }\ell\in\satlinks,\\
\sum_{r:\ell'\in r}\rho^{(0)}_r&<C_{\ell'},\quad\text{for }\ell'\in\mathcal{L}\setminus\satlinks.
\end{align}
Let $\delta_\ell=C_\ell-\sum_{r:\ell\in r}\rho^{(0)}_r$. Then $\delta_\ell=0$ for $\ell\in\satlinks$ and $\delta_{\ell'}>0$ for $\ell'\in\mathcal{L}\setminus\satlinks$. Let $L_s=|\satlinks|$. We also assume that $\lambda^{(0)}_r>0$ for any $r\in\mathcal{R}$.

\paragraph{Inner product.}
For the space $\mathbb{R}^{R}$ where the states of the flow count process lie in, we consider the following weighted inner product:
\begin{equation}\label{eq-inner-exp}
\langle \bm{y},\bm{z}\rangle =\sum_{r} \frac{\weight_r}{\lambda^{(0)}_r}y_rz_r, \quad \bm{y},\bm{z}\in\mathbb{R}^R,
\end{equation}
and its induced norm. For conciseness we just use $\langle \cdot,\cdot\rangle$ and $\|\cdot\|$ to denote this weighted inner product and the induced norm in $\mathbb{R}^R$.
When we study more general service time distributions in Section~\ref{sec-phase-type}, we will introduce a generalization of this inner product and we will see that the choice of inner product is crucial.

\section{Main results for exponential file size distributions.}\label{sec-main}
In this section we present two main results for exponential file sizes distributions: state-space collapse and bounds on a weighted sum of the flow counts. We present this simpler setting with exponential file size distributions first since it better reveals the insights of our approach.  We will generalize the results to the more general setting with phase-type file size distributions
in Section~\ref{sec-phase-type}, where we construct a very complicated inner product.

\subsection{State-space collapse.}
Our first main result is state-space collapse, which intuitively means that the steady-state flow count vector concentrates around a lower-dimensional subspace of the state space in heavy traffic.

Specifically, the states of the flow count process $(\bm{N}(t):t\ge 0)$ lie in the $R$-dimensional space $\mathbb{R}^{R}$. We introduce an $L_s$-dimensional cone $\mathcal{K}$ in $\mathbb{R}^{R}$, which is where the state space collapses to in heavy traffic. Note that $L_s\le R$ by the full-rank assumption. 
The cone $\mathcal{K}$ is finitely generated by a set of vectors $\{\bm{b}^{(\ell)},\ell\in\satlinks\}\subseteq\{\bm{b}^{(\ell)},\ell\in\mathcal{L}\}$, i.e.,
\begin{equation}\label{eq-cone-exp}
\mathcal{K}=\biggl\{\bm{y}\in\mathbb{R}^{R}\colon \bm{y}=\sum_{\ell\in\satlinks} \alpha_{\ell}\bm{b}^{(\ell)},\alpha_{\ell}\ge 0\text{ for all }\ell\in\satlinks\biggr\}.
\end{equation}
Here each $\bm{b}^{(\ell)}=(b^{(\ell)}_r)_{r\in\mathcal{R}}$ with $\ell\in\mathcal{L}$ is defined as
\begin{equation}
b^{(\ell)}_{r}=\frac{\rho_{r}^{(0)}\mathbbm{1}_{\{\ell\in r\}}}{\weight_r},
\end{equation}
where recall that $\mathbbm{1}_{\{\ell\in r\}}$ is equal to $1$ when route $r$ uses link $\ell$ and equal to $0$ otherwise, and recall that $\weight_r$'s are the weights in the weighted proportionally fair policy.

The intuition for the steady-state flow count vector to concentrate around the cone $\mathcal{K}$ is as follows. Recall that the rate allocation $\bm{x}=(x_r)_{r\in\mathcal{R}}$ under the weighted proportionally fair sharing satisfies
\begin{equation*}
x_r=
\begin{cases}
\frac{\weight_r}{\sum_{\ell:\ell\in r}p_{\ell}} & \text{when }n_r>0,\\
0 & \text{otherwise},
\end{cases}
\end{equation*}
where $p_{\ell}$ is the Lagrange multiplier of the capacity constraint of link $\ell$. Then the flow count $n_r$ can be written as $n_r=\frac{n_rx_r}{\weight_r}\sum_{\ell:\ell\in r} p_{\ell}$.
Since the system is stable, on average, the arrival rate of data in files for each route is equal to the departure rate.  In the heavy-traffic regime, this means $n_rx_r\approx \rho_r^{(0)}$, which gives $n_r\approx\frac{\rho_r^{(0)}}{\weight_r}\sum_{\ell:\ell\in r} p_{\ell}$. Writing it in a vector form, we have
\begin{equation*}
\bm{n}\approx\sum_{\ell}p_{\ell}\bm{b}^{(\ell)}.
\end{equation*}
We further note that the links in $\mathcal{L}\setminus\satlinks$ are not critically loaded in the heavy-traffic regime.  So intuitively, the capacity constraints for those links are non-binding, and thus $p_{\ell}\approx 0$ for $\ell\in\mathcal{L}\setminus\satlinks$ by complementary slackness. Therefore,
\begin{equation}\label{eq-appr-ssc}
\bm{n}\approx\sum_{\ell:\ell\in\satlinks}p_{\ell}\bm{b}^{(\ell)}.
\end{equation}
If we view the $p_{\ell}$'s as constant coefficients, then \eqref{eq-appr-ssc} indicates that the flow count vector $\bm{n}$ roughly lies within the cone $\mathcal{K}$.

We now precisely define our notion of state-space collapse.  For any state~$\bm{n}$, we consider the following decomposition:
\begin{equation}
\bm{n}=\npara+\nperp,
\end{equation}
where $\npara$ is the \emph{projection} of $\bm{n}$ onto the cone $\mathcal{K}$, referred to as the \emph{parallel component}, and the remainder $\nperp$ is referred to as the \emph{perpendicular component} since it is perpendicular to $\npara$. Then $\npara\in\mathcal{K}$ and $\|\nperp\|$ is the distance between $\bm{n}$ and $\mathcal{K}$. Note that here the norm and projection are under the weighted inner product defined in \eqref{eq-inner-exp}.

Let $\steady^{(\epsilon)}$ denote a random vector whose distribution is the stationary distribution of the flow count process $(\bm{N}^{(\epsilon)}(t)\colon t\ge 0)$. We consider its parallel and perpendicular components, $\steady^{(\epsilon)}_{\shortparallel}$ and $\steadyperp^{(\epsilon)}$.  The state-space collapse indicates that as the arrival rate vector approaches the boundary of the capacity region, the perpendicular component $\steadyperp^{(\epsilon)}$ becomes negligible compared to the parallel component $\steady^{(\epsilon)}_{\shortparallel}$. We formally state this result in terms of moments in the following theorem.

\begin{theorem}[State-Space Collapse]\label{THM-SSC-EXP}
Consider a sequence of bandwidth sharing networks under a weighted proportionally fair policy, indexed by a parameter $\epsilon$ with $0<\epsilon<1$.  The file sizes have exponential distributions.  The arrival rate vector in the $\epsilon$-th system satisfies that $\bm{\lambda}^{(\epsilon)}=(1-\epsilon)\bm{\lambda}^{(0)}$ for some $\bm{\lambda}^{(0)}$ such that a set of $L_s$ links are critically loaded. 
Let $\steady^{(\epsilon)}$ denote a random vector whose distribution is the stationary distribution of the flow count process $(\bm{N}^{(\epsilon)}(t)\colon t\ge 0)$, and $\|\steadyperp^{(\epsilon)}\|$ denote its distance to the $L_s$-dimensional cone $\mathcal{K}$ defined in \eqref{eq-cone-exp} under the weighted inner product. Then in the heavy-traffic regime where $\epsilon\to 0^+$, the $m$-th moment of $\|\steadyperp^{(\epsilon)}\|$ for any nonnegative integer $m$ can be bounded as follows:
\begin{equation*}
\expect\Bigl[\Bigl\|\steadyperp^{(\epsilon)}\Bigr\|^m\Bigr]=O\Biggl(\biggl(\frac{1}{\sqrt\epsilon}\biggr)^m\Biggr).
\end{equation*}
\end{theorem}

We remark that $\expect[\|\steadyperp^{(\epsilon)}\|]/\expect[\|\steady^{(\epsilon)}\|]\to 0$ as $\epsilon\to 0^+$ since it can be proved that $\expect[\|\steady^{(\epsilon)}\|]=\Theta(1/\epsilon)$. Therefore, our state-space collapse is of a \emph{multiplicative} type.

\subsection{Bounds on flow counts.}
Based on the state-space collapse result, we establish the following bounds on the weighted sum of the flow counts.
\begin{theorem}[Bounds on Flow Counts]\label{THM-BOUNDS-EXP}
Consider a sequence of bandwidth sharing networks indexed by a parameter $\epsilon$ with $0<\epsilon<1$.  The file sizes have exponential distributions.  The arrival rate vector in the $\epsilon$-th system satisfies that $\bm{\lambda}^{(\epsilon)}=(1-\epsilon)\bm{\lambda}^{(0)}$ for some $\bm{\lambda}^{(0)}$ such that a set of $L_s$ links are critically loaded. Suppose that a weighted proportionally fair policy with weights $\weight_1,\dots,\weight_R$ is used.  Let $\steady^{(\epsilon)}$ denote a random vector whose distribution is the stationary distribution of the flow count process $(\bm{N}^{(\epsilon)}(t)\colon t\ge 0)$. Then in the heavy-traffic regime where $\epsilon\to 0^+$,
\begin{equation*}
\frac{L_s\cdot \min_r \weight_r}{\epsilon}+o\biggl(\frac{1}{\epsilon}\biggr)\le \expect\left[\sum_r \weight_r \overline{N}_r^{(\epsilon)}\right]
\le\frac{L_s\cdot \max_r \weight_r}{\epsilon}+o\biggl(\frac{1}{\epsilon}\biggr).
\end{equation*}
\end{theorem}

\section{Proof outlines of Theorems~\ref{THM-SSC-EXP} and \ref{THM-BOUNDS-EXP}---Steps in the drift method.}\label{sec-proof-outlines}
In this section we present proof outlines of the main results for exponential file size distributions in Theorems~\ref{THM-SSC-EXP} and \ref{THM-BOUNDS-EXP}. We intend to use the proof outlines as a tutorial for illustrating the steps in the drift method. The detailed proofs are given in Section~\ref{sec-proofs-main-exp}.

Consider the flow count process $(\bm{N}^{(\epsilon)}(t)\colon t\ge 0)$.  For any Lyapunov function $V\colon \mathbb{Z}_+^R\rightarrow \mathbb{R}_+$, the drift of $V(\cdot)$ at a state $\bm{n}$ is defined as
\begin{equation}
\Delta^{(\epsilon)} V(\bm{n})=\sum_{\bm{n}':\bm{n}'\neq\bm{n}}q^{(\epsilon)}_{\bm{n}\bm{n}'}(V(\bm{n}')-V(\bm{n})).
\end{equation}
where $q^{(\epsilon)}_{\bm{n}\bm{n}'}$ is the transition rate from state $\bm{n}$ to $\bm{n}'$ of the flow count process $(\bm{N}^{(\epsilon)}(t)\colon t\ge 0)$, given in \eqref{eq-transition-rates-exp}. The drift method studies the steady state of the flow count process in heavy traffic in the following two steps. 
\begin{enumerate}[leftmargin=4.1em]
\item[\textsc{Step 1.}] Establish state-space collapse: bound the moments of $\|\steadyperp^{(\epsilon)}\|$ by analyzing the drift $\Delta^{(\epsilon)}\|\nperp\|$.
\item[\textsc{Step 2.}] Bound the weighted sum of the flow counts in steady state by setting 
$\expect\left[\Delta^{(\epsilon)}\|(\steady^{(\epsilon)})_{\shortparallel}^s\|^2\right]=0$, where $(\cdot)_{\shortparallel}^s$ denotes the projection onto the \emph{subspace} where the cone $\mathcal{K}$ lies in.
\end{enumerate}

Before elaborating these two steps, below we identify two properties of the inner product we define in \eqref{eq-inner-exp}.  These two properties are easy to verify so the corresponding proofs are omitted.  However, they provide a basis for understanding the role of the inner product in the proofs of the main results.  Later we will develop generalizations of these two properties when we study more general file size distributions in Section~\ref{sec-phase-type}.  These two properties are concerned with the difference between the loads on the routes and the instantaneous rate allocation.  We introduce a set of vectors $\{\widehat{\bm{b}}^{(\ell)},\ell\in\mathcal{L}\}$, where each $\widehat{\bm{b}}^{(\ell)}=(\widehat{b}^{(\ell)}_r)_{r\in\mathcal{R}}$ is defined based on the rate allocation as follows:
\begin{equation*}
\widehat{b}^{(\ell)}_r=\frac{n_rx_r\mathbbm{1}_{\{\ell\in r\}}}{\weight_r}.
\end{equation*}
Recall that the vector $\bm{b}^{(\ell)}$ is defined as $b_r^{(\ell)}=\rho_r^{(0)}\mathbbm{1}_{\{\ell\in r\}}/\weight_r$. So the vector $\widehat{\bm{b}}^{(\ell)}$ replaces the $\rho_r^{(0)}$'s in $\bm{b}^{(\ell)}$ with $n_rx_r$'s.  We claim that the inner product satisfies the following two properties:
\begin{enumerate}[label=(P\arabic*),leftmargin=3.5em]
\item \label{property-P1} For each link $\ell$,
\begin{equation*}
\langle \bm{b}^{(\ell)},\bm{\lambda}^{(0)}-\bm{nx\mu}\rangle = U_{\ell}-\delta_{\ell},
\end{equation*}
where $\bm{nx\mu}$ denotes the entrywise product of $\bm{n}$, $\bm{x}$ and $\bm{\mu}$, i.e., $\bm{nx\mu}=(n_rx_r\mu_r)_{r\in\mathcal{R}}$, and $U_{\ell}$ is the \emph{unused bandwidth} on link $\ell$, i.e., the amount of bandwidth that is not allocated to any flow.
\item \label{property-P2} For each link $\ell$,
\begin{equation*}
\langle\bm{b}^{(\ell)}-\widehat{\bm{b}}^{(\ell)},\bm{\lambda}^{(0)}-\bm{nx\mu}\rangle\ge\weight_{\min}\mu_{\min}\|\bm{b}^{(\ell)}-\widehat{\bm{b}}^{(\ell)}\|^2,
\end{equation*}
where $\weight_{\min}=\min_r\{\weight_r\}>0$, $\mu_{\min}=\min_r\{\mu_r\}>0$.
\end{enumerate}

\subsection{Proof outline of Theorem~\ref{THM-SSC-EXP} (State-Space Collapse).}\label{subsec-ssc-exp}
In this section we elaborate Step~1 of the drift method and give the proof outline of the state-space collapse result in Theorem~\ref{THM-SSC-EXP}. The detailed proof is given in Section~\ref{subsec-proof-thm-ssc-exp}.

In the prior work of the drift method, moment bounds on $\|\steadyperp^{(\epsilon)}\|$ are usually obtained by studying the drift of the Lyapunov function $\|\nperp\|$ and applying the moment bounds in \cite{Haj_82} or the more refined tail bounds in \cite{BerGamTsi_01}. 
\iftoggle{complete}{We include a continuous-time version of the results in \cite{BerGamTsi_01} and its proof in Appendix~\ref{app-tail} for easy reference.
}{%
}%
However, for the flow count process $(\bm{N}^{(\epsilon)}(t)\colon t\ge 0)$ under a weighted proportionally fair policy, it is hard if not impossible to obtain a drift bound for the Lyapunov function $\|\nperp\|$ such that the drift conditions in \cite{BerGamTsi_01} are satisfied. Specifically, one condition in \cite{BerGamTsi_01} requires the drift to be negative whenever the value of the Lyapunov function is large enough, but for $(\bm{N}^{(\epsilon)}(t)\colon t\ge 0)$, large $\|\nperp\|$ alone may not be enough to give a negative drift.

Although the results in \cite{BerGamTsi_01} are not directly applicable, we will show that we can still obtain moment bounds by studying the drift. We first prove two drift bounds in Lemma~\ref{lem-nperp-drift-exp}, which mainly state that the drift is negative under the \emph{additional} condition that the ratio $\|\nperp\|/\|\bm{n}\|$ is also large enough. We then show in Lemma~\ref{lem-distr-bound-exp} that similar to the results in \cite{BerGamTsi_01}, the drift bounds lead to certain tail bounds and then further the moment bounds in Theorem~\ref{THM-SSC-EXP}.
\begin{lemma}[Drift Bounds for $\|\nperp\|$]\label{lem-nperp-drift-exp}
In the $\epsilon$-th system, the drift of the Lyapunov function $\|\nperp\|$ satisfies that
\begin{equation}
\Delta^{(\epsilon)} \|\nperp\|\le -\sqrt{\epsilon}
\end{equation}
when
\begin{equation}\label{eq-conditions-drift-exp}
\epsilon\le \epsilon_{\max},\quad\|\nperp\|\ge\frac{A_1}{\xi_1\sqrt{\epsilon}},\quad\frac{\|\nperp\|}{\sum_r\weight_r n_r}\ge\frac{\xi_2\sqrt{\epsilon}}{A_2},
\end{equation}
and
\begin{equation}
\Delta^{(\epsilon)} \|\nperp\|\le(\xi_1+1)\sqrt{\epsilon}
\end{equation}
when
\begin{equation}
\epsilon\le \epsilon_{\max},\quad\|\nperp\|\ge\frac{A_1}{\xi_1\sqrt{\epsilon}},
\end{equation}
where $\epsilon_{\max},\xi_1,\xi_2,A_1,A_2$ are positive constants.
\end{lemma}
The proof of Lemma~\ref{lem-nperp-drift-exp} is given in Section~\ref{subsec-proof-lem-nperp-drift-exp}.
We remark that the last condition on $\|\nperp\|/\sum_r\weight_r n_r$ in \eqref{eq-conditions-drift-exp} is equivalent to that $\|\nperp\|/\|\bm{n}\|$ is large enough, since all norms are equivalent in $\mathbb{R}^{\Phase}$ and thus there exist positive constants $a_1$ and $a_2$ such that $a_1\|\bm{n}\|\le \sum_r\weight_r n_r\le a_2\|\bm{n}\|$.

\begin{lemma}[Tail and Moment Bounds for $\|\steadyperp^{(\epsilon)}\|$]\label{lem-distr-bound-exp}
For any nonnegative $\epsilon\le \epsilon_{\max}$, the tail distribution of $\|\steadyperp^{(\epsilon)}\|$ is bounded by an exponential term plus an additional term as follows: for any nonnegative integer $j$,
\begin{equation}\label{eq-nperp-tail-bounds-exp}
\Pr\biggl(\|\steadyperp^{(\epsilon)}\|>\frac{A_1}{\xi_1\sqrt{\epsilon}}+2\nu j\biggr)
\le\alpha^{j+1}+\xi_2(1-\alpha)\sum_{i=0}^j\alpha^i\Bigl(\beta^{\theta/\sqrt{\epsilon}}\Bigr)^{j-i},
\end{equation}
where $\epsilon_{\max},A_1,\xi_1,\xi_2$ are the constants in Lemma~\ref{lem-nperp-drift-exp}, $\nu $ and $\theta$ are positive constants, and
\begin{equation}\label{eq-alpha-beta-exp}
\alpha=\frac{a}{a+\sqrt{\epsilon}},\quad \beta=\frac{b}{b+\epsilon},
\end{equation}
for some positive constants $a$ and $b$. As a result, in the heavy-traffic regime where $\epsilon\to 0^+$, the $m$-th moment of $\|\steadyperp^{(\epsilon)}\|$ for any $m\in\mathbb{Z}_+$ can be bounded as follows:
\begin{equation}\label{eq-nperp-moment-bounds-exp}
\expect\Bigl[\Bigl\|\steadyperp^{(\epsilon)}\Bigr\|^m\Bigr]=O\Biggl(\biggl(\frac{1}{\sqrt\epsilon}\biggr)^m\Biggr).
\end{equation}
\end{lemma}
The proof of Lemma~\ref{lem-distr-bound-exp} is given in Section~\ref{subsec-proof-lem-distr-bound-exp}.
Compared with the exponential-type tail bounds in \cite{BerGamTsi_01}, the tail bounds in Lemma~\ref{lem-distr-bound-exp} have an additional term (the second term) due to the additional requirement on $\|\nperp\|/\sum_r\weight_rn_r$ in the drift bound in Lemma~\ref{lem-nperp-drift-exp}. In the proof of Lemma~\ref{lem-distr-bound-exp}, this additional requirement will be addressed by properly bounding $\sum_r\weight_r \overline{N}_r$, also through studying a Lyapunov drift. 
The derivation of the moment bounds based on the tail bounds is intuitive and similar to Lemma~3 in \cite{MagSri_16}.

\subsection{Proof outline of Theorem~\ref{THM-BOUNDS-EXP} (Bounds on Flow Counts).}\label{subsec-bounds-exp}
In this section we elaborate Step~2 of the drift method and give proof outlines of the bounds on flow counts in Theorem~\ref{THM-BOUNDS-EXP}. The detailed proof is given in Section~\ref{subsec-proof-thm-bounds-exp}. We obtain these bounds by setting the steady-state drift of the Lyapunov function $V(\bm{n})=\|\nparaspace\|^2$ to $0$, where $\nparaspace$ is the projection of the state $\bm{n}$ onto the \emph{subspace} where the cone $\mathcal{K}$ lies in, i.e., the subspace spanned by $\bm{b}^{(\ell)}$'s, denoted by $\mathcal{S}$. Note that in this section, we often consider the projection onto the \emph{subspace} instead of the projection onto the \emph{cone}. We use the superscript $^s$ to indicate when the projection is onto the subspace. Then $\nparaspace$ can be written as
\begin{equation}
\nparaspace=\sum_{\ell}\alpha_{\ell}^s\bm{b}^{(\ell)},
\end{equation}
where the coefficients $\alpha_{\ell}^s$'s can be negative. The projection onto the subspace is a linear operator, i.e., $(\bm{y}+\bm{z})_{\shortparallel}^s=\bm{y}_{\shortparallel}^s+\bm{z}_{\shortparallel}^s$ for any $\bm{y},\bm{z}\in\mathbb{R}^{\Phase}$. Note that $\|\nperpspace\|\le\|\nperp\|$ since the cone $\mathcal{K}$ is a subset of the subspace.

Below we fix an $\epsilon>0$ and temporarily omit the superscript $^{(\epsilon)}$ for conciseness. The drift of $\|\nparaspace\|^2$ can be written as follows:
\begin{align}
\Delta \|\nparaspace\|^2&= \sum_r\lambda_r\Bigl(\|(\bm{n}+\bm{e}^{(r)})_{\shortparallel}^s\|^2-\|\nparaspace\|^2\Bigr)+\sum_rn_rx_r\mu_r\Bigl(\|(\bm{n}-\bm{e}^{(r)})_{\shortparallel}^s\|^2-\|\nparaspace\|^2\Bigr)\nonumber\\
&=2\sum_r(\lambda_r-n_rx_r\mu_r)\langle \nparaspace,(\bm{e}^{(r)})_{\shortparallel}^s\rangle+\sum_r(\lambda_r+n_rx_r\mu_r)\|(\bm{e}^{(r)})_{\shortparallel}^s\|^2\nonumber\\
&=2\sum_r(\lambda_r-n_rx_r\mu_r)\langle \nparaspace,\bm{e}^{(r)}\rangle+\sum_r(\lambda_r+n_rx_r\mu_r)\|(\bm{e}^{(r)})_{\shortparallel}^s\|^2\label{eq-perp}\\
&=2\langle \nparaspace,\bm{\lambda}-\bm{nx\mu}\rangle +\sum_r(\lambda_r+n_rx_r\mu_r)\|(\bm{e}^{(r)})_{\shortparallel}^s\|^2\label{eq-nxmu}\\
&=-2\epsilon\langle \nparaspace,\bm{\lambda}^{(0)}\rangle + 2\langle \nparaspace,\bm{\lambda}^{(0)}-\bm{nx\mu}\rangle +\sum_r(\lambda_r+n_rx_r\mu_r)\|(\bm{e}^{(r)})_{\shortparallel}^s\|^2\nonumber\\
&=-2\epsilon\langle\bm{n},\bm{\lambda}^{(0)}\rangle+2\epsilon\langle\nperpspace,\bm{\lambda}^{(0)}\rangle + 2\langle \nparaspace,\bm{\lambda}^{(0)}-\bm{nx\mu}\rangle+B_1(\bm{n}),\label{eq-const-term}
\end{align}
where \eqref{eq-perp} follows from $\langle \nparaspace,(\bm{e}^{(r)})_{\perp}^s\rangle=0$, the $\bm{nx\mu}$ in \eqref{eq-nxmu} denotes the entrywise product of $\bm{n}$, $\bm{x}$ and $\bm{\mu}$, and in \eqref{eq-const-term}, $B_1(\bm{n})\triangleq\sum_r(\lambda_r+n_rx_r\mu_r)\|(\bm{e}^{(r)})_{\shortparallel}^s\|^2$. When the system is in steady state, we have $\expect[\Delta \|\steadyparaspace\|^2]=0$. Note that this is true since we can show that $\|\steady\|$ has finite moments (Lemma~\ref{LEM-SUM-NR-EXP}). Also note that by the definition of the inner product we choose, $\langle \steady,\bm{\lambda}^{(0)}\rangle=\sum_r\weight_r\overline{N}_r$.
Then setting $\expect[\Delta \|\steadyparaspace\|^2]=0$ yields
\begin{equation}\label{eq-drift-equal-0}
\epsilon\expect\left[\sum_r \weight_r \overline{N}_r\right]=\epsilon\expect[\langle \steadyperp^s,\bm{\lambda}^{(0)}\rangle]+\expect[\langle\steadyparaspace,\bm{\lambda}^{(0)}-\steady\bm{x\mu}\rangle]+\frac{1}{2}\expect[B_1(\steady)].
\end{equation}
The proof of Theorem~\ref{THM-BOUNDS-EXP} analyzes the terms on the right-hand-side of \eqref{eq-drift-equal-0}. A sketch is given below:
\begin{enumerate}[leftmargin=0em,itemindent=3em,label=({\roman*})]
\item We show $\epsilon\expect[\langle \steadyperp^s,\bm{\lambda}^{(0)}\rangle]=O(\sqrt{\epsilon})$ using state-space collapse.
\item Consider the term $\expect[\langle\steadyparaspace,\bm{\lambda}^{(0)}-\steady\bm{x\mu}\rangle]$. Since $\steadyparaspace$ is in the subspace $\mathcal{S}$, it can be written as $\steadyparaspace=\sum_{\ell\in\satlinks}\alpha_{\ell}^s\bm{b}^{(\ell)}$.
Then as before,
\begin{align*}
\langle \steadyparaspace,\bm{\lambda}^{(0)}-\steady\bm{x\mu}\rangle=\sum_{\ell\in\satlinks}\alpha_{\ell}^s\langle \bm{b}^{(\ell)},\bm{\lambda}^{(0)}-\steady\bm{x\mu}\rangle=\sum_{\ell\in\satlinks}\alpha_{\ell}^sU_{\ell}.
\end{align*}
Recall that due to complementary slackness, $p_{\ell}U_{\ell}=0$. So we can bound $\expect[\langle\steadyparaspace,\bm{\lambda}^{(0)}-\steady\bm{x\mu}\rangle]$ by showing that $\alpha_{\ell}$ and $p_{\ell}$ are close to each other in heavy-traffic for any $\ell\in\satlinks$. We know that for a state $\bm{n}$ such that $\bm{n}=\nparaspace$, i.e., $\|\nperp\|=\|\nperpspace\|=0$, the Lagrange multipliers $p_{\ell}$'s are equal to the coefficients $\alpha_{\ell}^s$'s of the projection. Then intuitively, when $\|\nperp\|$ is small, the rate allocation based on $\bm{n}$ should be not far away from the rate allocation based on $\nparaspace$, and thus the $p_{\ell}$'s should not be far away from the $\alpha_{\ell}$'s. Then we can use the state-space collapse result to bound the difference $|\alpha_{\ell}^s-p_{\ell}|$ in heavy traffic. Specifically, the following lemma bounds the difference $|\alpha_{\ell}^s-p_{\ell}|$ using $\|\nperp\|$, where notice that $\nperp$ is the projection onto the cone.

\begin{lemma}\label{lem-continuity2-exp}
There exist positive constants $B_2$ and $B_3$ such that for any state $\bm{n}$,
\begin{equation}\label{eq-diff-alpha-p-exp}
|\alpha_{\ell}^s-p_{\ell}|\le B_2\|\nperp\|^{1/2}\Biggl(\sum_r\weight_r n_r\Biggr)^{1/2},\quad \forall \ell \in\satlinks,
\end{equation}
and
\begin{equation}\label{eq-bound-p-exp}
\sum_{\ell\in\mathcal{L}\setminus\satlinks}p_{\ell}\le B_3\|\nperp\|,
\end{equation}
where the $\alpha_{\ell}^s$'s are the coefficients in the projection $\nparaspace=\sum_{\ell\in\satlinks}\alpha_{\ell}\bm{b}^{(\ell)}$ and the $p_{\ell}$'s are the Lagrange multipliers for the capacity constraints.
\end{lemma}
Based on Lemma~\ref{lem-continuity2-exp}, we show that $\expect[\langle\steadyparaspace,\bm{\lambda}^{(0)}-\steady\bm{x\mu}\rangle]=O(\epsilon^{\frac{1}{4}-\frac{1}{\tau_1}})$ for an even integer $\tau_1$ with $\tau_1>4$. We remark that the proof of Lemma~\ref{lem-continuity2-exp} is where the full-rank assumption is needed.

\item To bound the last term $\expect[B_1(\steady)]/2$, we first notice that $\expect[\overline{N}_rx_r]=\rho_r$ by the fact that $\expect[\Delta \overline{N}_r]=0$. Then $\frac{1}{2}\expect[B_1(\steady)]=\sum_r \lambda_r\|(\bm{e}^{(r)})_{\shortparallel}^s\|^2$.
We show that it can be bounded as follows,
\begin{equation*}
(1-\epsilon)L_s\cdot\min_r\weight_r\le \frac{1}{2}\expect[B_1(\steady)]\le (1-\epsilon)L_s\cdot\max_r\weight_r.
\end{equation*}
\end{enumerate}

Combining (i), (ii) and (iii) will yield the bounds in Theorem~\ref{THM-BOUNDS-EXP}.

\section{Proofs of Theorems~\ref{THM-SSC-EXP} and \ref{THM-BOUNDS-EXP}.}\label{sec-proofs-main-exp}

\subsection{Proof of Theorem~\ref{THM-SSC-EXP} (State-Space Collapse).}\label{subsec-proof-thm-ssc-exp}
By the arguments in Section~\ref{subsec-ssc-exp}, it suffices to prove Lemmas~\ref{lem-nperp-drift-exp} and \ref{lem-distr-bound-exp}.
\subsubsection{Proof of Lemma~\ref{lem-nperp-drift-exp} (Drift Bounds for \texorpdfstring{$\|\nperp\|$}{perp}).}\label{subsec-proof-lem-nperp-drift-exp}
In this proof we fix an $\epsilon>0$ and temporarily omit the superscript~$^{(\epsilon)}$ for conciseness.
Recall that $\bm{n}=[n_1,n_2,\dots,n_R]^T$ is the flow count vector, $\bm{x}=[x_1,x_2,\dots,x_R]^T$ is the rate allocation vector under a weighted proportionally fair policy, and $\bm{nx\mu}$ is the entrywise product of $\bm{n}$, $\bm{x}$ and $\bm{\mu}$.

We start by stating the following claim, the proof of which is given at the end of this proof.
\begin{claim}\label{claim-nperp-drift-exp}
The drift, $\Delta \|\nperp\|$, is upper bounded as follows:
\begin{equation*}
\Delta \|\nperp\|\le \frac{1}{\|\nperp\|}\langle \bm{n}-\npara,\bm{\lambda}^{(0)}-\bm{nx\mu}\rangle+\epsilon\|\bm{\lambda}^{(0)}\|+\frac{A_1}{\|\nperp\|},
\end{equation*}
where $A_1$ is a constant.
\end{claim}

Next we analyze the terms in Claim~\ref{claim-nperp-drift-exp}, utilizing the two properties \ref{property-P1} and \ref{property-P2} of the inner product.
We first consider the term $\langle \bm{n},\bm{\lambda}^{(0)}-\bm{nx\mu}\rangle$ in Claim~\ref{claim-nperp-drift-exp}. Recall that $n_r$ can be written in the following form according to the proportionally fair policy: $n_r=\frac{n_rx_r}{\weight_r}\sum_{\ell:\ell\in r} p_{\ell}$,
where $\weight_r$ is the weight used and $p_{\ell}$ is the Lagrange multiplier of the capacity constraint of link $\ell$. We can further write this equality in a vector form using the vectors $\{\widehat{\bm{b}}^{(\ell)},\ell\in\mathcal{L}\}$:
\begin{equation*}
\bm{n}=\sum_{\ell} p_{\ell}\widehat{\bm{b}}^{(\ell)}.
\end{equation*}
By property \ref{property-P1} and complementary slackness, $p_{\ell}\langle \bm{b}^{(\ell)},\bm{\lambda}^{(0)}-\bm{nx\mu}\rangle = p_{\ell}U_{\ell}-p_{\ell}\delta_{\ell}=-p_{\ell}\delta_{\ell}$, where recall that $\delta_{\ell}=0$ for $\ell\in\satlinks$.
Thus
\begin{align*}
\langle \bm{n},\bm{\lambda}^{(0)}-\bm{nx\mu}\rangle&=\sum_{\ell} p_{\ell} \langle \widehat{\bm{b}}^{(\ell)}-\bm{b}^{(\ell)},\bm{\lambda}^{(0)}-\bm{nx\mu}\rangle+\sum_{\ell\in\mathcal{L}\setminus\satlinks}p_{\ell}\langle\bm{b}^{(\ell)},\bm{\lambda}^{(0)}-\bm{nx\mu}\rangle\\
&\le -\weight_{\min}\mu_{\min}\sum_{\ell}p_{\ell}\|\widehat{\bm{b}}^{(\ell)}-\bm{b}^{(\ell)}\|^2-\sum_{\ell\in\mathcal{L}\setminus\satlinks}p_{\ell}\delta_{\ell},
\end{align*}
where the inequality follows from \ref{property-P2}. Note that
\begin{align}
\|\nperp\|^2&\le \biggl\|\bm{n}-\sum_{\ell\in\satlinks}p_{\ell}\bm{b}^{(\ell)}\biggr\|^2\label{eq-perp-bound-1-exp}\\
&=\biggl\|\sum_{\ell} p_{\ell}\Bigl(\widehat{\bm{b}}^{(\ell)}-\bm{b}^{(\ell)}\Bigr)+\sum_{\ell\in\mathcal{L}\setminus\satlinks}p_{\ell}\bm{b}^{(\ell)}\biggr\|^2\nonumber\\
&\le\biggl(\sum_{\ell} p_{\ell}\|\widehat{\bm{b}}^{(\ell)}-\bm{b}^{(\ell)}\|+\sum_{\ell\in\mathcal{L}\setminus\satlinks}p_{\ell}\|\bm{b}^{(\ell)}\|\biggr)^2\nonumber\\
&\le\biggl(\sum_{\ell} p_{\ell}+\sum_{\ell\in\mathcal{L}\setminus\satlinks}p_{\ell}\biggr)\biggl(\sum_{\ell} p_{\ell}\|\widehat{\bm{b}}^{(\ell)}-\bm{b}^{(\ell)}\|^2+\sum_{\ell\in\mathcal{L}\setminus\satlinks}p_{\ell}\|\bm{b}^{(\ell)}\|^2\biggr)\label{eq-perp-bound-2-exp}\\
&\le\frac{2}{C_{\min}}\biggl(\sum_r\weight_rn_r\biggr)\biggl(\sum_{\ell} p_{\ell}\|\widehat{\bm{b}}^{(\ell)}-\bm{b}^{(\ell)}\|^2+\sum_{\ell\in\mathcal{L}\setminus\satlinks}p_{\ell}\|\bm{b}^{(\ell)}\|^2\biggr),\label{eq-perp-bound-3-exp}
\end{align}
where \eqref{eq-perp-bound-1-exp} follows from the definition of projection, \eqref{eq-perp-bound-2-exp} follows from Cauchy-Schwarz inequality, and \eqref{eq-perp-bound-3-exp} is due to the equality $\sum_{\ell} p_{\ell}C_{\ell}=\sum_r \weight_rn_r$ derived from the weighted proportionally fair sharing policy with $C_{\min}=\min_{\ell} C_{\ell}$ {with complementary slackness}. Let $A_2=\min\biggl\{\weight_{\min}\mu_{\min},\min_{\ell\in\mathcal{L}\setminus\satlinks}\frac{\delta_{\ell}}{\|\bm{b}^{(\ell)}\|^2}\biggr\}\frac{C_{\min}}{2}$.  Then $A_2$ is a positive constant independent of $\epsilon$ and
\begin{equation}\label{eq-bound-npart-exp}
\langle \bm{n},\bm{\lambda}^{(0)}-\bm{nx\mu}\rangle\le -A_2\frac{\|\nperp\|^2}{\sum_r\weight_rn_r}.
\end{equation}

We then consider the term $\langle\npara,\bm{\lambda}^{(0)}-\bm{nx\mu}\rangle$ in Claim~\ref{claim-nperp-drift-exp}. Since $\npara\in\mathcal{K}$, we can write it as
\begin{equation*}
\npara=\sum_{\ell\in\satlinks} \alpha_{\ell}\bm{b}^{(\ell)},\quad\text{for some $\alpha_{\ell}$'s where }\alpha_{\ell}\ge 0\text{ for each $\ell\in\satlinks$}.
\end{equation*}
Then
\begin{align}
\langle \npara,\bm{\lambda}^{(0)}-\bm{nx\mu}\rangle&=\sum_{\ell\in\satlinks} \alpha_{\ell}\langle \bm{b}^{(\ell)},\bm{\lambda}^{(0)}-\bm{nx\mu}\rangle\nonumber\\
&=\sum_{\ell\in\satlinks} \alpha_{\ell}U_{\ell}\label{eq-bound-nparapart-1-exp}\\
&\ge 0,\label{eq-bound-nparapart-exp}
\end{align}
where \eqref{eq-bound-nparapart-1-exp} follows from \ref{property-P1}.

Combining the bounds for the terms $\langle \bm{n},\bm{\lambda}^{(0)}-\bm{nx\mu}\rangle$ and $\langle\npara,\bm{\lambda}^{(0)}-\bm{nx\mu}\rangle$ in \eqref{eq-bound-npart-exp} and \eqref{eq-bound-nparapart-exp} yields:
\begin{align*}
\Delta \|\nperp\|&\le -A_2\frac{\|\nperp\|}{\sum_r\weight_r n_r} +\epsilon\|\bm{\lambda}^{(0)}\|+\frac{A_1}{\|\nperp\|}.
\end{align*}
We choose any constants $\xi_1>0,\xi_2>0$ such that $\xi_2-\xi_1=2$.
Then when
\begin{equation*}
\epsilon\le\epsilon_{\max}\triangleq \frac{1}{\|\bm{\lambda}^{(0)}\|^2},\quad\|\nperp\|\ge\frac{A_1}{\xi_1\sqrt{\epsilon}},\quad
\frac{\|\nperp\|}{\sum_r\weight_r n_r}\ge\frac{\xi_2\sqrt{\epsilon}}{A_2},
\end{equation*}
we have $\Delta \|\nperp\|\le -\xi_2\sqrt{\epsilon}+\sqrt{\epsilon}+\xi_1\sqrt{\epsilon}=-\sqrt{\epsilon}$, and when
\begin{equation*}
\epsilon\le\epsilon_{\max},\quad\|\nperp\|\ge\frac{A_1}{\xi_1\sqrt{\epsilon}},
\end{equation*}
we have $\Delta \|\nperp\|\le \sqrt{\epsilon}+\xi_1\sqrt{\epsilon}=(\xi_1+1)\sqrt{\epsilon}$, which are the drift bounds in Lemma~\ref{lem-nperp-drift-exp}.

Lastly, we prove the Claim~\ref{claim-nperp-drift-exp} at the beginning of this proof. We first bound $\Delta\|\nperp\|$ in the following form
\begin{align*}
\Delta \|\nperp\|&\le\frac{1}{2\|\nperp\|}\Delta \|\nperp\|^2=\frac{1}{2\|\nperp\|}(\Delta \|\bm{n}\|^2-\Delta\|\npara\|^2),
\end{align*}
where the inequality follows from the fact that $\|\nperp\|=\sqrt{\|\nperp\|^2}$ and that the square-root function is concave. Below we analyze the drifts $\Delta\|\bm{n}\|^2$ and $\Delta\|\npara\|^2$.
\begin{align}
\Delta \|\bm{n}\|^2 &= \sum_r\lambda_r\Bigl(\|\bm{n}+\bm{e}^{(r)}\|^2-\|\bm{n}\|^2\Bigr)\nonumber+\sum_rn_rx_r\mu_r\Bigl(\|\bm{n}-\bm{e}^{(r)}\|^2-\|\bm{n}\|^2\Bigr)\nonumber\\
&=2\sum_r(\lambda_r-n_rx_r\mu_r)\langle \bm{n},\bm{e}^{(r)}\rangle+\sum_r(\lambda_r+n_rx_r\mu_r)\|\bm{e}^{(r)}\|^2\nonumber\\
&\le 2\sum_r(\lambda_r-n_rx_r\mu_r)\langle \bm{n},\bm{e}^{(r)}\rangle+2A_1\nonumber\\
&=2\langle \bm{n},\bm{\lambda}-\bm{nx\mu}\rangle+2A_1,\label{eq-bound-n2-drift-exp}
\end{align}
where the first equality follows from the transition rates of the flow count process, and $A_1$ is a constant. We can derive a lower bound on $\Delta\|\npara\|^2$ in a similar way:
\begin{align}
\Delta \|\npara\|^2 &= \sum_r\lambda_r\Bigl(\|(\bm{n}+\bm{e}^{(r)})_{\shortparallel}\|^2-\|\npara\|^2\Bigr)+\sum_rn_rx_r\mu_r\Bigl(\|(\bm{n}-\bm{e}^{(r)})_{\shortparallel}\|^2-\|\npara\|^2\Bigr)\nonumber\\
&=\sum_r\lambda_r\cdot 2\langle \npara,(\bm{n}+\bm{e}^{(r)})_{\shortparallel}-\npara\rangle+\sum_rn_rx_r\mu_r\cdot 2\langle \npara,(\bm{n}-\bm{e}^{(r)})_{\shortparallel}-\npara\rangle\nonumber\\
&\mspace{23mu}+\sum_r\lambda_r\|(\bm{n}+\bm{e}^{(r)})_{\shortparallel}-\npara\|^2+\sum_rn_rx_r\mu_r\|(\bm{n}-\bm{e}^{(r)})_{\shortparallel}-\npara\|^2\nonumber\\
&\ge\sum_r\lambda_r\cdot 2\langle \npara,(\bm{n}+\bm{e}^{(r)})-\bm{n}\rangle+\sum_rn_rx_r\mu_r\cdot 2\langle \npara,(\bm{n}-\bm{e}^{(r)})-\bm{n}\rangle\label{eq-proj-lower}\\
&=2\langle \npara,\bm{\lambda}-\bm{nx\mu}\rangle,\nonumber
\end{align}
where \eqref{eq-proj-lower} follows from that $\langle \npara,\nperp\rangle = 0$ and $\langle \npara,(\bm{n}+\bm{e}^{(r)})_{\perp}\rangle\le 0$, $\langle \npara,(\bm{n}-\bm{e}^{(r)})_{\perp}\rangle\le 0$ since perpendicular components are in the polar cone of the cone $\mathcal{K}$. Combining the above bounds yields
\begin{align*}
\Delta\|\nperp\|
&\le\frac{1}{\|\nperp\|}\langle \bm{n}-\npara,\bm{\lambda}-\bm{nx\mu}\rangle+\frac{A_1}{\|\nperp\|}\\
&=\frac{1}{\|\nperp\|}\langle \bm{n}-\npara,\bm{\lambda}^{(0)}-\bm{nx\mu}\rangle-\epsilon\frac{\langle \nperp,\bm{\lambda}^{(0)}\rangle}{\|\nperp\|}+\frac{A_1}{\|\nperp\|}\\
&\le\frac{1}{\|\nperp\|}\langle \bm{n}-\npara,\bm{\lambda}^{(0)}-\bm{nx\mu}\rangle+\epsilon\|\bm{\lambda}^{(0)}\|+\frac{A_1}{\|\nperp\|},
\end{align*}
which completes the proof of the claim.
\Halmos

\subsubsection{Proof of Lemma~\ref{lem-distr-bound-exp} (Tail and Moment Bounds for \texorpdfstring{$\|\steadyperp\|$}{perp}).}\label{subsec-proof-lem-distr-bound-exp}
We first introduce some notation that will be needed in the proof. We again fix an $\epsilon>0$ and omit the superscript $^{(\epsilon)}$ for conciseness. Recall that $q_{\bm{n}\bm{n}'}$ is the transition rate from state $\bm{n}$ to $\bm{n}'$ of the flow count process $(\bm{N}(t)\colon t\ge 0)$. Let
\begin{gather*}
\overline{q}= \sup_{\bm{n}}(-q_{\bm{nn}}),\quad
\nu=\max_r\frac{\weight_r}{\lambda_r^{(0)}},\quad\zeta=2\left(\max_{\ell}C_{\ell}\right)\cdot\left(\max_r\mu_r\right),\quad
\alpha=\frac{\zeta \nu }{\zeta \nu +\sqrt{\epsilon}}.
\end{gather*}
{Here $\overline{q}$ is the maximum total transition rate out of a state.}
It can be verified that $\overline{q}<+\infty$, and that $\nu$ and $\zeta$ are positive constants such that
\begin{gather*}
\sup_{\bm{n},\bm{n}'\colon q_{\bm{n}\bm{n}'}>0} \bigl|\|\bm{n}'\|-\|\bm{n}\|\bigr|\le \nu,\quad \sup_{\bm{n},\bm{n}'\colon q_{\bm{n}\bm{n}'}>0} \bigl|\|\nperp'\|-\|\nperp\|\bigr|\le \nu,\\
\sup_{\bm{n}}\sum_{\bm{n}'\colon \|\bm{n}\|<\|\bm{n}'\|}q_{\bm{n}\bm{n}'}\le\zeta,\quad \sup_{\bm{n}}\mspace{-9mu}\sum_{\bm{n}'\colon \|\nperp\|<\|\nperp'\|}\mspace{-9mu}q_{\bm{n}\bm{n}'}\le \zeta.
\end{gather*}
Note that for $\alpha$, the constant $a$ in \eqref{eq-alpha-beta-exp} of Lemma~\ref{lem-distr-bound-exp} equals to $\zeta \nu $.

\begin{sloppypar}
We need the Lemma~\ref{LEM-SUM-NR-EXP} below, which gives tail bounds for $\sum_r\weight_r \overline{N}_r$.
{Note that these tail bounds are similar to the tail bounds in \cite{ShaTsiZho_14} for the maximum flow count among different routes.
}
\iftoggle{complete}{%
The proof of Lemma~\ref{LEM-SUM-NR-EXP} is given in Appendix~\ref{app-proof-lem-sum-nr-exp}, which analyzes the drift $\Delta\|\bm{n}\|$ and applies Lemma~\ref{LEM-TAIL} (a continuous-time version of Theorem~1 in \citet{BerGamTsi_01}).
}{%
Lemma~\ref{LEM-SUM-NR-EXP} can be proven using a continuous-time version of Theorem~1 in \citet{BerGamTsi_01}.  The proof is given in our technical report \cite{WanMagSri_20} due to space limitations.
}%
Note that the definition of $\beta$ below corresponds to $b=2\zeta \nu/A_3$ for the constant $b$ in \eqref{eq-alpha-beta-exp} of Lemma~\ref{lem-distr-bound-exp}.
\end{sloppypar}
\begin{lemma}\label{LEM-SUM-NR-EXP}
For any nonnegative $\epsilon$ with $\epsilon\le 1$, the distribution of $\sum_r \weight_r\overline{N}_r$ has the following exponential tail bounds: for any nonnegative integer $j$,
\begin{align*}
\Pr\biggl(\sum_r\weight_r\overline{N}_r>\frac{2A_1A_4}{\epsilon A_3}+2\nu A_4j\biggr)\le\beta^{j+1},
\end{align*}
where $A_3, A_4$ are positive constants, and $\beta=\frac{\zeta \nu}{\zeta \nu+\epsilon A_3/2}<1$.
\end{lemma}

\begin{sloppypar}
Note that  $\expect[\|\steadyperp\|]<+\infty$ since $\expect[\sum_r \weight_r\overline{N}_r]<+\infty$ by Lemma~\ref{LEM-SUM-NR-EXP}. Let $A$ denote $\frac{A_1}{\xi_1\sqrt{\epsilon}}$. Fix a $c\ge A-\nu $. Let $\hat{V}(\bm{n})=\max\{c,\|\nperp\|\}$. Let $\overline{\pi}$ denote the distribution of $\steady$. Note that $\expect[\hat{V}(\steady)]<+\infty$ and $\overline{q}<+\infty$.  
{
Then since $\overline{\pi}$ is the stationary distribution, we have
\begin{equation*}
\sum_{\bm{n}}\overline{\pi}(\bm{n})\sum_{\bm{n}'\colon \bm{n}'\neq\bm{n}}Q_{\bm{n}\bm{n}'}\hat{V}(\bm{n}')=\sum_{\bm{n}'}\hat{V}(\bm{n}')\sum_{\bm{n}\colon \bm{n}\neq\bm{n}'}\overline{\pi}(\bm{n})Q_{\bm{n}\bm{n}'}
=-\sum_{\bm{n}'}\hat{V}(\bm{n}')\overline{\pi}(\bm{n}')Q_{\bm{n}'\bm{n}'},
\end{equation*}
and
\begin{equation*}
\sum_{\bm{n}}\overline{\pi}(\bm{n})\sum_{\bm{n}'\colon \bm{n}'\neq\bm{n}}Q_{\bm{n}\bm{n}'}\hat{V}(\bm{n})
=\sum_{\bm{n}}\overline{\pi}(\bm{n})\hat{V}(\bm{n})\sum_{\bm{n}'\colon \bm{n}'\neq\bm{n}}Q_{\bm{n}\bm{n}'}
=-\sum_{\bm{n}}\hat{V}(\bm{n})\overline{\pi}(\bm{n})Q_{\bm{n}\bm{n}},
\end{equation*}
which imply that $\sum_{\bm{n}}\overline{\pi}(\bm{n})\sum_{\bm{n}'\colon \bm{n}'\neq\bm{n}}Q_{\bm{n}\bm{n}'}(\hat{V}(\bm{n}')-\hat{V}(\bm{n}))=0$.
}
Then similar to the proof of the exponential-type tail bounds in \cite{BerGamTsi_01}, we have
\begin{align}
0&=\sum_{\bm{n}}\overline{\pi}(\bm{n})\sum_{\bm{n}'\colon \bm{n}'\neq\bm{n}}Q_{\bm{n}\bm{n}'}(\hat{V}(\bm{n}')-\hat{V}(\bm{n}))\nonumber\\
&=\sum_{\bm{n}\colon \|\nperp\|\le c-\nu }\overline{\pi}(\bm{n})\sum_{\bm{n}'\colon \bm{n}'\neq\bm{n}}Q_{\bm{n}\bm{n}'}(\hat{V}(\bm{n}')-\hat{V}(\bm{n}))\label{eq-term1-exp}\\
&\mspace{23mu}+\mspace{-9mu}\sum_{\bm{n}\colon c-\nu <\|\nperp\|\le c+\nu }\mspace{-12mu}\overline{\pi}(\bm{n})\sum_{\bm{n}'\colon \bm{n}'\neq\bm{n}}Q_{\bm{n}\bm{n}'}(\hat{V}(\bm{n}')-\hat{V}(\bm{n}))\label{eq-term2-exp}\\
&\mspace{23mu}+\sum_{\bm{n}\colon\|\nperp\|>c+\nu }\overline{\pi}(\bm{n})\sum_{\bm{n}'\colon \bm{n}'\neq\bm{n}}Q_{\bm{n}\bm{n}'}(\hat{V}(\bm{n}')-\hat{V}(\bm{n})).\label{eq-term3-exp}
\end{align}
\end{sloppypar}
\begin{enumerate}[leftmargin=0em,itemindent=3.2em,label=(\roman*)]
\item The first summand \eqref{eq-term1-exp} is $0$ since when $\|\nperp\|\le c-\nu $, $\hat{V}(\bm{n}')=\hat{V}(\bm{n})=c$ for $\bm{n}'$ with $Q_{\bm{n}\bm{n}'}>0$.
	
\item Consider the second summand \eqref{eq-term2-exp}. We can check that for any two states $\bm{n}$ and $\bm{n}'$, either $0\le \hat{V}(\bm{n}')-\hat{V}(\bm{n})\le \|\nperp'\|-\|\nperp\|$, or $\|\nperp'\|-\|\nperp\|\le \hat{V}(\bm{n}')-\hat{V}(\bm{n})\le 0$, regardless of the relation between $c$ and $\|\nperp'\|$, $\|\nperp\|$. Then,
\begin{align*}
\sum_{\bm{n}'\colon \bm{n}'\neq \bm{n}}Q_{\bm{n}\bm{n}'}(\hat{V}(\bm{n}')-\hat{V}(\bm{n}))
&=\sum_{\bm{n}'\colon \hat{V}(\bm{n}')>\hat{V}(\bm{n})}Q_{\bm{n}\bm{n}'}(\hat{V}(\bm{n}')-\hat{V}(\bm{n}))+\mspace{-18mu}\sum_{\bm{n}'\colon \hat{V}(\bm{n}')\le\hat{V}(\bm{n})}Q_{\bm{n}\bm{n}'}(\hat{V}(\bm{n}')-\hat{V}(\bm{n}))\\
&\le\sum_{\bm{n}'\colon \hat{V}(\bm{n}')>\hat{V}(\bm{n})}Q_{\bm{n}\bm{n}'}(\hat{V}(\bm{n}')-\hat{V}(\bm{n}))\\
&\le \sum_{\bm{n}'\colon \|\nperp'\|>\|\nperp\|}Q_{\bm{n}\bm{n}'}\nu \\
&\le \zeta \nu .
\end{align*}
Thus the second summand satisfies
\begin{align*}
\sum_{\bm{n}\colon c-\nu <\|\nperp\|\le c+\nu }\overline{\pi}(\bm{n})\sum_{\bm{n}'\colon \bm{n}'\neq\bm{n}}Q_{\bm{n}\bm{n}'}(\hat{V}(\bm{n}')-\hat{V}(\bm{n}))\le \zeta \nu \Bigl(\Pr(\|\steadyperp\|>c-\nu )-\Pr(\|\steadyperp\|> c+\nu )\Bigr).
\end{align*}

\item Consider the third summand \eqref{eq-term3-exp}. Note that this is the part that differentiates our tail bounds from the bounds in \cite{BerGamTsi_01}.  Here we need to utilize the multiplicative condition in \eqref{eq-conditions-drift-exp}.
When $\|\nperp\|>c+\nu $, $\hat{V}(\bm{n})=\|\nperp\|$ and $\hat{V}(\bm{n}')=\|\nperp'\|$ for $\bm{n}'$ with $Q_{\bm{n}\bm{n}'}>0$. Therefore,
\begin{align}
&\mspace{23mu}\sum_{\bm{n}\colon\|\nperp\|>c+\nu }\overline{\pi}(\bm{n})\sum_{\bm{n}'\colon \bm{n}'\neq\bm{n}}Q_{\bm{n}\bm{n}'}(\hat{V}(\bm{n}')-\hat{V}(\bm{n}))\nonumber\\
&=\sum_{\substack{\bm{n}\colon\|\nperp\|>c+\nu \\\frac{\|\nperp\|}{\sum_r\weight_r n_r}\ge\frac{\xi_2\sqrt{\epsilon}}{A_2}}}
\mspace{-18mu}\overline{\pi}(\bm{n})\Delta \|\nperp\|+\mspace{-24mu}\sum_{\substack{\bm{n}\colon\|\nperp\|>c+\nu \\\frac{\|\nperp\|}{\sum_r\weight_r n_r}<\frac{\xi_2\sqrt{\epsilon}}{A_2}}}\mspace{-18mu}\overline{\pi}(\bm{n})\Delta \|\nperp\|\nonumber\\
&\le-\sqrt{\epsilon}\Pr\biggl(\|\steadyperp\|>c+\nu ,\frac{\|\steadyperp\|}{\sum_r \weight_r\overline{N}_r}\ge\frac{\xi_2\sqrt{\epsilon}}{A_2}\biggr)+(\xi_1+1)\sqrt{\epsilon}\Pr\biggl(\|\steadyperp\|>c+\nu ,\frac{\|\steadyperp\|}{\sum_r \weight_r\overline{N}_r}<\frac{\xi_2\sqrt{\epsilon}}{A_2}\biggr)\label{eq-ineq-1}\\
&=-\sqrt{\epsilon}\Pr(\|\steadyperp\|>c+\nu )+\xi_2\sqrt{\epsilon}\Pr\biggl(\|\steadyperp\|>c+\nu ,\frac{\|\steadyperp\|}{\sum_r \weight_r\overline{N}_r}<\frac{\xi_2\sqrt{\epsilon}}{A_2}\biggr)\label{eq-eq-2}\\
&\le -\sqrt{\epsilon}\Pr(\|\steadyperp\|>c+\nu )+\xi_2\sqrt{\epsilon}\Pr\Biggl(\sum_r \weight_r\overline{N}_r>\frac{(c+\nu )A_2}{\xi_2\sqrt{\epsilon}}\Biggr).\nonumber
\end{align}
The inequality \eqref{eq-ineq-1} follows from the drift bounds given in Lemma~\ref{lem-nperp-drift-exp}, and \eqref{eq-eq-2} follows from the choice of $\xi_1$ and $\xi_2$ in the proof of Lemma~\ref{lem-nperp-drift-exp}.
\end{enumerate}

Combining the three summands we have
\begin{align*}
\Pr(\|\steadyperp\|>c+\nu )
\le\frac{\zeta \nu }{\zeta \nu +\sqrt{\epsilon}}\Pr(\|\steadyperp\|>c-\nu)+\frac{\xi_2\sqrt{\epsilon}}{\zeta \nu +\sqrt{\epsilon}}\Pr\Biggl(\sum_r \weight_r\overline{N}_r>\frac{(c+\nu )A_2}{\xi_2\sqrt{\epsilon}}\Biggr).
\end{align*}
Recall that we let $\alpha$ denote $\frac{\zeta \nu }{\zeta \nu +\sqrt{\epsilon}}$. Let $c = A+(2j-1)\nu $ for a nonnegative integer $j$. Then
\begin{align*}
\Pr(\|\steadyperp\|>A+2\nu j)
\le \alpha\Pr(\|\steadyperp\|>A+2\nu (j-1))+\xi_2(1-\alpha)\Pr\Biggl(\sum_r \weight_r\overline{N}_r>\frac{(A+2\nu j)A_2}{\xi_2\sqrt{\epsilon}}\Biggr).
\end{align*}

\iftoggle{complete}{%
Now we use Lemma~\ref{LEM-SUM-NR-EXP} to bound the last probability above. Recall that we have chosen $\xi_1$ and $\xi_2$ in the proof of Lemma~\ref{lem-nperp-drift-exp} such that $\xi_2-\xi_1=2$. We can further require that $\xi_1\xi_2=\frac{A_2A_3}{2A_4}$.
It can be verified that such constants $\xi_1$ and $\xi_2$ are well-defined since $A_2, A_3$ and $A_4$ are all positive.
Also recall that $A=\frac{A_1}{\xi_1\sqrt{\epsilon}}$. Define a constant $\theta = \frac{\nu A_2}{\xi_2\nu A_4}$.
Then
\begin{align*}
&\mspace{23mu}\Pr\Biggl(\sum_r \weight_r\overline{N}_r>\frac{(A+2\nu j)A_2}{\xi_2\sqrt{\epsilon}}\Biggr)\\
&=\Pr\Biggl(\sum_r \weight_r\overline{N}_r>\frac{A_1A_2}{\xi_1\xi_2\epsilon}+\frac{2\nu A_2}{\xi_2\sqrt{\epsilon}}j\Biggr)\\
&\le \Pr\Biggl(\sum_r \weight_r\overline{N}_r>\frac{2A_1A_4}{\epsilon A_3}+2\nu A_4\biggl\lfloor \frac{j\theta}{\sqrt{\epsilon}}\biggr\rfloor\Biggr)\\
&\le \beta^{\lfloor j\theta/\sqrt{\epsilon}\rfloor+1}\\
&\le \Bigl(\beta^{\theta/\sqrt{\epsilon}}\Bigr)^j.
\end{align*}
Therefore,
\begin{align*}
\Pr(\|\steadyperp\|>A+2\nu j)\le \alpha\Pr(\|\steadyperp\|>A+2\nu (j-1))+\xi_2(1-\alpha)\Bigl(\beta^{\theta/\sqrt{\epsilon}}\Bigr)^j.
\end{align*}
Using this inequality for $k-1,k-2,\cdots$ yields
\begin{align*}
\Pr\biggl(\|\steadyperp\|>\frac{A_1}{\xi_1\sqrt{\epsilon}}+2\nu j\biggr)\le\alpha^{j+1}+\xi_2(1-\alpha)\sum_{i=0}^j\alpha^i\Bigl(\beta^{\theta/\sqrt{\epsilon}}\Bigr)^{j-i}.
\end{align*}
This completes the proof of the tail bound in Lemma~\ref{lem-distr-bound-exp}.

Now we bound the moments of $\|\steadyperp\|$ using the tail bound above. Recall that $A=\frac{A_1}{\xi_1\sqrt{\epsilon}}$. For any nonnegative integer $m$,
\begin{align}
\expect\Bigl[\|\steadyperp\|^m\Bigr]
&=\int_{0}^{\infty}mt^{m-1}\Pr(\|\steadyperp\|>t)dt\nonumber\\
&=\int_{0}^{A}mt^{m-1}\Pr(\|\steadyperp\|>t)dt+\sum_{j=0}^{\infty}\int_{A+2\nu j}^{A+2\nu (j+1)}mt^{m-1}\Pr(\|\steadyperp\|>t)dt\nonumber\\
&\le A^m+\sum_{j=0}^{\infty}\Biggl(\alpha^{j+1}+\xi_2(1-\alpha)\sum_{i=0}^j\alpha^i\Bigl(\beta^{\theta/\sqrt{\epsilon}}\Bigr)^{j-i}\Biggr)\cdot\Bigl((A+2\nu (k+1))^m-(A+2\nu j)^m\Bigr)\label{eq-moment-sum-1}\\
&\le A^m+\sum_{j=0}^{\infty}\Biggl(\alpha^{j+1}+\xi_2(1-\alpha)\sum_{i=0}^j\alpha^i\Bigl(\beta^{\theta/\sqrt{\epsilon}}\Bigr)^{j-i}\Biggr)\cdot\Bigl(\nu 2^{m-1}A^{m-1}+\nu 2^{m-1}(k+1)^{m-1}\Bigr),\label{eq-moment-sum}
\end{align}
where \eqref{eq-moment-sum-1} follows from the tail bound, and \eqref{eq-moment-sum} is due to the convexity of the function $f(x)=x^m$. We bound the summands in \eqref{eq-moment-sum} as follows, where we have used the inequality $\sum_{i=1}^{\infty}a^ii^m\le m!\frac{1}{(1-a)^{m+1}}$ for any $0<a<1$:
\begin{equation}
\sum_{j=0}^{\infty}\alpha^{j+1}\nu 2^{m-1}A^{m-1}=\nu 2^{m-1}\alpha\frac{A^{m-1}}{1-\alpha},
\end{equation}
\begin{equation}
\begin{split}
\sum_{j=0}^{\infty}\alpha^{j+1}\nu 2^{m-1}(k+1)^{m-1}&=\nu 2^{m-1}\sum_{j=1}^{\infty}\alpha^jk^{m-1}\\
&\le\nu 2^{m-1}(m-1)!\frac{1}{(1-\alpha)^{m}},
\end{split}
\end{equation}
\begin{equation}
\begin{split}
&\mspace{23mu}\sum_{j=0}^{\infty}\xi_2(1-\alpha)\sum_{i=0}^j\alpha^i\Bigl(\beta^{\theta/\sqrt{\epsilon}}\Bigr)^{j-i}\nu 2^{m-1}A^{m-1}\\
&=\xi_2(1-\alpha)\nu 2^{m-1}A^{m-1}\sum_{i=0}^{\infty}\alpha^i\sum_{j=i}^{\infty}\Bigl(\beta^{\theta/\sqrt{\epsilon}}\Bigr)^{j-i}\\
&=\xi_2\nu 2^{m-1}\frac{A^{m-1}}{1-\beta^{\theta/\sqrt{\epsilon}}},
\end{split}
\end{equation}
\begin{equation}
\begin{split}
&\mspace{23mu}\sum_{j=0}^{\infty}\xi_2(1-\alpha)\sum_{i=0}^j\alpha^i\Bigl(\beta^{\theta/\sqrt{\epsilon}}\Bigr)^{j-i}\nu 2^{m-1}(k+1)^{m-1}\\
&=\xi_2(1-\alpha)\nu 2^{m-1}\sum_{i=0}^{\infty}\alpha^i\sum_{j=i}^{\infty}\Bigl(\beta^{\theta/\sqrt{\epsilon}}\Bigr)^{j-i}(k+1)^{m-1}\\
&\le\xi_2(1-\alpha)\nu 2^{2m-1}\sum_{i=0}^{\infty}\alpha^i\sum_{j=i}^{\infty}\Bigl(\beta^{\theta/\sqrt{\epsilon}}\Bigr)^{j-i}\cdot\Bigl((k-i)^{m-1}+(i+1)^{m-1}\Bigr)\\
&\le\xi_2\nu 2^{2m-1}(m-1)!\Biggl(\frac{1}{(1-\beta^{\theta/\sqrt{\epsilon}})^{m}}+\frac{1}{\alpha}\frac{1}{(1-\alpha)^{m-1}(1-\beta^{\theta/\sqrt{\epsilon}})}
\Biggr).
\end{split}
\end{equation}

By the definitions of $\alpha,\beta$ and $A$, $A=O\left(\frac{1}{\sqrt{\epsilon}}\right),\frac{1}{1-\alpha}=O\left(\frac{1}{\sqrt{\epsilon}}\right)$.
We also have that
\begin{align*}
1-\beta^{\theta/\sqrt{\epsilon}}&= 1-\biggl(1-\frac{\epsilon}{2\zeta \nu /A_3+\epsilon}\biggr)^{\theta/\sqrt{\epsilon}}\\
&=1-e^{\frac{\theta}{\sqrt{\epsilon}}\ln\Bigl(1-\frac{\epsilon}{2\zeta \nu /A_3+\epsilon}\Bigr)}\\
&=\frac{\theta\sqrt{\epsilon}}{2\zeta \nu /A_3+\epsilon}+o(\sqrt{\epsilon}).
\end{align*}
Thus $1/(1-\beta^{\theta/\sqrt{\epsilon}})=O(1/\sqrt{\epsilon})$.

Now we can see that each summand in \eqref{eq-moment-sum} is $O\bigl((1/{\sqrt{\epsilon}})^m\bigr)$, which implies that $\expect\Bigl[\bigl\|\steadyperp^{(\epsilon)}\bigr\|^m\Bigr]=O\bigl((1/{\sqrt{\epsilon}})^m\bigr)$.
\Halmos
}{%
Now we insert the tail bounds in Lemma~\ref{LEM-SUM-NR-EXP} and further require that $\xi_1\xi_2=\frac{A_2A_3}{2A_4}$.
It can be verified that such constants $\xi_1$ and $\xi_2$ are well-defined since $A_2, A_3$ and $A_4$ are all positive.   Letting $\theta = \frac{\nu A_2}{\xi_2\nu A_4}$ yields the tail bounds \eqref{eq-nperp-tail-bounds-exp} in Lemma~\ref{lem-distr-bound-exp}.  Then the moment bounds \eqref{eq-nperp-moment-bounds-exp} in Lemma~\ref{lem-distr-bound-exp} can be obtained using arguments similar to those in Lemma~3 in \cite{MagSri_16}.  The details are given in our technical report \cite{WanMagSri_20} due to space limitations. \Halmos
}

\subsection{Proof of Theorem~\ref{THM-BOUNDS-EXP} (Bounds on Flow Counts).}\label{subsec-proof-thm-bounds-exp}
Recall that by the arguments in Section~\ref{subsec-bounds-exp}, setting $\expect[\Delta \|\steadyparaspace\|^2]=0$ yields
\begin{equation}\label{eq-drift-equal-1}
\begin{split}
\epsilon\expect\left[\sum_r \weight_r \overline{N}_r\right]&=\epsilon\expect[\langle \steadyperp^s,\bm{\lambda}^{(0)}\rangle]+\expect[\langle\steadyparaspace,\bm{\lambda}^{(0)}-\steady\bm{x\mu}\rangle]+\frac{1}{2}\expect[B_1(\steady)].
\end{split}
\end{equation}
We analyze the terms on the right-hand-side.
\begin{enumerate}[leftmargin=0em,itemindent=3em,label=({\roman*})]
\item We first consider the term $\epsilon\expect[\langle \steadyperp^s,\bm{\lambda}^{(0)}\rangle]$. By Cauchy-Schwarz inequality,
\begin{align*}
\left|\epsilon\expect[\langle \steadyperp^s,\bm{\lambda}^{(0)}\rangle]\right|\le \epsilon\|\steadyperp^s\|\cdot\|\bm{\lambda}^{(0)}\|\le \epsilon\|\steadyperp\|\cdot\|\bm{\lambda}^{(0)}\|=O(\sqrt{\epsilon}).
\end{align*}

\item We next consider the term $\expect[\langle\steadyparaspace,\bm{\lambda}^{(0)}-\steady\bm{x\mu}\rangle]$. Recall that we have shown in Section~\ref{subsec-bounds-exp} that $\langle\steadyparaspace,\bm{\lambda}^{(0)}-\steady\bm{x\mu}\rangle=\sum_{\ell\in\satlinks}\alpha_{\ell}^sU_{\ell}$.
By Lemma~\ref{lem-continuity2-exp} and H\"{o}lder's inequality, for each $\ell\in\satlinks$,
\begin{align*}
\expect[|\alpha_{\ell}^sU_{\ell}|]=\expect[|(\alpha_{\ell}^s-p_{\ell})U_{\ell}|]\le B_2\Biggl(\expect\Biggl[\|\steadyperp\|^{\frac{\tau_1}{2}}\biggl(\sum_r\weight_r \overline{N}_r\biggr)^{\frac{\tau_1}{2}}\Biggr]\Biggr)^{\frac{1}{\tau_1}}\Bigl(\expect[U_{\ell}^{\tau_2}]\Bigr)^{\frac{1}{\tau_2}},
\end{align*}
where we pick $\tau_1$ and $\tau_2$ such that $\tau_1$ is an even integer with $\tau_1>4$ and $\frac{1}{\tau_1}+\frac{1}{\tau_2}=1$. The proof of Lemma~\ref{lem-continuity2-exp} is given at the end of this section. Using Cauchy-Schwarz inequality we have
\begin{align*}
\Biggl(\expect\Biggl[\|\steadyperp\|^{\frac{\tau_1}{2}}\biggl(\sum_r\weight_r \overline{N}_r\biggr)^{\frac{\tau_1}{2}}\Biggr]\Biggr)^{\frac{1}{\tau_1}}
&\le(\expect[\|\steadyperp\|^{\tau_1}])^{\frac{1}{2\tau_1}}
\Biggl(\expect\Biggl[\biggl(\sum_r\weight_r \overline{N}_r\biggr)^{\tau_1}\Biggr]\Biggr)^{\frac{1}{2\tau_1}}\\
&=O(\epsilon^{-\frac{3}{4}}),
\end{align*}
where again the last equality follows from the state-space collapse result in Theorem~\ref{THM-SSC-EXP} and the bound on $\expect[\sum_r\weight_r\overline{N}_r]$ indicated by Lemma~\ref{LEM-SUM-NR-EXP}. Next we bound $\expect[U_{\ell}^{\tau_2}]$. We can prove that $\expect[U_{\ell}]=\epsilon C_{\ell}$ by considering the Lyapunov function $w_{\ell}(\bm{n})=\langle \bm{b}^{(\ell)},\bm{n}\rangle$. Its drift is $\Delta w_{\ell}(\bm{n})=-\epsilon\langle \bm{b}^{(\ell)},\bm{\lambda}^{(0)}\rangle+\langle \bm{b}^{(\ell)},\bm{\lambda}^{(0)}-\bm{nx\mu}\rangle=-\epsilon C_{\ell}+U_{\ell}$.
Since in the steady state $\expect[\Delta w_{\ell}(\steady)]=0$, we have $\expect[U_{\ell}]=\epsilon C_{\ell}$.
Since $0\le U_{\ell}\le C_{\ell}$, there holds $\Bigl(\expect\Bigl[U_{\ell}^{\tau_2}\Bigr]\Bigr)^{\frac{1}{\tau_2}}\le\Bigl(\expect\Bigl[U_{\ell}\cdot C_{\ell}^{\tau_2-1}\Bigr]\Bigr)^{\frac{1}{\tau_2}}=\epsilon^{\frac{1}{\tau_2}} C_{\ell}$.
Combining these bounds we have $\expect[|\alpha_{\ell}^sU_{\ell}|]=O(\epsilon^{\frac{1}{4}-\frac{1}{\tau_1}})$, and thus
\begin{align*}
\expect[|\langle \steadyparaspace,\bm{\lambda}^{(0)}-\steady\bm{x\mu}\rangle|]&=\expect\Biggl[\biggl|\sum_{\ell\in\satlinks}\alpha_{\ell}^sU_{\ell}\biggr|\Biggr]=O(\epsilon^{\frac{1}{4}-\frac{1}{\tau_1}}).
\end{align*}

\item Lastly, we bound the last term $\expect[B_1(\steady)]/2$. Recall that we have shown in Section~\ref{subsec-bounds-exp} that
$\frac{1}{2}\expect[B_1(\steady)]=\sum_r\lambda_r\|(\bm{e}^{(r)})_{\shortparallel}^s\|^2=(1-\epsilon)\sum_r\lambda^{(0)}_r\langle(\bm{e}^{(r)})_{\shortparallel}^s,(\bm{e}^{(r)})_{\shortparallel}^s\rangle$.
Let $M=\textrm{diag}(\bm{\weight}/\bm{\lambda}^{(0)})$ denote the diagonal matrix whose diagonal consists of entries of the vector $\bm{\weight}/\bm{\lambda}^{(0)}\triangleq(\weight_r/\lambda^{(0)}_r)_{r\in\mathcal{R}}$. Then the inner product can be written in a matrix form: $\langle \bm{y},\bm{z}\rangle =\bm{y}^T M\bm{z}$ for any $\bm{y},\bm{z}\in\mathbb{R}^R$.
Then
\begin{align*}
\frac{1}{2}\expect[B_1(\steady)]&\le(1-\epsilon)\left(\max_r\weight_r\right)\sum_r(\lambda^{(0)}_r/\weight_r)\langle(\bm{e}^{(r)})_{\shortparallel}^s,(\bm{e}^{(r)})_{\shortparallel}^s\rangle\\
&=(1-\epsilon)\left(\max_r\weight_r\right)\sum_r(\lambda^{(0)}_r/\weight_r)\langle \bm{e}^{(r)},(\bm{e}^{(r)})_{\shortparallel}^s\rangle\\
&=(1-\epsilon)\left(\max_r\weight_r\right)\sum_r\langle M^{-1}\bm{e}^{(r)},(\bm{e}^{(r)})_{\shortparallel}^s\rangle\\
&=(1-\epsilon)\left(\max_r\weight_r\right)\sum_r(\bm{e}^{(r)})^T(\bm{e}^{(r)})_{\shortparallel}^s.
\end{align*}
We now express $(\bm{e}^{(r)})_{\shortparallel}^s$ in a matrix form. Let $B_s$ denote the matrix whose rows are $(\bm{b}^{(\ell)})^T$'s with $\ell\in\satlinks$. Then $(\bm{e}^{(r)})_{\shortparallel}^s=B_s^T(B_sMB_s^T)^{-1}B_sM\bm{e}^{(r)}$. Thus
\begin{align}
\frac{1}{2}\expect[B_1(\steady)]&\le(1-\epsilon)\left(\max_r\weight_r\right)\sum_r (\bm{e}^{(r)})^TB_s^T(B_sMB_s^T)^{-1}B_sM\bm{e}^{(r)}\nonumber\\
&=(1-\epsilon)\left(\max_r\weight_r\right)\textrm{tr}(B_s^T(B_sMB_s^T)^{-1}B_sM)\nonumber\\
&=(1-\epsilon)\left(\max_r\weight_r\right)\textrm{tr}(B_sMB_s^T(B_sMB_s^T)^{-1})\label{eq-trace-exchangable}\\
&=(1-\epsilon)\left(\max_r\weight_r\right)L_s,\nonumber
\end{align}
\end{enumerate}
where the notation $\textrm{tr}(\cdot)$ denotes the trace of a matrix, and \eqref{eq-trace-exchangable} follows from that $\textrm{tr}(XY)=\textrm{tr}(YX)$ for any matrices $X$ and $Y$, and recall that $L_s$ is the number of critically loaded links in the network. Similarly, we can show that
\begin{equation*}
\frac{1}{2}\expect[B_1(\steady)]\ge(1-\epsilon)\left(\min_r\weight_r\right)L_s.
\end{equation*}

Combining (i), (ii) and (iii) for the terms in \eqref{eq-drift-equal-1} yields
\begin{equation*}
\frac{L_s\cdot \min_r \weight_r}{\epsilon}+o\biggl(\frac{1}{\epsilon}\biggr)\le \expect\left[\sum_r \weight_r \overline{N}_r^{(\epsilon)}\right]
\le\frac{L_s\cdot \max_r \weight_r}{\epsilon}+o\biggl(\frac{1}{\epsilon}\biggr).
\end{equation*}
The proof will be completed after we prove Lemma~\ref{lem-continuity2-exp} below.

\proof{Proof of Lemma~\ref{lem-continuity2-exp}.}
We first give the following bounds on the rate allocation and the Lagrange multipliers.
\begin{claim}\label{claim-continuity-instant-exp}
There exist positive constants $B_4$ and $B_5$ such that for any state $\bm{n}$,
\begin{equation*}
\sum_{\ell} p_{\ell}\|\widehat{\bm{b}}^{(\ell)}-\bm{b}^{(\ell)}\|^2\le B_4\|\nperp\|,\text{ and }\sum_{\ell:\ell\in r,\ell\in\mathcal{L}\setminus\satlinks}p_{\ell}\le B_5\|\nperp\|,\forall r.
\end{equation*}
\end{claim}
\proof{Proof of Claim~\ref{claim-continuity-instant-exp}.}
Consider the term $\langle \bm{n},\bm{\lambda}^{(0)}-\bm{nx\mu}\rangle$.  By the proof of Lemma~\ref{lem-nperp-drift-exp} in Section~\ref{subsec-proof-lem-nperp-drift-exp}, we know that
\begin{align}
\langle \bm{n},\bm{\lambda}^{(0)}-\bm{nx\mu}\rangle \le -\weight_{\min}\mu_{\min}\sum_{\ell}p_{\ell}\|\widehat{\bm{b}}^{(\ell)}-\bm{b}^{(\ell)}\|^2-\sum_{\ell\in\mathcal{L}\setminus\satlinks}p_{\ell}\delta_{\ell}.\label{eq-left-exp}
\end{align}
On the other hand, by the duality principle for minimum norm problems \citep{Lue_69},
\begin{equation*}
\|\nperp\|=\sup_{\bm{y}\in\mathcal{K}^{\circ}\colon \|\bm{y}\|\le 1}\langle \bm{n},\bm{y}\rangle,
\end{equation*}
where $\mathcal{K}^{\circ}$ is the polar cone of the cone $\mathcal{K}$. Let $\bm{y}=\frac{-\bm{\lambda}^{(0)}+\bm{nx\mu}}{\|\bm{\lambda}^{(0)}-\bm{nx\mu}\|}$.
We can easily verify that $\bm{y}\in\mathcal{K}^{\circ}$ since $\langle \bm{b}^{(\ell)},-\bm{\lambda}^{(0)}+\bm{nx\mu}\rangle = -U_{\ell}\le 0$ for all $\ell\in\satlinks$ by \ref{property-P1}. Thus
\begin{align}
-\langle\bm{n}, \bm{\lambda}^{(0)}-\bm{nx\mu}\rangle=\|\bm{\lambda}^{(0)}-\bm{nx\mu}\|\langle \bm{n},\bm{y}\rangle
\le \|\bm{\lambda}^{(0)}-\bm{nx\mu}\|\|\nperp\|.\label{eq-duality-exp}
\end{align}
Combining \eqref{eq-left-exp} and \eqref{eq-duality-exp} gives
\begin{equation*}
\weight_{\min}\mu_{\min}\sum_{\ell}p_{\ell}\|\widehat{\bm{b}}^{(\ell)}-\bm{b}^{(\ell)}\|^2+\sum_{\ell\in\mathcal{L}\setminus\satlinks}p_{\ell}\delta_{\ell}\le \|\bm{\lambda}^{(0)}-\bm{nx\mu}\|\|\nperp\|.
\end{equation*}

Since each summand on the left-hand-side above is nonnegative, we have that
\begin{equation*}
\sum_{\ell}p_{\ell}\|\widehat{\bm{b}}^{(\ell)}-\bm{b}^{(\ell)}\|^2\le\frac{\|\bm{\lambda}^{(0)}-\bm{nx\mu}\|}{\weight_{\min}\mu_{\min}}\|\nperp\|,
\end{equation*}
and for each $r$,
\begin{equation*}
\sum_{\ell:\ell\in r,\ell\in\mathcal{L}\setminus\satlinks}p_{\ell}\le\frac{1}{\delta_{\min}}\sum_{\ell:\ell\in r,\ell\in\mathcal{L}\setminus\satlinks}p_{\ell}\delta_{\ell}\le\frac{\|\bm{\lambda}^{(0)}-\bm{nx\mu}\|}{\delta_{\min}}\|\nperp\|,
\end{equation*}
where $\delta_{\min}>0$ is defined as $\delta_{\min}=\min\{\delta_{\ell}\colon \ell\in\mathcal{L}\setminus \satlinks\}$.
Note that each entry of the rate allocation $\bm{nx}$ can be bounded using a constant independent of $\bm{n}$ and $\epsilon$ (e.g., total bandwidth capacity in the network), there exist positive constant $B_3$ and $B_4$ such that $\frac{\|\bm{\lambda}^{(0)}-\bm{nx\mu}\|}{\weight_{\min}\mu_{\min}}\le B_4$ and $\frac{\|\bm{\lambda}^{(0)}-\bm{nx\mu}\|}{\delta_{\min}}\le B_5$. Therefore,
\begin{equation*}
\sum_{\ell} p_{\ell}\|\widehat{\bm{b}}^{(\ell)}-\bm{b}^{(\ell)}\|^2\le B_4\|\nperp\|,\text{ and }\sum_{\ell:\ell\in r,\ell\in\mathcal{L}\setminus\satlinks}p_{\ell}\le B_5\|\nperp\|,\forall r,
\end{equation*}
which completes the proof of the claim. \Halmos
\endproof

With Claim~\ref{claim-continuity-instant-exp}, it is easy to verify \eqref{eq-bound-p-exp} in Lemma~\ref{lem-continuity2-exp} by setting $B_3=B_5R$ where $R$ is the total number of routes. So next we focus on bounding $|\alpha_{\ell}^s-p_{\ell}|$ for $\ell\in\satlinks$ and thus proving \eqref{eq-diff-alpha-p-exp} in Lemma~\ref{lem-continuity2-exp}. We first write each $(n_{\perp}^s)_r=n_r-(n_{\shortparallel}^s)_r$ in the following form
\begin{align*}
(n_{\perp}^s)_r&=n_r-(n_{\shortparallel}^s)_r\\
&=\frac{n_rx_r}{\weight_r}\sum_{\ell:\ell\in r}p_{\ell}-\frac{\rho_r^{(0)}}{\weight_r}\sum_{\ell:\ell\in r,\ell\in\satlinks}\alpha_{\ell}^s\\
&=\frac{\rho_r^{(0)}}{\weight_r}\sum_{\ell:\ell\in r,\ell\in\satlinks}(p_{\ell}-\alpha_{\ell}^s)+\frac{n_rx_r-\rho_r^{(0)}}{\weight_r}\sum_{\ell:\ell\in r}p_{\ell}+\frac{\rho_r^{(0)}}{\weight_r}\sum_{\ell:\ell\in r,\ell\in\mathcal{L}\setminus\satlinks}p_{\ell}.
\end{align*}
Then for each route $r$,
\begin{align}
\Biggl|\sum_{\ell:\ell\in r,\ell\in\satlinks}(\alpha_{\ell}^s-p_{\ell})\Biggr|\le\frac{1}{\rho_r^{(0)}}\Bigl|n_rx_r-\rho_r^{(0)}\Bigr|\sum_{\ell:\ell\in r}p_{\ell}+B_5\|\nperp\|+\Biggl|\frac{\weight_r(n_{\perp}^s)_r}{\rho_r^{(0)}}\Biggr|,\label{eq-diff-bound-exp}
\end{align}
where we have used the second bound in Claim~\ref{claim-continuity-instant-exp}.  Due to the equivalence of norms, there exists a constant $B_6$ such that $\sum_{r:\ell\in r}(n_rx_r-\rho_r^{(0)})^2\le B_6\|\widehat{\bm{b}}^{(\ell)}-\bm{b}^{(\ell)}\|^2$.
Then
\begin{align}
\sum_{r}(n_rx_r-\rho_r^{(0)})^2\sum_{\ell:\ell\in r}p_{\ell}\le B_6\sum_{\ell} p_{\ell}\|\widehat{\bm{b}}^{(\ell)}-\bm{b}^{(\ell)}\|^2\le B_4B_6\|\nperp\|,\nonumber
\end{align}
where we have used the first bound in Claim~\ref{claim-continuity-instant-exp}.  Since each summand on the left hand side is nonnegative, we have that for each $r$,
\begin{equation*}
(n_rx_r-\rho_r^{(0)})^2\sum_{\ell:\ell\in r}p_{\ell}\le B_4B_6\|\nperp\|.
\end{equation*}
Inserting this to \eqref{eq-diff-bound-exp} we get
\begin{align}
\Biggl|\sum_{\ell:\ell\in r,\ell\in\satlinks}(\alpha_{\ell}^s-p_{\ell})\Biggr|
&\le\frac{\sqrt{B_4B_6}}{\rho_r^{(0)}}\|\nperp\|^{1/2}\Biggl(\sum_{\ell:\ell\in r}p_{\ell}\Biggr)^{1/2}+B_5\|\nperp\|+\Biggl|\frac{\weight_r(n_{\perp}^s)_r}{\rho_r^{(0)}}\Biggr|\nonumber\\
&\le\frac{\sqrt{B_4B_6}}{\rho_r^{(0)}\sqrt{C_{\min}}}\|\nperp\|^{1/2}\Biggl(\sum_r\weight_r n_r\Biggr)^{1/2}+B_5\|\nperp\|+\Biggl|\frac{\weight_r(n_{\perp}^s)_r}{\rho_r^{(0)}}\Biggr|,\label{eq-diff-bound-1-exp}
\end{align}
where \eqref{eq-diff-bound-1-exp} follows from that $\sum_{\ell:\ell\in r}p_{\ell}\le \sum_{\ell} p_{\ell}$ and $\sum_{\ell} p_{\ell}C_{\ell} = \sum_r\weight_r n_r$. Since $\|\nperpspace\|\le \|\nperp\|\le \|\bm{n}\|$, \eqref{eq-diff-bound-1-exp} indicates that there exists a constant $B_7$ such that
\begin{equation*}
\Biggl|\sum_{\ell:\ell\in r,\ell\in\satlinks}(\alpha_{\ell}^s-p_{\ell})\Biggr|\le B_7\|\nperp\|^{1/2}\Biggl(\sum_r\weight_r n_r\Biggr)^{1/2}.
\end{equation*}

Next we bound $|\alpha_{\ell}^s-p_{\ell}|$ for each link $\ell\in\satlinks$. Recall that $H$ is the routing matrix defined as $H=(h_{\ell r})_{\ell\in\mathcal{L},r\in\mathcal{R}}$ with $h_{\ell r}=1$ if $\ell\in r$, and $h_{\ell r}=0$ otherwise. Let $H_s$ be the submatrix of $H$ such that its rows correspond to links in $\satlinks$. Since we assume that $H_s$ has full row rank, $H_sH_s^T$ is invertible. Let $\overline{h}^{(\ell)T}$ be the $\ell$th row of $(H_sH_s^T)^{-1}H_s$ for $\ell\in\satlinks$. Then $\overline{h}^{(\ell)T}H_s^T=\overline{e}^{(\ell)}$, where $\overline{e}^{(\ell)}$ is a $L_s\times 1$ vector with the $\ell$th entry being $1$ and other entries being $0$. Thus
\begin{align*}
\alpha_{\ell}^s-p_{\ell}&=\overline{h}^{(\ell)T}H_s^T(\bm{\alpha}^s-\bm{p})=\sum_r\overline{h}^{(\ell)}_r\sum_{\ell':\ell'\in r,\ell'\in\satlinks}(\alpha_{\ell'}^s-p_{\ell'}).
\end{align*}
Therefore, we have the following bound for each $\ell\in\satlinks$
\begin{align*}
|\alpha_{\ell}^s-p_{\ell}|\le\sum_r\Bigl|\overline{h}^{(\ell)}_r\Bigr|\Biggl|\sum_{\ell':\ell'\in r,\ell'\in\satlinks}(\alpha_{\ell'}^s-p_{\ell'})\Biggr|\le\biggl(\sum_r\Bigl|\overline{h}^{(\ell)}_r\Bigr|\biggr)B_7\|\nperp\|^{1/2}\Biggl(\sum_r\weight_r n_r\Biggr)^{1/2}.
\end{align*}
Note that $\overline{h}^{(\ell)}_r$'s are constants independent of $\epsilon$, i.e., $B_2=(\sum_r|\overline{h}^{(\ell)}_r|)B_7$ is a constant. Thus,
\begin{equation*}
|\alpha_{\ell}^s-p_{\ell}|\le B_2\|\nperp\|^{1/2}\Biggl(\sum_r\weight_r n_r\Biggr)^{1/2},\quad \forall \ell \in\satlinks,
\end{equation*}
which completes the proof of Lemma~\ref{lem-continuity2-exp}. \Halmos
\endproof

\section{Phase-type file size distributions.}\label{sec-phase-type}
In this section we generalize our results for weighted proportionally fair policies to a class of phase-type file size distributions. The assumptions on the arrival processes and the heavy-traffic regime are the same as in previous sections. To address the phase-type file size distributions, our main idea is to construct a proper inner product in the state space that allows us to establish results by applying the drift method.  We will reuse many symbols in our notation, but redefine them for the setting of phase-type distributions.

\paragraph{A class of phase-type distributions.}
We assume that for each route~$r$, the file size distribution belongs to a special class of phase-type distributions specified below, denoted by $\mathcal{D}$.  It can be proved that any probability distribution on $[0,\infty)$ can be approximated arbitrarily closely by a distribution in $\mathcal{D}$, i.e., the class $\mathcal{D}$ is \emph{dense} in the space of probability distributions on $[0,\infty)$.  A phase-type distribution can be specified by the absorption time of a Markov chain.  For each route $r$:
\begin{itemize}[leftmargin=2em]
\item The Markov chain has $\Phase_r+1$ states, where state $0$ is an absorbing state and states $1,2,\dots,\Phase_r$ are transient states (or \emph{phases}).

\item The initial distribution is $(\pi_0,\bm{\pi}_r)$ with $\pi_0=0$, where $\bm{\pi}_r$ is a $1\times \Phase_r$ (row) vector.

\item The transition rate matrix is
	\begin{equation}
	\begin{bmatrix}
	0 & \bm{0}\\
	\bm{s}_r & \transition_r
	\end{bmatrix},
	\end{equation}
	where $\bm{s}_r$ is a $\Phase_r\times 1$ vector and $\transition_r$ is a $\Phase_r \times \Phase_r$ matrix.
\end{itemize}
We assume that $\transition_r$ is a block-diagonal matrix in the following form:
\begin{equation}\label{eq-transition_r-block}
\transition_r=
\begin{bmatrix}
\transition_r^{(1)} & 0 & \cdots & 0\\
0 & \transition_r^{(2)} & \cdots & 0\\
\vdots & \vdots & \ddots & \vdots \\
0 & 0 & \cdots & \transition_r^{(B_r)}
\end{bmatrix},
\end{equation}
where each matrix $\transition_r^{(b)}$ with $1\le b\le B_r$ has the following form:
\begin{equation}
\transition_r^{(b)}=
\begin{bmatrix}
-\mu_r^{(b)} & \mu_r^{(b)} & 0 & \cdots & 0\\
0 & -\mu_r^{(b)} & \mu_r^{(b)} & \cdots & 0\\
\vdots & \ddots & \ddots & \ddots & \vdots \\
0 & \cdots & 0 & -\mu_r^{(b)} & \mu_r^{(b)}\\
0 & \cdots & 0 & 0 & -\mu_r^{(b)}
\end{bmatrix}
\end{equation}
with distinct $\mu_r^{(b)}$'s for different $b$'s.  The initial distribution $\bm{\pi}_r$ has the following structure: $\pi_{r,\phase}>0$ for each phase $\phase$ that corresponds to the first row of some $\transition_r^{(b)}$, and $\pi_{r,\phase}\ge 0$ for other phases.  Note that if $\pi_{r,\phase}=0$ for those other phases, then the phase-type distribution is a finite mixture of Erlang distributions \cite{Asm_03}.
{
Then the fact that the class $\mathcal{D}$ we specified is dense has an intuitive explanation that is similar to the explanation for the set of finite mixtures of Erlangs to be dense.  In particular, it is known that any point mass can be approximated by an Erlang distribution by increasing the number of stages of the Erlang while keeping its mean unchanged; then since any distribution can be approximated arbitrarily closely by a distribution with a finite support, it can be approximated arbitrarily closely by a finite mixture of Erlang distributions.
}

For the phase-type distribution with parameter $(\bm{\pi}_r,S_r)$, the expected flow size is given by
\begin{equation}
\frac{1}{\mu_r}=\bm{\pi}_r(-\transition_r)^{-1}\bm{1}_{\Phase_r\times 1}.
\end{equation}
Then the load on route $r$ is $\rho_r=\lambda_r/\mu_r$.

\paragraph{State representation.}
With phase-type file size distributions, a Markovian representation of the flow dynamics consists of the flow counts for \emph{every phase} on every route. Let $N_{r,\phase}(t)$ denote the number of flows present on route $r$ that are in phase $\phase$ at time~$t$. Let
\begin{gather}
\bm{N}_r(t)=[N_{r,1}(t),\dots,N_{r,\Phase_r}(t)]^T,\\
\bm{N}(t)=[(\bm{N}_1(t))^T,\dots,(\bm{N}_R(t))^T]^T,
\end{gather}
i.e., $\bm{N}_r(t)$ is a vector stacking together the $N_{r,\phase}(t)$'s, and $\bm{N}(t)$ is a vector concatenating the $\bm{N}_r(t)$'s. Note that we have redefined the notation $\bm{N}(t)$ in this section. We still keep the notation $N_r(t)$, which is the total number of flows on route $r$ at time $t$, i.e., $N_r(t)=\sum_{\phase\in[\Phase_r]}N_{r,\phase}(t)$. A weighted proportionally fair policy allocates rates according to the $N_r(t)$'s. Now the flow count vector $\bm{N}(t)$ is a $\Phase$-dimensional vector with $\Phase=\sum_{r}\Phase_r$, and the flow count process $(\bm{N}(t)\colon t\ge 0)$ is a Markov chain.

Below we give the state transition rates for this Markov chain.  Let $\bm{e}^{(r,\phase)}\in\mathbb{R}^{\Phase}$ be a vector in the state space whose entry that corresponds to phase $\phase$ of route $r$ is equal to $1$ and other entries are equal to~$0$.  Then the transition rate $q_{\bm{n}\bm{n}'}$ from a state $\bm{n}$ to a state $\bm{n}'\neq\bm{n}$ is as follows:
\begin{equation}\label{eq-transition-rates}
q_{\bm{n}\bm{n}'}=
\begin{cases}
\lambda_r\pi_{r,\phase} &\text{if }\bm{n}'=\bm{n}+\bm{e}^{(r,\phase)},\\
n_{r,\phase_1}x_r(\transition_r)_{\phase_1,\phase_2}&\text{if }\bm{n}'=\bm{n}-\bm{e}^{(r,\phase_1)}+\bm{e}^{(r,\phase_2)}, \phase_1\neq\phase_2,\\
n_{r,\phase}x_r\sum
_{\phase'}(-\transition_{r})_{\phase,\phase'}&\text{if }\bm{n}'=\bm{n}-\bm{e}^{(r,\phase)},n_{r,\phase}>0,\\
0&\text{otherwise}.
\end{cases}
\end{equation}

It can be easily verified that the flow count process $(\bm{N}(t)\colon t\ge 0)$ is irreducible and aperiodic. 
We can further show that $(\bm{N}(t)\colon t\ge 0)$ is positive recurrent by establishing drift conditions in Lemma~\ref{LEM-SUM-NRK}.
Therefore, the flow count process $(\bm{N}(t)\colon t\ge 0)$ has a unique stationary distribution.

\paragraph{An $L$-dimensional cone.}
\begin{sloppypar}
We will show that the flow count process $(\bm{N}(t)\colon t\ge 0)$ still collapses to an $L_s$-dimensional cone. This cone, still denoted by $\mathcal{K}$, is also finitely generated by a set of vectors $\{\bm{b}^{(\ell)},\ell\in\satlinks\}\subseteq\{\bm{b}^{(\ell)},\ell\in\mathcal{L}\}$, i.e.,
\begin{equation}\label{eq-cone}
\mathcal{K}=\biggl\{\bm{y}\in\mathbb{R}^{\Phase}\colon \bm{y}=\sum_{\ell\in\satlinks} \alpha_{\ell}\bm{b}^{(\ell)},\alpha_{\ell}\ge 0\text{ for all }\ell\in\satlinks\biggr\},
\end{equation}
where the $\bm{b}^{(\ell)}$'s are now defined in a more complicated manner below based on the loads. For each route $r$, we define the load vector as follows:
\begin{equation}
\bm{\rho}^{(\epsilon)}_r=\lambda_r^{(\epsilon)}(-\transition_r)^{-T}\bm{\pi}_r^T.
\end{equation}
The $\phase$-th entry of this vector, $\rho^{(\epsilon)}_{r,\phase}$, can be thought of as the \emph{load of phase $\phase$ on route $r$}, since the $\phase$-th entry of $(-\transition_r)^{-T}\bm{\pi}_r^T$ is the expected time a flow spends in phase $\phase$ if given a unit of bandwidth. We can verify that the load on route $r$, $\rho_r^{(\epsilon)}$, is the sum of loads of phases on this route, i.e., $\rho_r^{(\epsilon)}=\sum_{\phase\in[\Phase_r]}\rho^{(\epsilon)}_{r,\phase}$. Let $\bm{\rho}^{(\epsilon)}$ be a vector concatenating the $\bm{\rho}_r^{(\epsilon)}$'s, i.e.,
\begin{equation}\label{eq-load-vector}
\bm{\rho}^{(\epsilon)}=[(\bm{\rho}_1^{(\epsilon)})^T,\dots,(\bm{\rho}_R^{(\epsilon)})^T]^T.
\end{equation}
Now we construct a $\Phase\times 1$ vector $\bm{b}^{(\ell)}$ from $\bm{\rho}^{(0)}$ for each link $\ell$: we index the entries of $\bm{b}^{(\ell)}$ using the route and phase $(r,\phase)$, and let
\begin{equation}
b^{(\ell)}_{r,\phase}=\frac{\rho_{r,\phase}^{(0)}\mathbbm{1}_{\{\ell\in r\}}}{\weight_r},
\end{equation}
where $\mathbbm{1}_{\{\ell\in r\}}$ is equal to $1$ when route $r$ uses link $\ell$ and equal to $0$ otherwise, and recall that $\weight_r$'s are the weights in the weighted proportionally fair policy. That is, we keep the entries of $\bm{\rho}^{(0)}$ that correspond to phases of the routes that use link $\ell$, and set other entries to zero. We give a concrete example below to explain this structure of $\bm{b}^{(\ell)}$'s.
\end{sloppypar}
\begin{example}
Consider the network illustrated in Figure~\ref{fig-example}.
\begin{figure}[h]
\vspace{-0.1in}
\centering
\includegraphics[scale=0.25]{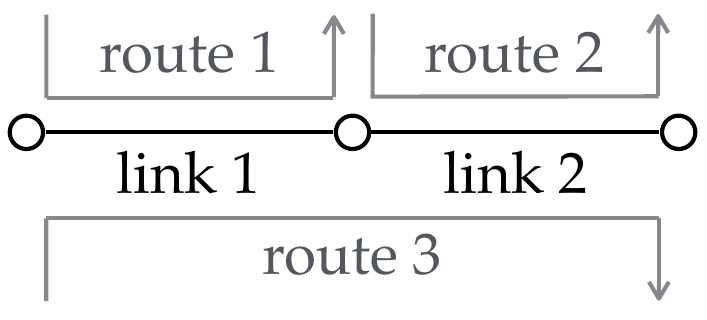}
\caption{Example with two links and three routes.}
\label{fig-example}
\end{figure}
The routing matrix is
$H =
\begin{bmatrix}
1&0&1\\
0&1&1
\end{bmatrix}$.
Suppose the service time distributions of routes $1,2$ and $3$ have $\Phase_1=2, \Phase_2=1$ and $\Phase_3=3$ phases, respectively. Then
\begin{equation*}
\bm{\rho}^{(0)}=
\begin{bmatrix}
\rho^{(0)}_{1,1} & \rho^{(0)}_{1,2} & \rho^{(0)}_{2,1} & \rho^{(0)}_{3,1} & \rho^{(0)}_{3,2} & \rho^{(0)}_{3,3}
\end{bmatrix},
\end{equation*}
and
\begin{equation*}
\bm{b}^{(1)}=\begin{bmatrix}
\frac{\rho^{(0)}_{1,1}}{\weight_1} & \frac{\rho^{(0)}_{1,2}}{\weight_1} & 0 & \frac{\rho^{(0)}_{3,1}}{\weight_3} & \frac{\rho^{(0)}_{3,2}}{\weight_3} & \frac{\rho^{(0)}_{3,3}}{\weight_3}
\end{bmatrix}^T,\quad
\bm{b}^{(2)}=\begin{bmatrix}
0 & 0 & \frac{\rho^{(0)}_{2,1}}{\weight_2} & \frac{\rho^{(0)}_{3,1}}{\weight_3} & \frac{\rho^{(0)}_{3,2}}{\weight_3} & \frac{\rho^{(0)}_{3,3}}{\weight_3}
\end{bmatrix}^T.
\end{equation*}
\end{example}

\paragraph{An inner product.}
The construction of a proper inner product in $\mathbb{R}^{\Phase}$ is a key step in our approach. We consider the following weighted inner product defined by a block-diagonal matrix $M$:
\begin{equation}\label{eq-inner}
\langle \bm{y},\bm{z}\rangle = \bm{y}^TM\bm{z}, \quad \bm{y},\bm{z}\in\mathbb{R}^{\Phase},
\end{equation}
where
\begin{equation}\label{eq-M}
M=
\begin{bmatrix}
M_1 & 0 & \cdots & 0\\
0 & M_2 & \cdots & 0\\
\vdots & \vdots & \ddots & \vdots\\
0 & 0 & \cdots & M_R
\end{bmatrix},
\end{equation}
and each $M_r$ is a $\Phase_r\times\Phase_r$ matrix defined as follows:
\begin{equation}\label{eq-Mr}
M_r=\frac{\weight_r}{\lambda_r^{(0)}}\int_0^{+\infty}\frac{\exp(\transition_r\sigma)\bm{1}_{\Phase_r\times 1}\bm{1}_{\Phase_r\times 1}^T\exp(\transition_r^T\sigma)}{\bm{\pi}_r(-\transition_r)^{-1}\exp(\transition_r\sigma)\bm{1}_{\Phase_r\times 1}}d\sigma.
\end{equation}
Then the induced norm is defined as:
\begin{equation*}
\|\bm{y}\|=\sqrt{\langle \bm{y},\bm{y}\rangle} =\sqrt{ \bm{y}^TM\bm{y}},\quad \bm{y}\in\mathbb{R}^{\Phase}.
\end{equation*}
Again, although this inner product is not the usual dot product, we still keep the notation $\langle \cdot,\cdot\rangle$ and $\|\cdot\|$ for conciseness.  
\iftoggle{complete}{%
We justify the validity of the constructed inner product Appendix~\ref{app-proof-lem-inner-well-defined}.
}{%
We can verify that the constructed inner product is well-defined.  The detailed proof is given in our technical report \cite{WanMagSri_20} due to space limitations.
}%
In the proof, we make an interesting connection to the Popov-Belevitch-Hautus (PBH) test, which is a linear algebraic result well-known to control theorists.

Note that the projection we work with in this section is then defined with respect to this constructed inner product.  For a state $\bm{n}$, which now is a $\Phase$-dimensional vector, we still use $\npara$ to denote its projection onto the cone $\mathcal{K}$ and $\nperp$ to denote its perpendicular component.

If the phase-type distribution only has one phase, i.e., it is an exponential distribution, then this inner product is reduced to the inner product defined in \eqref{eq-inner-exp} for exponential file size distributions. For exponential file size distributions, properties \ref{property-P1} and \ref{property-P2} of the inner product are what the proofs need. For phase-type file size distributions, we will identify two properties of the inner product in Section~\ref{subsec-inner}.  These two properties, indexed as \ref{cond-C1} and \ref{cond-C2}, play an analogous role as \ref{property-P1} and \ref{property-P2} in the proofs.

We remark that the Lyapunov function $\|\bm{n}\|^2$ under this norm, where $\bm{n}$ is a flow count vector, is equivalent to the Lyapunov function in \cite{PagTanFer_12} with a parameter there (denoted by $K$ there) chosen to be $1$. But this particular choice of the parameter is not allowed in \cite{PagTanFer_12}.  To study weighted delay in switches, \cite{LuMagSqu_17} also considers a weighted norm defined by a diagonal matrix.

\subsection{Results.}
The following two theorems, Theorems~\ref{THM-SSC} and \ref{THM-BOUNDS}, generalize the results in Theorems~\ref{THM-SSC-EXP} and \ref{THM-BOUNDS-EXP}, respectively. The proofs are given in Sections~\ref{subsec-ssc} and \ref{subsec-bounds}, respectively.

\begin{theorem}[State-Space Collapse]\label{THM-SSC}
Consider a sequence of bandwidth sharing networks under a weighted proportionally fair policy, indexed by a parameter $\epsilon$ with $0<\epsilon<1$.  The file size distributions belong to a class of phase-type distributions, $\mathcal{D}$.  The arrival rate vector in the $\epsilon$-th system satisfies that $\bm{\lambda}^{(\epsilon)}=(1-\epsilon)\bm{\lambda}^{(0)}$ for some $\bm{\lambda}^{(0)}$ such that a set of $L_s$ links are critically loaded. 
Let $\steady^{(\epsilon)}$ denote a random vector whose distribution is the stationary distribution of the flow count process $(\bm{N}^{(\epsilon)}(t)\colon t\ge 0)$, and $\|\steadyperp^{(\epsilon)}\|$ denote its distance to the $L_s$-dimensional cone $\mathcal{K}$ defined in \eqref{eq-cone} under the constructed inner product. Then in the heavy-traffic regime where $\epsilon\to 0^+$, the $m$-th moment of $\|\steadyperp^{(\epsilon)}\|$ for any nonnegative integer $m$ can be bounded as follows:
\begin{equation*}
\expect\Bigl[\Bigl\|\steadyperp^{(\epsilon)}\Bigr\|^m\Bigr]=O\Biggl(\biggl(\frac{1}{\sqrt\epsilon}\biggr)^m\Biggr).
\end{equation*}
\end{theorem}

\begin{theorem}[Bounds on Flow Counts]\label{THM-BOUNDS}
Consider a sequence of bandwidth sharing networks indexed by a parameter $\epsilon$ with $0<\epsilon<1$.  The file size distributions belong to a class of phase-type distributions, $\mathcal{D}$.  The arrival rate vector in the $\epsilon$-th system satisfies that $\bm{\lambda}^{(\epsilon)}=(1-\epsilon)\bm{\lambda}^{(0)}$ for some $\bm{\lambda}^{(0)}$ such that a set of $L_s$ links are critically loaded. Suppose that a weighted proportionally fair policy with weights $\weight_1,\dots,\weight_R$ is used.  Let $\steady^{(\epsilon)}$ denote a random vector whose distribution is the stationary distribution of the flow count process $(\bm{N}^{(\epsilon)}(t)\colon t\ge 0)$. Then in the heavy-traffic regime where $\epsilon\to 0^+$,
\begin{equation*}
\frac{L_s\cdot \min_r \weight_r}{\epsilon}+o\biggl(\frac{1}{\epsilon}\biggr)\le \expect\left[\sum_r \weight_r \overline{N}_r^{(\epsilon)}\right]
\le\frac{L_s\cdot \max_r \weight_r}{\epsilon}+o\biggl(\frac{1}{\epsilon}\biggr).
\end{equation*}
\end{theorem}
The dominant terms in the upper and lower bounds in Theorem~\ref{THM-BOUNDS}, $L_s\cdot\max_r\weight_r/\epsilon$ and $L_s\cdot\min_r\weight_r/\epsilon$, do not depend on the specific forms of the file size distributions in the class $\mathcal{D}$ except for their means.  Therefore, we say that the upper and lower bounds are \emph{heavy-traffic insensitive}.  

\paragraph{Implication for interchange of limits.}
As a by-product, our state-space collapse result in Theorem~\ref{THM-SSC} provides a key component needed in establishing interchange of limits to justify the diffusion approximation for proportionally fair sharing (\emph{$\weight_r=1$ for all $r$}) derived by \citet{VlaZhaZwa_14}.  Interested readers are referred to \cite{VlaZhaZwa_14} for further details of the diffusion approximation. In this section we restrict our discussions to the heavy-traffic regime where \emph{all} the links are critically loaded since this is the setting studied in \cite{VlaZhaZwa_14}.

In \cite{VlaZhaZwa_14} it is assumed that the file size distributions are in a class of phase-type distributions such that the initial distribution has a positive probability on all the phases.  To take advantage of both the results in \cite{VlaZhaZwa_14} and in our paper, in this section we study the class of phase-type distributions that satisfy the assumptions in both papers, i.e., we study the subclass of $\mathcal{D}$ such that the initial distribution $\pi_{r,\phase}>0$ for all route $r$ and phase $\phase$.  It should be noted that this class is still dense in the space of probability distributions on $[0,\infty)$.

Let $(\hat{\bm{N}}^{(\epsilon)}(t)\colon t\ge 0)$ denote the diffusion-scaled flow-count process in the $\epsilon$-th system, i.e.,
\begin{equation*}
\hat{\bm{N}}^{(\epsilon)}(t)=\epsilon \bm{N}^{(\epsilon)}(t/\epsilon^2).
\end{equation*}
\citet{VlaZhaZwa_14} have shown that $(\hat{\bm{N}}^{(\epsilon)}(t)\colon t\ge 0)$ converges weakly to a process $(\hat{\bm{N}}(t)\colon t\ge 0)$ that is determined by a semimartingale reflected Brownian motion (SRBM) (Theorem~6.1 in \cite{VlaZhaZwa_14}), and have derived the stationary distribution of $(\hat{\bm{N}}(t)\colon t\ge 0)$ (Theorem~7.1 in \cite{VlaZhaZwa_14}).  Let $\hat{\bm{N}}^{(\epsilon)}(\infty)$ and $\hat{\bm{N}}(\infty)$ denote two random variables that follow the stationary distributions of $(\hat{\bm{N}}^{(\epsilon)}(t)\colon t\ge 0)$ and $(\hat{\bm{N}}(t)\colon t\ge 0)$, respectively.  The diagram in Figure~\ref{fig-interchange} summarizes these results.  The missing component in \cite{VlaZhaZwa_14} is a convergence result along the arrow at the bottom of the diagram, which is usually referred to as ``interchange of limits'' since it implies that the two limits as $\epsilon\to 0$ and as $T\to\infty$ can be interchanged. We establish interchange of limits in the following theorem, which formally connects the stationary distributions of the flow-count processes to the stationary distribution of the diffusion limit. The proof is given in Section~\ref{sec-interchange}.
\begin{figure}
\centering
\begin{tikzcd}[column sep=12em, row sep=6em]
\hat{\bm{N}}^{(\epsilon)}(\cdot)|_{[0,T]}
\arrow[r, "\text{\footnotesize $\epsilon\to 0$}" inner sep=1ex, "\text{\footnotesize (Theorem~6.1 in \cite{VlaZhaZwa_14})}"' inner sep=1ex]
\arrow[d, "\text{\footnotesize $T\to\infty$}"' inner sep=1ex, pos=0.25, "\text{\footnotesize (Lemma~\ref{LEM-SUM-NRK} in our paper}"' pos=0.5, "\text{\footnotesize or Theorem~1 in \cite{Mas_07})}"' pos=0.75]
&
\hat{\bm{N}}(\cdot)|_{[0,T]}
\arrow[d, "\text{\footnotesize $T\to\infty$}" inner sep=1ex, pos=0.3, "\text{\footnotesize (Theorem~7.1 in \cite{VlaZhaZwa_14})}" pos=0.6]\\
\hat{\bm{N}}^{(\epsilon)}(\infty)
\arrow[r, dashed, "\text{\footnotesize $\epsilon\to 0$}" inner sep=1ex, "\text{\footnotesize Our result: Theorem~\ref{THM-INTERCHANGE}}"' inner sep=1ex]
&
\hat{\bm{N}}(\infty)
\end{tikzcd}
\caption{Interchange of limits.}
\label{fig-interchange}
\end{figure}

\begin{theorem}[Interchange of Limits]\label{THM-INTERCHANGE}
Under the proportionally fair policy, the stationary distributions of the diffusion-scaled flow-count processes converge weakly to the stationary distribution of the diffusion limit in the heavy-traffic regime as $\epsilon\to 0^+$, i.e.,
\begin{equation*}
\hat{\bm{N}}^{(\epsilon)}(\infty)\Rightarrow\hat{\bm{N}}(\infty),\quad\text{as }\epsilon\to 0^+.
\end{equation*}
\end{theorem}

\subsection{Properties of the inner product.}\label{subsec-inner}
\begin{sloppypar}
In this section we identify two properties of the inner product constructed for phase-type file size distributions.  These two properties allow us to prove Theorems~\ref{THM-SSC} and \ref{THM-BOUNDS} analogously to the proofs for exponential file size distributions.  Similarly as in the setting with exponential file size distributions, these two properties here are also concerned with the difference between the loads and the instantaneous rate allocation.  We again introduce a set of vectors $\{\widehat{\bm{b}}^{(\ell)},\ell\in\mathcal{L}\}$ defined based on the rate allocation as follows. Recall that a state $\bm{n}$ is a $\Phase\times 1$ vector that has the form $\bm{n}=[\bm{n}_1^T,\dots,\bm{n}_R^T]^T$ with $\bm{n}_r=[n_{r,1},\dots,n_{r,\Phase_r}]^T$, and $x_r$ is the bandwidth allocated to each flow on route $r$ based on $\bm{n}$ by proportionally fair sharing. Let $\bm{nx}=[\bm{n}_1^Tx_1,\dots,\bm{n}_R^Tx_R]^T$, whose $(r,\phase)$-th entry, $n_{r,\phase}x_r$, is the total bandwidth allocated to the flows in phase $\phase$ on route $r$. Then $\widehat{\bm{b}}^{(\ell)}$ is defined as
\begin{equation*}
\widehat{b}^{(\ell)}_{r,\phase}=n_{r,\phase}x_r\mathbbm{1}_{\{\ell\in r\}}/\weight_r.
\end{equation*}
Recall that the vector $\bm{b}^{(\ell)}$ is defined based on the loads on the phases as ${b}^{(\ell)}_{r,\phase}=\rho^{(0)}_{r,\phase}\mathbbm{1}_{\{\ell\in r\}}/\weight_r$.
We claim that the constructed inner product satisfies the following two properties, where the norm is the induced norm:
\begin{enumerate}[label=(P\arabic*'),leftmargin=3.8em]
\item \label{cond-C1} For each link $\ell$,
\begin{equation*}
\langle \bm{b}^{(\ell)},(-\transition^T)(\bm{\rho}^{(0)}-\bm{nx})\rangle = U_{\ell}-\delta_{\ell},
\end{equation*}
where $U_{\ell}$ is still the unused bandwidth on link $\ell$.
\item \label{cond-C2} For each link $\ell$,
\begin{equation*}
\langle\bm{b}^{(\ell)}-\widehat{\bm{b}}^{(\ell)},(-\transition^T)(\bm{\rho}^{(0)}-\bm{nx})\rangle\ge\weight_{\min}\eta_{\min}\|\bm{b}^{(\ell)}-\widehat{\bm{b}}^{(\ell)}\|^2,
\end{equation*}
where $\weight_{\min}$ is a positive constant with $\weight_{\min}=\min_r\{\weight_r\}$, $\eta_{\min}$ is a positive constant to be defined later, and
\begin{equation*}
\transition =
\begin{bmatrix}
\transition_1 & 0 & \cdots & 0\\
0 & \transition_2 & \cdots & 0\\
\vdots & \vdots & \ddots & \vdots\\
0 & 0 & \cdots & \transition_R
\end{bmatrix}.
\end{equation*}
\end{enumerate}
\end{sloppypar}

We remark that in both properties, the difference between the load vector $\bm{\rho}^{(0)}$ and the bandwidth allocation vector $\bm{nx}$ is rotated and scaled by $-\transition^T$. Property \ref{cond-C1} requires the projection of this altered difference $(-\transition^T)(\bm{\rho}^{(0)}-\bm{nx})$ onto the vector $\bm{b}^{(\ell)}$ to be the unused bandwidth minus the capacity surplus. For property \ref{cond-C2}, observe that under the regular dot product of Euclidean space, $\langle\bm{b}^{(\ell)}-\widehat{\bm{b}}^{(\ell)},\bm{\rho}^{(0)}-\bm{nx}\rangle_{\text{Euclidean}}=\|\bm{b}^{(\ell)}-\widehat{\bm{b}}^{(\ell)}\|^2_{\text{Euclidean}}$. Property \ref{cond-C2} requires that this Euclidean inner product is not diminished by the matrix $(-\transition^T)$ under the constructed inner product.

Properties \ref{cond-C1} and \ref{cond-C2} are not hard to verify once we identify the two lemmas below. 
\iftoggle{complete}{%
The proofs are given in Appendix~\ref{app-proof-lem-inner-well-defined}.
}{%
The detailed proof is given in our technical report \cite{WanMagSri_20} due to space limitations.
}%
\begin{lemma}\label{LEM-ROTATE}
For each route $r$, $\frac{1}{\weight_r}(\bm{\rho}^{(0)}_r)^TM_r(-\transition_r^T)=\bm{1}_{\Phase_r\times 1}^T$.
\end{lemma}
\begin{lemma}\label{LEM-LARGER-NORM}
For each route $r$, there exists a constant $\eta_r>0$, such that the matrix $\frac{1}{2}M_r(-\transition_r^T)+\frac{1}{2}(-\transition_r)M_r-\eta_rM_r$ is positive semi-definite.
\end{lemma}

\subsection{Proof of Theorem~\ref{THM-SSC} (State-Space Collapse).}\label{subsec-ssc}
In this proof we repeat the arguments in the proof of Theorem~\ref{THM-SSC-EXP} almost verbatim, by virtue of the properties \ref{cond-C1} and \ref{cond-C2} of the constructed inner product. Again, we fix an $\epsilon$ and omit the superscript~$^{(\epsilon)}$ when it is clear from context for conciseness.  As with the proof of Theorem~\ref{THM-SSC-EXP}, the main part of the proof of Theorem~\ref{THM-SSC} is to show proper bounds on the drift of $\|\nperp\|$. Below Lemma~\ref{LEM-SUM-NRK} first gives drift bounds for $\|\bm{n}\|$ and obtain tail bounds.
\iftoggle{complete}{%
The proof of Lemma~\ref{LEM-SUM-NRK} is given in Appendix~\ref{app-proof-lem-sum-nrk}.
}{%
The proof of Lemma~\ref{LEM-SUM-NRK} is given in our technical report \cite{WanMagSri_20} due to space limitations.
}
Then Lemma~\ref{lem-nperp-drift} establishes drift bounds for $\|\nperp\|$ under phase-type file size distributions, which is an analogy to Lemma~\ref{lem-nperp-drift-exp} for exponential file size distributions.

\begin{lemma}[Drift Bounds for $\|\bm{n}\|$ and Tail Bounds]\label{LEM-SUM-NRK}
For any nonnegative $\epsilon$ with $\epsilon<1$, the drift of the Lyapunov function $\|\bm{n}\|$ satisfies that
\begin{equation*}
\Delta\|\bm{n}\|\le -\frac{\eta_{\min}A_3\epsilon}{\mu_{\max}},\quad\text{when }\|\bm{n}\|\ge\frac{\mu_{\max}A_1}{\eta_{\min}A_3\epsilon},
\end{equation*}
and
\begin{equation*}
\Delta\|\bm{n}\|\le \nu\zeta,\quad\text{for all }\bm{n},
\end{equation*}
where $\mu_{\max},A_3,\nu$ and $\zeta$ are positive constants.
Consequently, the flow count process $(\bm{N}(t)\colon t\ge 0)$ is positive recurrent, and the distribution of $\sum_{r,\phase} \weight_r\overline{N}_{r,\phase}$ has the following exponential tail bound: for any nonnegative integer $j$,
\begin{align*}
\Pr\biggl(\sum_{r,\phase}\weight_r\overline{N}_{r,\phase}>\frac{\mu_{\max}A_1A_4}{\eta_{\min}A_3\epsilon }+2\nu A_4j\biggr)\le\beta^{j+1},
\end{align*}
where $A_4$ is a positive constant, and $\beta=\frac{\zeta \nu }{\zeta \nu +\epsilon\eta_{\min}A_3/\mu_{\max}}<1$.
\end{lemma}

\begin{lemma}[Drift Bounds for $\|\nperp\|$]\label{lem-nperp-drift}
In the $\epsilon$-th system, the drift of the Lyapunov function $\|\nperp\|$ satisfies that
\begin{equation}
\Delta^{(\epsilon)} \|\nperp\|\le -\sqrt{\epsilon}
\end{equation}
when
\begin{equation}\label{eq-conditions-drift}
\epsilon\le \epsilon_{\max},\quad\|\nperp\|\ge\frac{A_1}{\xi_1\sqrt{\epsilon}},\quad\frac{\|\nperp\|}{\sum_r\weight_r n_r}\ge\frac{\xi_2\sqrt{\epsilon}}{A_2},
\end{equation}
and
\begin{equation}
\Delta^{(\epsilon)} \|\nperp\|\le(\xi_1+1)\sqrt{\epsilon}
\end{equation}
when
\begin{equation}
\epsilon\le \epsilon_{\max},\quad\|\nperp\|\ge\frac{A_1}{\xi_1\sqrt{\epsilon}},
\end{equation}
where $\epsilon_{\max},\xi_1,\xi_2,A_1,A_2$ are positive constants.
\end{lemma}

\proof{Proof of Lemma~\ref{lem-nperp-drift} (Drift Bounds for $\|\nperp\|$).}
\begin{sloppypar}
This proof follows arguments similar to those in the proof of Lemma~\ref{lem-nperp-drift-exp}. Basically we need to replace the vector $\bm{\lambda}^{(0)}-\bm{nx\mu}$ in the proof of Lemma~\ref{lem-nperp-drift-exp} with its generalized form $(-\transition^T)(\bm{\rho}^{(0)}-\bm{nx})$, and use the properties \ref{cond-C1} and \ref{cond-C2} in place of the properties \ref{property-P1} and \ref{property-P2}.
\end{sloppypar}

We also start by stating the following claim, the proof of which is given at the end of this proof.
\begin{claim}\label{claim-nperp-drift}
The drift, $\Delta \|\nperp\|$, is upper bounded as follows:
\begin{equation*}
\Delta\|\nperp\|\le\frac{1}{\|\nperp\|}\langle \bm{n}-\npara,(-\transition^T)(\bm{\rho}^{(0)}-\bm{nx})\rangle+\epsilon\|(-\transition^T)\bm{\rho}^{(0)}\|+\frac{A_1}{\|\nperp\|},
\end{equation*}
where $A_1$ is a constant.
\end{claim}

\begin{sloppypar}
Then we analyze the terms in Claim~\ref{claim-nperp-drift}.
First consider the $\langle \bm{n},(-\transition^T)(\bm{\rho}^{(0)}-\bm{nx})\rangle$ in the first term. Recall that $n_{r,\phase}$ can be written in the following form according to the weighted proportionally fair policy: $n_{r,\phase}=\frac{n_{r,\phase}x_r}{\weight_r}\sum_{\ell:\ell\in r} p_{\ell}$,
where $\weight_r$ is the weight used and $p_{\ell}$ is the Lagrange multiplier of the capacity constraint of link $\ell$. We can further write this equality in a vector form using the vectors $\{\widehat{\bm{b}}^{(\ell)},\ell\in\mathcal{L}\}$: $\bm{n}=\sum_{\ell} p_{\ell}\widehat{\bm{b}}^{(\ell)}.$
Note that by property \ref{cond-C1} and complementary slackness, $p_{\ell}\langle \bm{b}^{(\ell)},(-\transition^T)(\bm{\rho}^{(0)}-\bm{nx})\rangle = p_{\ell}U_{\ell}-p_{\ell}\delta_{\ell}=-p_{\ell}\delta_{\ell}$,
where recall that $\delta_{\ell}=0$ for $\ell\in\satlinks$. Thus
\begin{align*}
&\mspace{23mu}\langle \bm{n}, (-\transition^T)(\bm{\rho}^{(0)}-\bm{nx})\rangle\\
& = \sum_{\ell} p_{\ell}\langle \widehat{\bm{b}}^{(\ell)}-\bm{b}^{(\ell)},(-\transition^T)(\bm{\rho}^{(0)}-\bm{nx})\rangle+\sum_{\ell\in\mathcal{L}\setminus\satlinks}p_{\ell}\langle\bm{b}^{(\ell)},(-\transition^T)(\bm{\rho}^{(0)}-\bm{nx})\rangle\nonumber\\
&\le-\weight_{\min}\eta_{\min}\sum_{\ell}p_{\ell}\|\widehat{\bm{b}}^{(\ell)}-\bm{b}^{(\ell)}\|^2-\sum_{\ell\in\mathcal{L}\setminus\satlinks}p_{\ell}\delta_{\ell},
\end{align*}
where the inequality follows from \ref{cond-C2}. Then following the same arguments as those in the proof of Lemma~\ref{lem-nperp-drift-exp} (\eqref{eq-perp-bound-1-exp}--\eqref{eq-perp-bound-3-exp}), we have
$\|\nperp\|^2
\le\frac{2}{C_{\min}}\biggl(\sum_{r,\phase}\weight_rn_{r,\phase}\biggr)\biggl(\sum_{\ell} p_{\ell}\|\widehat{\bm{b}}^{(\ell)}-\bm{b}^{(\ell)}\|^2+\sum_{\ell\in\mathcal{L}\setminus\satlinks}p_{\ell}\|\bm{b}^{(\ell)}\|^2\biggr).$
Let $A_2=\min\biggl\{\weight_{\min}\eta_{\min},\min_{\ell\in\mathcal{L}\setminus\satlinks}\frac{\delta_{\ell}}{\|\bm{b}^{(\ell)}\|^2}\biggr\}\frac{C_{\min}}{2}$.
Then $A_2$ is a positive constant independent of $\epsilon$ and there holds
\begin{equation}\label{eq-bound-npart}
\langle \bm{n}, (-\transition^T)(\bm{\rho}^{(0)}-\bm{nx})\rangle\le-A_2\frac{\|\nperp\|^2}{\sum_{r,\phase}\weight_r n_{r,\phase}}.
\end{equation}
\end{sloppypar}

We can prove that $\langle\npara,(-\transition^T)(\bm{\rho}^{(0)}-\bm{nx})\rangle$ is nonnegative following similar arguments as those in the proof of Lemma~\ref{lem-nperp-drift-exp}. Then combining the bounds for $\langle \bm{n},(-\transition^T)(\bm{\rho}^{(0)}-\bm{nx})\rangle$ and $\langle\npara,(-\transition^T)(\bm{\rho}^{(0)}-\bm{nx})\rangle$ yields:
\begin{align*}
\Delta \|\nperp\|&\le -A_2\frac{\|\nperp\|}{\sum_{r,\phase}\weight_r n_{r,\phase}} +\epsilon\|(-\transition^T)\bm{\rho}^{(0)}\|+\frac{A_1}{\|\nperp\|}.
\end{align*}
We choose any constants $\xi_1>0,\xi_2>0$ such that $\xi_2-\xi_1=2$.
Then when
\begin{equation*}
\epsilon\le\epsilon_{\max}\triangleq \frac{1}{\|(-\transition^T)\bm{\rho}^{(0)}\|^2},\quad\|\nperp\|\ge\frac{A_1}{\xi_1\sqrt{\epsilon}},\quad
\frac{\|\nperp\|}{\sum_{r,\phase}\weight_r n_{r,\phase}}\ge\frac{\xi_2\sqrt{\epsilon}}{A_2},
\end{equation*}
we have $\Delta \|\nperp\|\le -\xi_2\sqrt{\epsilon}+\sqrt{\epsilon}+\xi_1\sqrt{\epsilon}=-\sqrt{\epsilon}$; and when $\epsilon\le\epsilon_{\max},\|\nperp\|\ge\frac{A_1}{\xi_1\sqrt{\epsilon}},$
we have $\Delta \|\nperp\|\le \sqrt{\epsilon}+\xi_1\sqrt{\epsilon}=(\xi_1+1)\sqrt{\epsilon}$, which are the drift bounds in Lemma~\ref{lem-nperp-drift}.

Lastly, we prove the Claim~\ref{claim-nperp-drift} at the beginning of this proof. 
We first bound $\Delta\|\nperp\|$ in the form below following arguments similar to those in the proof of Claim~\ref{claim-nperp-drift-exp}: $\Delta \|\nperp\|\le\frac{1}{2\|\nperp\|}(\Delta \|\bm{n}\|^2-\Delta\|\npara\|^2)$.
We then analyze the drifts $\Delta\|\bm{n}\|^2$ and $\Delta\|\npara\|^2$. First,
\begin{align}
\Delta \|\bm{n}\|^2&= \sum_r\Biggl(\sum_{\phase\in[\Phase_r]}\lambda_r\pi_{r,\phase}\Bigl(\|\bm{n}+\bm{e}^{(r,\phase)}\|^2-\|\bm{n}\|^2\Bigr)\label{eq-bound-n2-drift-1}\\
&\mspace{23mu}+\sum_{\substack{\phase_1,\phase_2\in[\Phase_r]\\\phase_1\neq\phase_2}}n_{r,\phase_1}x_r(\transition_{r})_{\phase_1,\phase_2}\Bigl(\|\bm{n}-\bm{e}^{(r,\phase_1)}+\bm{e}^{(r,\phase_2)}\|^2-\|\bm{n}\|^2\Bigr)\nonumber\\
&\mspace{23mu}+\sum_{\phase\in[\Phase_r]}n_{r,\phase}x_r\sum_{\phase'}(-\transition_r)_{\phase,\phase'}\Bigl(\|\bm{n}-\bm{e}^{(r,\phase)}\|^2-\|\bm{n}\|^2\Bigr)\Biggr)\nonumber\\
&\le 2\langle \bm{n},
\begin{bmatrix}
\lambda_1\bm{\pi}_1^T\\
\lambda_2\bm{\pi}_2^T\\
\vdots\\
\lambda_R\bm{\pi}_R^T
\end{bmatrix}
-
\begin{bmatrix}
-\transition_1^T & 0 & \cdots & 0\\
0 & -\transition_2^T & \cdots & 0\\
\vdots & \vdots & \ddots & \vdots\\
0 & 0 & \cdots & -\transition_R^T
\end{bmatrix}
\begin{bmatrix}
\bm{n}_1x_1\\
\bm{n}_2x_2\\
\vdots\\
\bm{n}_Rx_R
\end{bmatrix}\rangle + 2A_1,\label{eq-bound-n2-drift-2}\\
&=2\langle \bm{n}, (-\transition^T)(\bm{\rho}-\bm{nx})\rangle+2A_1.\label{eq-bound-n2-drift-3}
\end{align}
where \eqref{eq-bound-n2-drift-1} follows from the transition rates of the flow count process $(\bm{N}(t)\colon t\ge 0)$ under the weighted proportionally fair sharing, \eqref{eq-bound-n2-drift-2} is obtained by expressing norms in terms of inner products and bounding the sum of the terms with $\|\bm{e}^{(r,\phase)}\|^2$'s by $2A_1$ for a constant $A_1$, and \eqref{eq-bound-n2-drift-3} follows from the definition of $\bm{\rho}$ in \eqref{eq-load-vector}.  We can derive a lower bound on $\Delta\|\npara\|^2$ in a similar way:
\begin{align}
\Delta \|\npara\|^2&= \sum_r\Biggl(\sum_{\phase\in[\Phase_r]}\lambda_r\pi_{r,\phase}\Bigl(\|(\bm{n}+\bm{e}^{(r,\phase)})_\shortparallel\|^2-\|\npara\|^2\Bigr)\label{eq-bound-npara2-drift-1}\\
&\mspace{66mu}+\mspace{-9mu}\sum_{\substack{\phase_1,\phase_2\in[\Phase_r]\\\phase_1\neq\phase_2}}\mspace{-12mu}n_{r,\phase_1}x_r(\transition_{r})_{\phase_1,\phase_2}\Bigl(\|(\bm{n}-\bm{e}^{(r,\phase_1)}+\bm{e}^{(r,\phase_2)})_{\shortparallel}\|^2-\|\npara\|^2\Bigr)\nonumber\\
&\mspace{66mu}+\sum_{\phase\in[\Phase_r]}n_{r,\phase}x_r\sum_{\phase'}(-\transition_r)_{\phase,\phase'}\Bigl(\|(\bm{n}-\bm{e}^{(r,\phase)})_{\shortparallel}\|^2-\|\npara\|^2\Bigr)\Biggr)\nonumber\\
&\ge\sum_r\Biggl(\sum_{\phase\in[\Phase_r]}\lambda_r\pi_{r,\phase}\Bigl(2\langle \npara,\bm{e}^{(r,\phase)}\rangle \Bigr)
+\sum_{\substack{\phase_1,\phase_2\in[\Phase_r]\\\phase_1\neq\phase_2}}\mspace{-12mu}n_{r,\phase_1}x_r(\transition_{r})_{\phase_1,\phase_2}\Bigl(2\langle \npara,-\bm{e}^{(r,\phase_1)}+\bm{e}^{(r,\phase_2)}\rangle\Bigr)\label{eq-bound-npara2-drift-2}\\
&\mspace{66mu}+\sum_{\phase\in[\Phase_r]}n_{r,\phase}x_r\sum_{\phase'}(-\transition_r)_{\phase,\phase'}\Bigl(2\langle \npara,-\bm{e}^{(r,\phase)}\rangle\Bigr)\Biggr)\nonumber\\
&=2\langle \npara, (-\transition^T)(\bm{\rho}-\bm{nx})\rangle\label{eq-bound-npara2-drift-3},
\end{align}
where \eqref{eq-bound-npara2-drift-1} still follows from the transitions rates of the flow count process, \eqref{eq-bound-npara2-drift-2} follows from that $\langle \npara,\nperp\rangle = 0$ and $\langle \npara,(\bm{n}+\bm{e}^{(r,\phase)})_{\perp}\rangle\le 0$, $\langle \npara,(\bm{n}-\bm{e}^{(r,\phase_1)}+\bm{e}^{(r,\phase_2)})_{\perp}\rangle\le 0$, $\langle \npara,(\bm{n}-\bm{e}^{(r,\phase)})_{\perp}\rangle\le 0$ since perpendicular components are in the polar cone of the cone $\mathcal{K}$, and \eqref{eq-bound-npara2-drift-3} still follows from the definition of $\bm{\rho}$. Combining the above bounds yields
\begin{align*}
\Delta\|\nperp\|
&\le\frac{1}{\|\nperp\|}\langle \bm{n}-\npara,(-\transition^T)(\bm{\rho}-\bm{nx})\rangle+\frac{A_1}{\|\nperp\|}\\
&=\frac{1}{\|\nperp\|}\langle \bm{n}-\npara,(-\transition^T)(\bm{\rho}^{(0)}-\bm{nx})\rangle-\epsilon\frac{\langle \nperp,(-\transition^T)\bm{\rho}^{(0)}\rangle}{\|\nperp\|}+\frac{A_1}{\|\nperp\|}\\
&\le\frac{1}{\|\nperp\|}\langle \bm{n}-\npara,(-\transition^T)(\bm{\rho}^{(0)}-\bm{nx})\rangle+\epsilon\|(-\transition^T)\bm{\rho}^{(0)}\|+\frac{A_1}{\|\nperp\|},
\end{align*}
which completes the proof of the claim.

Now we have established proper drift bounds for $\|\nperp\|$ in Lemma~\ref{lem-nperp-drift}.
We can then obtain the moment bounds on $\|\steadyperp\|$ in Theorem~\ref{THM-SSC} following arguments similar to those in the proof of Lemma~\ref{lem-distr-bound-exp}, noting a tail bound on $\sum_{r,\phase}\weight_r\overline{N}_{r,\phase}$ in Lemma~\ref{LEM-SUM-NRK}, which is similar to the bound in Lemma~\ref{LEM-SUM-NR-EXP}.
The proof of Theorem~\ref{THM-SSC} is thus completed.
\Halmos
\endproof

\subsection{Proof of Theorem~\ref{THM-BOUNDS} (Bounds on Flow Counts).}\label{subsec-bounds}
As with the proof of Theorem~\ref{THM-BOUNDS-EXP}, in this proof we also obtain the bounds by setting the steady-state drift of the Lyapunov function $V(\bm{n})=\|\nparaspace\|$ to $0$, where $\nparaspace$ is the projection of the state $\bm{n}$ onto the \emph{subspace} where the cone $\mathcal{K}$ lies in, i.e., the subspace spanned by $\bm{b}^{(\ell)}$'s, denoted by $\mathcal{S}$. Again we use the superscript~$^s$ to indicate when the projection is onto the subspace.

Again we fix an $\epsilon>0$ and temporarily omit the superscript $^{(\epsilon)}$ for conciseness. Following steps similar to those in \eqref{eq-bound-npara2-drift-1}, we write the drift of $\|\nparaspace\|$ as follows:
\begin{align}
\Delta \|\nparaspace\|^2&= \sum_r\Biggl(\sum_{\phase\in[\Phase_r]}\lambda_r\pi_{r,\phase}\Bigl(\|(\bm{n}+\bm{e}^{(r,\phase)})_\shortparallel^s\|^2-\|\nparaspace\|^2\Bigr)\nonumber\\
&\mspace{23mu}+\mspace{-9mu}\sum_{\substack{\phase_1,\phase_2\in[\Phase_r]\\\phase_1\neq\phase_2}}\mspace{-12mu}n_{r,\phase_1}x_r(\transition_{r})_{\phase_1,\phase_2}\Bigl(\|(\bm{n}-\bm{e}^{(r,\phase_1)}+\bm{e}^{(r,\phase_2)})_{\shortparallel}^s\|^2-\|\nparaspace\|^2\Bigr)\nonumber\\
&\mspace{23mu}+\sum_{\phase\in[\Phase_r]}n_{r,\phase}x_r\sum_{\phase'}(-\transition_r)_{\phase,\phase'}\Bigl(\|(\bm{n}-\bm{e}^{(r,\phase)})_{\shortparallel}^s\|^2-\|\nparaspace\|^2\Bigr)\Biggr)\nonumber\\
&= \sum_r\Biggl(\sum_{\phase\in[\Phase_r]}\lambda_r\pi_{r,\phase}\Bigl(2\langle \nparaspace,(\bm{e}^{(r,\phase)})_\shortparallel^s\rangle +\|(\bm{e}^{(r,\phase)})_\shortparallel^s\|^2\Bigr)\nonumber\\
&\mspace{23mu}+\mspace{-12mu}\sum_{\substack{\phase_1,\phase_2\in[\Phase_r]\\\phase_1\neq\phase_2}}n_{r,\phase_1}x_r(\transition_{r})_{\phase_1,\phase_2}\Bigl(2\langle \nparaspace,-(\bm{e}^{(r,\phase_1)})_{\shortparallel}^s+(\bm{e}^{(r,\phase_2)})_{\shortparallel}^s\rangle+\|-(\bm{e}^{(r,\phase_1)})_{\shortparallel}^s+(\bm{e}^{(r,\phase_2)})_{\shortparallel}^s\|^2\Bigr)\nonumber\\
&\mspace{23mu}+\sum_{\phase\in[\Phase_r]}n_{r,\phase}x_r\biggl(-\sum_{\phase'}(\transition_r)_{\phase,\phase'}\biggr)\Bigl(-2\langle \nparaspace,(\bm{e}^{(r,\phase)})_{\shortparallel}^s\rangle+
\|(\bm{e}^{(r,\phase)})_{\shortparallel}^s\|^2\Bigr)\Biggr)\nonumber\\
&=2\langle\nparaspace,(-\transition^T)(\bm{\rho}-\bm{nx})\rangle + B_1(\bm{n})\label{eq-bound-nparaspace2-drift}\\
&=-2\epsilon\langle \nparaspace,(-\transition^T)\bm{\rho}^{(0)}\rangle +2\langle\nparaspace,(-\transition^T)(\bm{\rho}^{(0)}-\bm{nx})\rangle + B_1(\bm{n}),\nonumber
\end{align}
where in \eqref{eq-bound-nparaspace2-drift} we have used the fact that $\langle \nparaspace,(\bm{e}^{(r,\phase)})_{\shortparallel}^s\rangle = \langle \nparaspace,\bm{e}^{(r,\phase)}\rangle$ since $\langle \nparaspace,(\bm{e}^{(r,\phase)})_{\perp}^s\rangle =0$ and the definition of $\bm{\rho}$, and $B_1(\cdot)$ is defined as
\begin{align*}
B_1(\bm{n})&=\sum_r\Biggl(\sum_{\phase\in[\Phase_r]}\lambda_r\pi_{r,\phase}\|(\bm{e}^{(r,\phase)})_\shortparallel^s\|^2+\mspace{-12mu}\sum_{\substack{\phase_1,\phase_2\in[\Phase_r]\\\phase_1\neq\phase_2}}\mspace{-12mu}n_{r,\phase_1}x_r(\transition_{r})_{\phase_1,\phase_2}\|-(\bm{e}^{(r,\phase_1)})_{\shortparallel}^s+(\bm{e}^{(r,\phase_2)})_{\shortparallel}^s\|^2\\
&\mspace{66mu}+\sum_{\phase\in[\Phase_r]}n_{r,\phase}x_r\biggl(-\sum_{\phase'}(\transition_r)_{\phase,\phase'}\biggr)
\|(\bm{e}^{(r,\phase)})_{\shortparallel}^s\|^2\Biggr).
\end{align*}
When the system is in steady state, we have $\expect[\Delta \|\steadyparaspace\|^2]=0$. Therefore,
\begin{equation}\label{eq-Lyapunov=0}
\epsilon\expect[\langle \steadyparaspace,(-\transition^T)\bm{\rho}^{(0)}]=\expect[\langle \steadyparaspace,(-\transition^T)(\bm{\rho}^{(0)}-\bm{\steady x})]+\frac{1}{2}\expect[B_1(\steady)].
\end{equation}

The remainder of this proof analyzes the three terms in \eqref{eq-Lyapunov=0} term by term. To facilitate the analysis, as with the proof of Theorem~\ref{THM-BOUNDS-EXP}, we first need to bound the difference $|\alpha_{\ell}^s-p_{\ell}|$, where recall that $p_{\ell}$ is the Lagrange multiplier for the capacity constraint of link $\ell$. Lemma~\ref{lem-continuity2} below is an analogy to Lemma~\ref{lem-continuity2-exp}. The proof of Lemma~\ref{lem-continuity2} is given at the end of this proof.

\begin{lemma}\label{lem-continuity2}
There exist positive constants $B_2$ and $B_3$ such that for any state $\bm{n}$,
\begin{equation}
|\alpha_{\ell}^s-p_{\ell}|\le B_2\|\nperp\|^{1/2}\Biggl(\sum_{r,\phase}\weight_r n_{r,\phase}\Biggr)^{1/2},\forall \ell \in\satlinks,\text{ and }\sum_{\ell\in\mathcal{L}\setminus\satlinks}p_{\ell}\le B_3\|\nperp\|,
\end{equation}
where the $\alpha_{\ell}^s$'s are the coefficients in the projection $\nparaspace=\sum_{\ell\in\satlinks}\alpha_{\ell}\bm{b}^{(\ell)}$ and the $p_{\ell}$'s are the Lagrange multipliers for the capacity constraints.
\end{lemma}

Now we analyzes the three terms in \eqref{eq-Lyapunov=0}.
\begin{enumerate}[leftmargin=0em,itemindent=3.2em,label=({\roman*})]
\item We first consider the term $\epsilon\expect[\langle \steadyparaspace,(-\transition^T)\bm{\rho}^{(0)}\rangle]$ and show that it is close to $\epsilon\expect[\sum_{r,\phase}\weight_r\overline{N}_{r,\phase}]$. Recall that since $\steadyparaspace$ is in the subspace $\mathcal{S}$, it can be written as $\steadyparaspace=\sum_{\ell\in\satlinks}\alpha_{\ell}^s\bm{b}^{(\ell)}$, where the coefficients $\alpha_{\ell}^s$'s can be negative.
Thus
\begin{align}
\langle \steadyparaspace,(-\transition^T)\bm{\rho}^{(0)}\rangle =\sum_{\ell\in\satlinks}\alpha_{\ell}^s\langle \bm{b}^{(\ell)},(-\transition^T)\bm{\rho}^{(0)}\rangle=\sum_{\ell\in\satlinks} \alpha_{\ell}^s C_{\ell},\label{eq-nparaspace2-term1}
\end{align}
\iftoggle{complete}{%
where \eqref{eq-nparaspace2-term1} follows from arguments similar to those in the proof of property \ref{cond-C1} in \eqref{eq-condition-C1-proof-1}--\eqref{eq-condition-C1-proof-5} noting that $\delta_{\ell}=0$ for $\ell\in\satlinks$.
}{%
where \eqref{eq-nparaspace2-term1} follows from arguments similar to those in the proof of property \ref{cond-C1} noting that $\delta_{\ell}=0$ for $\ell\in\satlinks$.
}
We also know that $\sum_{r,\phase}\weight_r \overline{N}_{r,\phase}=\sum_{\ell} p_{\ell}C_{\ell}$. Let $\widetilde{C}=\sum_{\ell} C_{\ell}$. Then
\begin{align*}
\biggl|\sum_{r,\phase}\weight_r \overline{N}_{r,\phase}-\langle \steadyparaspace,(-\transition^T)\bm{\rho}^{(0)}\rangle\biggr|
&\le\sum_{\ell\in\satlinks}|\alpha_{\ell}^s-p_{\ell}|C_{\ell}+\sum_{\ell\in\mathcal{L}\setminus\satlinks}p_{\ell}C_{\ell}\\
&
\le \widetilde{C}B_2\|\steadyperp\|^{1/2}\Biggl(\sum_{r,\phase}\weight_r \overline{N}_{r,\phase}\Biggr)^{1/2}+\widetilde{C}B_3\|\steadyperp\|,
\end{align*}
where the second inequality follows from Lemma~\ref{lem-continuity2}.  Therefore, by Cauchy-Schwarz inequality,
\begin{align*}
\epsilon\expect\Biggl[\biggl|\sum_{r,\phase}\weight_r \overline{N}_{r,\phase}-\langle \steadyparaspace,(-\transition^T)\bm{\rho}^{(0)}\rangle\biggr|\Biggr]
&\le\epsilon \widetilde{C}B_2\Bigl(\expect[\|\steadyperp\|]\Bigr)^{1/2}\Biggl(\expect\Biggl[\sum_{r,\phase}\weight_r\overline{N}_{r,\phase}\Biggr]\Biggr)^{1/2}+\epsilon\widetilde{C} B_3\expect[\|\steadyperp\|]\\
&=O(\epsilon^{1/4}),
\end{align*}
where the equality follows from the state-space collapse result in Theorem~\ref{THM-SSC} and the bound on $\expect[\sum_{r,\phase}\weight_r\overline{N}_{r,\phase}]$ indicated by Lemma~\ref{LEM-SUM-NRK}.

\item Next, we bound the term $\expect[\langle \steadyparaspace,(-\transition^T)(\bm{\rho}^{(0)}-\steady\bm{x})]$. Again, since $\steadyparaspace\in\mathcal{S}$ and recall the property \ref{cond-C1} for the inner product, we have
\begin{align*}
\langle \steadyparaspace,(-\transition^T)(\bm{\rho}^{(0)}-\steady\bm{x})\rangle=\sum_{\ell\in\satlinks} \alpha_{\ell}^s\langle \bm{b}^{(\ell)},(-\transition^T)(\bm{\rho}^{(0)}-\steady\bm{x})\rangle=\sum_{\ell\in\satlinks} \alpha_{\ell}^sU_{\ell}.
\end{align*}
Consider an $\ell\in\satlinks$. By Lemma~\ref{lem-continuity2} and H\"{o}lder's inequality:
\begin{align*}
\expect[|\alpha_{\ell}^sU_{\ell}|]=\expect[|(\alpha_{\ell}^s-p_{\ell})U_{\ell}|]\le B_2\Biggl(\expect\Biggl[\|\steadyperp\|^{\frac{\tau_1}{2}}\biggl(\sum_{r,\phase}\weight_r \overline{N}_{r,\phase}\biggr)^{\frac{\tau_1}{2}}\Biggr]\Biggr)^{\frac{1}{\tau_1}}\Bigl(\expect[U_{\ell}^{\tau_2}]\Bigr)^{\frac{1}{\tau_2}},
\end{align*}
where we pick $\tau_1$ and $\tau_2$ such that $\tau_1$ is an even integer with $\tau_1>4$ and $\frac{1}{\tau_1}+\frac{1}{\tau_2}=1$. Using Cauchy-Schwarz inequality we have
\begin{align*}
\Biggl(\expect\Biggl[\|\steadyperp\|^{\frac{\tau_1}{2}}\biggl(\sum_{r,\phase}\weight_r \overline{N}_{r,\phase}\biggr)^{\frac{\tau_1}{2}}\Biggr]\Biggr)^{\frac{1}{\tau_1}}&\le(\expect[\|\steadyperp\|^{\tau_1}])^{\frac{1}{2\tau_1}}
\Biggl(\expect\Biggl[\biggl(\sum_{r,\phase}\weight_r \overline{N}_r\biggr)^{\tau_1}\Biggr]\Biggr)^{\frac{1}{2\tau_1}}=O(\epsilon^{-\frac{3}{4}}),
\end{align*}
where again the last equality follows from the state-space collapse result in Theorem~\ref{THM-SSC} and the bound on $\expect[\sum_{r,\phase}\weight_r\overline{N}_{r,\phase}]$ is indicated by Lemma~\ref{LEM-SUM-NRK}. Next we bound $\expect[U_{\ell}^{\tau_2}]$. We can prove that $\expect[U_{\ell}]=\epsilon C_{\ell}$ by considering the Lyapunov function $w_{\ell}(\bm{n})=\langle \bm{b}^{(\ell)},\bm{n}\rangle$. Its drift is
\begin{align*}
\Delta w_{\ell}(\bm{n})&=-\epsilon\langle \bm{b}^{(\ell)},(-\transition^T)\bm{\rho}^{(0)}\rangle+\langle \bm{b}^{(\ell)},(-\transition^T)(\bm{\rho}^{(0)}-\bm{nx})\rangle\\
&=-\epsilon C_{\ell}+U_{\ell}.
\end{align*}
Since in the steady state $\expect[\Delta w_{\ell}(\steady)]=0$, we have $\expect[U_{\ell}]=\epsilon C_{\ell}$.
Since $0\le U_{\ell}\le C_{\ell}$, there holds
\begin{align*}
\Bigl(\expect\Bigl[U_{\ell}^{\tau_2}\Bigr]\Bigr)^{\frac{1}{\tau_2}}&\le\Bigl(\expect\Bigl[U_{\ell}\cdot C_{\ell}^{\tau_2-1}\Bigr]\Bigr)^{\frac{1}{\tau_2}}=\epsilon^{\frac{1}{\tau_2}} C_{\ell}.
\end{align*}
Combining these bounds we have $\expect[|\alpha_{\ell}^sU_{\ell}|]=O(\epsilon^{\frac{1}{4}-\frac{1}{\tau_1}})$ for each $\ell\in\satlinks$, and thus
\begin{align*}
\expect[|\langle \steadyparaspace,(-\transition^T)(\bm{\rho}^{(0)}-\steady\bm{x})|]&=\expect\Biggl[\biggl|\sum_{\ell}\alpha_{\ell}^sU_{\ell}\biggr|\Biggr]=O(\epsilon^{\frac{1}{4}-\frac{1}{\tau_1}}).
\end{align*}

\item Lastly, we bound the last term $\expect[B_1(\steady)]/2$. Recall that
\begin{align*}
B_1(\steady)&=\sum_r\Biggl(\underbrace{\sum_{\phase\in[\Phase_r]}\lambda_r\pi_{r,\phase}\|(\bm{e}^{(r,\phase)})_\shortparallel^s\|^2}_{\mathcal{T}_{r,1}}+\underbrace{\sum_{\substack{\phase_1,\phase_2\in[\Phase_r]\\\phase_1\neq\phase_2}}\overline{N}_{r,\phase_1}x_r(\transition_{r})_{\phase_1,\phase_2}\|-(\bm{e}^{(r,\phase_1)})_{\shortparallel}^s+(\bm{e}^{(r,\phase_2)})_{\shortparallel}^s\|^2}_{\mathcal{T}_{r,2}}\\
&\mspace{54mu}+\underbrace{\sum_{\phase\in[\Phase_r]}\overline{N}_{r,\phase}x_r\biggl(-\sum_{\phase'}(\transition_r)_{\phase,\phase'}\biggr)\|(\bm{e}^{(r,\phase)})_{\shortparallel}^s\|^2}_{\mathcal{T}_{r,3}}\Biggr).
\end{align*}
Note that
\begin{align}
\mathcal{T}_{r,1}&\le (\max_r \weight_r) \sum_{\phase\in[\Phase_r]}(\lambda_r\pi_{r,\phase}/\weight_r)\|(\bm{e}^{(r,\phase)})_\shortparallel^s\|^2,\nonumber\\
\mathcal{T}_{r,2}&\le (\max_r \weight_r) \sum_{\substack{\phase_1,\phase_2\in[\Phase_r]\\\phase_1\neq\phase_2}}(\overline{N}_{r,\phase_1}x_r/\weight_r)(\transition_{r})_{\phase_1,\phase_2}\|-(\bm{e}^{(r,\phase_1)})_{\shortparallel}^s+(\bm{e}^{(r,\phase_2)})_{\shortparallel}^s\|^2,\label{eq-boundTr2}\\
\mathcal{T}_{r,3}&\le (\max_r \weight_r) \sum_{\phase\in[\Phase_r]}(\overline{N}_{r,\phase}x_r/\weight_r)\biggl(-\sum_{\phase'}(\transition_r)_{\phase,\phase'}\biggr)\|(\bm{e}^{(r,\phase)})_{\shortparallel}^s\|^2,\label{eq-boundTr3}
\end{align}
where \eqref{eq-boundTr2} holds since $(\transition_{r})_{\phase_1,\phase_2}\ge 0$ when $\phase_1\neq\phase_2$, and \eqref{eq-boundTr3} is true since $\sum_{\phase'}(\transition_r)_{\phase,\phase'}\le 0$. It can be easily verified that $\expect[\overline{N}_{r,\phase}x_r]=\rho_{r,\phase}$ using the fact that $\expect[\Delta \overline{N}_{r,\phase}]=0$. Then
\begin{align*}
\expect[\mathcal{T}_{r,2}]&\le
(\max_r \weight_r) \sum_{\substack{\phase_1,\phase_2\in[\Phase_r]\\\phase_1\neq\phase_2}}(\rho_{r,\phase_1}/\weight_r)(\transition_{r})_{\phase_1,\phase_2}\|-(\bm{e}^{(r,\phase_1)})_{\shortparallel}^s+(\bm{e}^{(r,\phase_2)})_{\shortparallel}^s\|^2\\
&=\underbrace{(\max_r \weight_r)\sum_{\phase_1}(\rho_{r,\phase_1}/\weight_r)\biggl(\sum_{\phase_2\neq\phase_1}(\transition_r)_{\phase_1,\phase_2}\biggr)\|(\bm{e}^{(r,\phase_1)})_{\shortparallel}^s\|^2}_{\mathcal{T}_{r,2a}}\\
&\mspace{23mu}+\underbrace{(\max_r \weight_r)\sum_{\phase_2}\Biggl(\sum_{\phase_1\neq\phase_2} (\rho_{r,\phase_1}/\weight_r)(\transition_r)_{\phase_1,\phase_2}\Biggr)\|(\bm{e}^{(r,\phase_2)})_{\shortparallel}^s\|^2}_{\mathcal{T}_{r,2b}}\\
&\mspace{23mu}+\underbrace{2(\max_r \weight_r)\sum_{\phase_1,\phase_2:\phase_1\neq\phase_2}(\rho_{r,\phase_1}/\weight_r)(-\transition_r)_{\phase_1,\phase_2}\langle (\bm{e}^{(r,\phase_1)})_{\shortparallel}^s,(\bm{e}^{(r,\phase_2)})_{\shortparallel}^s\rangle}_{\mathcal{T}_{r,2c}},\\
\expect[\mathcal{T}_{r,3}]&\le(\max_r \weight_r)\sum_{\phase_1}(\rho_{r,\phase_1}/\weight_r)\biggl(-\sum_{\phase_2}(\transition_r)_{\phase_1,\phase_2}\biggr)\|(\bm{e}^{(r,\phase_1)})_{\shortparallel}^s\|^2.
\end{align*}
Thus,
\begin{equation*}
\mathcal{T}_{r,2a}+\expect[\mathcal{T}_{r,3}]\le(\max_r \weight_r)\sum_{\phase}(\rho_{r,\phase}/\weight_r)(-\transition_r)_{\phase,\phase}\|(\bm{e}^{(r,\phase)})_{\shortparallel}^s\|^2.
\end{equation*}
Since $(-\transition_r)^T\bm{\rho}_r=\lambda_r\bm{\pi}_r^T$, we have $\sum_{\phase_1:\phase_1\neq\phase_2}\rho_{r,\phase_1}(\transition_r)_{\phase_1,\phase_2}=-\lambda_r\pi_{r,\phase_2}+\rho_{r,\phase_2}(-\transition_r)_{\phase_2,\phase_2},\forall \phase_2$.
Therefore,
\begin{equation*}
\mathcal{T}_{r,1}+\mathcal{T}_{r,2b}\le(\max_r \weight_r)\sum_{\phase}(\rho_{r,\phase}/\weight_r)(-\transition_r)_{\phase,\phase}\|(\bm{e}^{(r,\phase)})_{\shortparallel}^s\|^2.
\end{equation*}
Arranging these terms we have the following bound for $\expect[B_1(\steady)]$:
\begin{align*}
\expect[B_1(\steady)]&\le 2(\max_r \weight_r)\sum_r\sum_{\phase_1,\phase_2}(\rho_{r,\phase_1}/\weight_r)(-\transition_r)_{\phase_1,\phase_2}\langle (\bm{e}^{(r,\phase_1)})_{\shortparallel}^s,(\bm{e}^{(r,\phase_2)})_{\shortparallel}^s\rangle.
\end{align*}
Since $\langle (\bm{e}^{(r,\phase_1)})_{\perp}^s, (\bm{e}^{(r,\phase_2)})_{\shortparallel}^s\rangle=0$, we further have
\begin{align*}
\frac{1}{2}\expect[B_1(\steady)]
&\le (\max_r \weight_r)\sum_r\sum_{\phase_1,\phase_2}(\rho_{r,\phase_1}/\weight_r)(-\transition_r)_{\phase_1,\phase_2}\langle \bm{e}^{(r,\phase_1)}, (\bm{e}^{(r,\phase_2)})_{\shortparallel}^s\rangle\\
&=(1-\epsilon)(\max_r \weight_r)\sum_r\sum_{\phase_1,\phase_2}(\rho^{(0)}_{r,\phase_1}/\weight_r)(-\transition_r)_{\phase_1,\phase_2}\langle \bm{e}^{(r,\phase_1)}, (\bm{e}^{(r,\phase_2)})_{\shortparallel}^s\rangle\\
&=(1-\epsilon)(\max_r \weight_r)\sum_{r,\phase_2}\langle \textrm{diag}(\bm{\rho}^{(0)}/\bm{\weight})(-\transition)\bm{e}^{(r,\phase_2)}, (\bm{e}^{(r,\phase_2)})_{\shortparallel}^s\rangle,
\end{align*}
\begin{sloppypar}
where $\bm{\rho}^{(0)}/\bm{\weight}$ denotes a vector whose $(r,\phase)$-th entry is $\rho^{(0)}_{r,\phase}/\weight_r$, and $\textrm{diag}(\bm{\rho}^{(0)}/\bm{\weight})$ denotes a diagonal matrix whose diagonal consists of entries of the vector $\bm{\rho}^{(0)}/\bm{\weight}$.

We now express $(\bm{e}^{(r,\phase)})_{\shortparallel}^s$ in a matrix form. Let $B_s$ denote the matrix whose rows are $(\bm{b}^{(\ell)})^T$'s with $\ell\in\satlinks$. Then $(\bm{e}^{(r,\phase)})_{\shortparallel}^s=B_s^T(B_sMB_s^T)^{-1}B_sM\bm{e}^{(r,\phase)}$, so we have
\begin{align}
\frac{1}{2}\expect[B_1(\steady)]
&\le (1-\epsilon)(\max_r \weight_r)\cdot\sum_{r,\phase}(\bm{e}^{(r,\phase)})^T(-\transition^T)\textrm{diag}(\bm{\rho}^{(0)}/\bm{\weight})M(B_s^T(B_sMB_s^T)^{-1}B_sM)\bm{e}^{(r,\phase)}\nonumber\\
&=(1-\epsilon)(\max_r \weight_r)\textrm{tr}\Bigl((-\transition^T)\textrm{diag}(\bm{\rho}^{(0)}/\bm{\weight})MB_s^T(B_sMB_s^T)^{-1}B_sM\Bigr)\nonumber\\
&=(1-\epsilon)(\max_r \weight_r)\textrm{tr}\Bigl(B_sM(-\transition^T)\textrm{diag}(\bm{\rho}^{(0)}/\bm{\weight})MB_s^T(B_sMB_s^T)^{-1}\Bigr)\label{eq-B2-bound-1}\\
&=(1-\epsilon)(\max_r \weight_r)\textrm{tr}(B_sMB_s^T(B_sMB_s^T)^{-1})\label{eq-B2-bound-2}\\
&=(1-\epsilon)(\max_r \weight_r)L_s\nonumber,
\end{align}
where \eqref{eq-B2-bound-1} follows from that $\textrm{tr}(XY)=\textrm{tr}(YX)$ for any matrices $X$ and $Y$, \eqref{eq-B2-bound-2} follows from Lemma~\ref{LEM-ROTATE}, and recall that $L_s$ is the number of critically loaded links in the network. Similarly, we can show that
\begin{equation*}
\frac{1}{2}\expect[B_1(\steady)]\ge (1-\epsilon)(\min_r\weight_r) L_s.
\end{equation*}
\end{sloppypar}
\end{enumerate}

Combining (i), (ii) and (iii) for the terms in \eqref{eq-Lyapunov=0} gives
\begin{equation*}
\frac{L_s\cdot \min_r \weight_r}{\epsilon}+o\biggl(\frac{1}{\epsilon}\biggr)\le \expect\left[\sum_r \weight_r \overline{N}_r^{(\epsilon)}\right]
\le\frac{L_s\cdot \max_r \weight_r}{\epsilon}+o\biggl(\frac{1}{\epsilon}\biggr).
\end{equation*}
The proof will be completed after we prove Lemma~\ref{lem-continuity2} below.

\proof{Proof of Lemma~\ref{lem-continuity2}.}
\begin{sloppypar}
This proof is analogous to the proof of Lemma~\ref{lem-continuity2-exp}. We first give the following bounds on the rate allocation and the Lagrange multipliers.
\end{sloppypar}

\begin{claim}\label{claim-continuity-instant}
There exist positive constants $B_4$ and $B_5$ such that for any state $\bm{n}$,
\begin{equation*}
\sum_{\ell} p_{\ell}\|\widehat{\bm{b}}^{(\ell)}-\bm{b}^{(\ell)}\|^2\le B_4\|\nperp\|,\text{ and }\sum_{\ell:\ell\in r,\ell\in\mathcal{L}\setminus\satlinks}p_{\ell}\le B_5\|\nperp\|,\forall r.
\end{equation*}
\end{claim}
\proof{Proof of Claim~\ref{claim-continuity-instant}}
Consider the term $\langle \bm{n}, (-\transition^T)\allowbreak(\bm{\rho}^{(0)}-\bm{nx})\rangle$.  By the proof of Lemma~\ref{lem-nperp-drift}, we have
\begin{align}
\langle \bm{n}, (-\transition^T)(\bm{\rho}^{(0)}-\bm{nx})\rangle \le -\weight_{\min}\eta_{\min}\sum_{\ell}p_{\ell}\|\widehat{\bm{b}}^{(\ell)}-\bm{b}^{(\ell)}\|^2-\sum_{\ell\in\mathcal{L}\setminus\satlinks}p_{\ell}\delta_{\ell}.\label{eq-left}
\end{align}
Then as in the proof of Lemma~\ref{claim-continuity-instant-exp}, we consider $\bm{y}=\frac{\transition^T(\bm{\rho}^{(0)}-\bm{nx})}{\|\transition^T(\bm{\rho}^{(0)}-\bm{nx})\|}$. We can verify that $\bm{y}\in\mathcal{K}^{\circ}$ since $\langle \bm{b}^{(\ell)},\transition^T(\bm{\rho}^{(0)}-\bm{nx})\rangle = -U_{\ell}\le 0$ for all $\ell\in\satlinks$ by \ref{cond-C1}. Then by arguments similar to those in Lemma~\ref{claim-continuity-instant-exp}, we have
\begin{equation*}
\weight_{\min}\eta_{\min}\sum_{\ell}p_{\ell}\|\widehat{\bm{b}}^{(\ell)}-\bm{b}^{(\ell)}\|^2+\sum_{\ell\in\mathcal{L}\setminus\satlinks}p_{\ell}\delta_{\ell}
\le \|\transition^T(\bm{\rho}^{(0)}-\bm{nx})\|\cdot\|\nperp\|.
\end{equation*}

Since each summand on the left-hand-side above is nonnegative, we have that
\begin{equation*}
\sum_{\ell}p_{\ell}\|\widehat{\bm{b}}^{(\ell)}-\bm{b}^{(\ell)}\|^2\le\frac{\|\bm{\lambda}^{(0)}-\bm{nx\mu}\|}{\weight_{\min}\eta_{\min}}\|\nperp\|,
\end{equation*}
and for each $r$,
\begin{equation*}
\sum_{\ell:\ell\in r,\ell\in\mathcal{L}\setminus\satlinks}p_{\ell}\le\frac{1}{\delta_{\min}}\sum_{\ell:\ell\in r,\ell\in\mathcal{L}\setminus\satlinks}p_{\ell}\delta_{\ell}\le\frac{\|\bm{\lambda}^{(0)}-\bm{nx\mu}\|}{\delta_{\min}}\|\nperp\|,
\end{equation*}
where $\delta_{\min}>0$ is defined as $\delta_{\min}=\min\{\delta_{\ell}\colon \ell\in\mathcal{L}\setminus \satlinks\}$. Since each entry of the rate allocation $\bm{nx}$ can be bounded using a constant independent of $\bm{n}$ and $\epsilon$, there exist positive constants $B_4$ and $B_5$ such that $\frac{\|\transition^T(\bm{\rho}^{(0)}-\bm{nx})\|}{\weight_{\min}\eta_{\min}}\le B_4$ and $\frac{\|\transition^T(\bm{\rho}^{(0)}-\bm{nx})\|}{\delta_{\min}\eta_{\min}}\le B_5$. Therefore,
\begin{equation*}
\sum_{\ell}p_{\ell}\|\widehat{\bm{b}}^{(\ell)}-\bm{b}^{(\ell)}\|^2\le B_4\|\nperp\|,\text{ and }\sum_{\ell:\ell\in r,\ell\in\mathcal{L}\setminus\satlinks}p_{\ell}\le B_5\|\nperp\|,\forall r,
\end{equation*}
which completes the proof of the claim. \Halmos
\endproof

With Claim~\ref{claim-continuity-instant}, the remainder of the proof for Lemma~\ref{lem-continuity2} follows arguments similar to those in  the proof of Lemma~\ref{lem-continuity2-exp}, noting the equivalence of all the norms in the space $\mathbb{R}^{\Phase}$. \Halmos
\endproof

\subsection{Proof of Theorem~\ref{THM-INTERCHANGE} (Interchange of Limits).}\label{sec-interchange}
The proof is fairly standard (see, e.g., the proof of interchange of limits in \cite{ShaTsiZho_14}).  We will first show the tightness of the set of probability distributions of $\hat{\bm{N}}^{(\epsilon)}(\infty)$'s, using the bound on $\expect\bigl[\sum_{r,\phase}\overline{N}_{r,\phase}^{(\epsilon)}\bigr]$ indicated by the tail bound in Lemma~\ref{LEM-SUM-NRK}.  We will then use our state-space collapse result in steady state to show that every convergent subsequence of $(\hat{\bm{N}}^{(\epsilon)}(\infty)\colon\epsilon>0)$ converges to $\hat{\bm{N}}(\infty)$.  Then the convergence $\hat{\bm{N}}^{(\epsilon)}(\infty)\Rightarrow\hat{\bm{N}}(\infty)$ as $\epsilon\to 0^+$ will follow from the Prokhorov's theorem \cite{Bil_71}.

\emph{Tightness.}  Note that $\hat{\bm{N}}^{(\epsilon)}(\infty)\stackrel{d}{=}\epsilon\steady^{(\epsilon)}$, where ``$\stackrel{d}{=}$'' denotes being identically distributed.  The tail bound in Lemma~\ref{LEM-SUM-NRK} indicates that there exists a constant $B_8$ such that for any $\epsilon>0$,
\begin{equation*}
\expect\Biggl[\sum_{r,\phase}\hat{N}^{(\epsilon)}_{r,\phase}(\infty)\Biggr]=\epsilon\expect\Biggl[\sum_{r,\phase}\overline{N}^{(\epsilon)}_{r,\phase}\Biggr]\le B_8.
\end{equation*}
This implies the tightness of the set of probability distributions of $\hat{\bm{N}}^{(\epsilon)}(\infty)$'s.

\emph{Limit of a convergent subsequence.}  We arbitrarily pick a convergent subsequence $(\hat{\bm{N}}^{(\epsilon_i)}(\infty)\colon i=1,2,\dots)$ and assume that $\hat{\bm{N}}^{(\epsilon_i)}(\infty)\Rightarrow \tilde{\bm{N}}$ as $i\to\infty$.  We will show that $\tilde{\bm{N}}\stackrel{d}{=}\hat{\bm{N}}(\infty)$.  Let each $\epsilon_i$-th system start from steady state, i.e., $\hat{\bm{N}}^{(\epsilon_i)}(0)\stackrel{d}{=}\hat{\bm{N}}^{(\epsilon_i)}(\infty)$.  We wish to apply Theorem~6.1 in \cite{VlaZhaZwa_14}, so we need to verify that $\tilde{\bm{N}}$ is in the cone $\mathcal{K}$ almost surely, i.e., $\tilde{\bm{N}}_{\perp}$ is the all-zero vector almost surely.  Since the function that maps a vector to its perpendicular component is continuous, we have $\hat{\bm{N}}^{(\epsilon_i)}_{\perp}(\infty)\Rightarrow \tilde{\bm{N}}_{\perp}$ as $i\to\infty$.  By our state-space collapse result in Theorem~\ref{THM-SSC},
\begin{align*}
\expect[\|\hat{\bm{N}}^{(\epsilon_i)}_{\perp}(\infty)\|]=\epsilon_i\expect[\|\steady^{(\epsilon_i)}_{\perp}\|]=O(\sqrt{\epsilon_i})\to 0,\quad \text{as }i\to \infty.
\end{align*}
Therefore, $\hat{\bm{N}}^{(\epsilon_i)}_{\perp}(\infty)\Rightarrow 
\bm{0}$ as $i\to\infty$, and thus $\tilde{\bm{N}}_{\perp}=\bm{0}$ almost surely.  Now we can apply Theorem~6.1 in \cite{VlaZhaZwa_14}.  Recall that we let each process $(\hat{\bm{N}}^{(\epsilon_i)}(t)\colon t\ge 0)$ start from steady state, so $\hat{\bm{N}}^{(\epsilon_i)}(t)\stackrel{d}{=}\hat{\bm{N}}^{(\epsilon_i)}(\infty)$ for all $t\ge 0$.  Therefore, by the convergence of $(\hat{\bm{N}}^{(\epsilon_i)}(t)\colon t\ge 0)$ to $(\hat{\bm{N}}(t)\colon t\ge 0)$ in Theorem~6.1 in \cite{VlaZhaZwa_14}, $\hat{\bm{N}}(t)\stackrel{d}{=}\tilde{\bm{N}}$ for all $t\ge 0$.  This shows that the process $(\hat{\bm{N}}(t)\colon t\ge 0)$ is also stationary, and thus $\tilde{\bm{N}}\stackrel{d}{=}\hat{\bm{N}}(\infty)$, which completes the proof.

\section{Conclusions and future work.}\label{sec-conclusions}
In this paper, we studied the weighted proportionally fair policy for bandwidth sharing networks under the assumption that the file size distributions belong to a dense class of phase-type distributions.  We directly analyzed the steady state of the system using the drift method. We established a multiplicative-type state-space collapse result in steady state and obtained explicit-form upper and lower bounds on the weighted sum of the expected number of flows on different routes, where the weights are the same as those used in the weighted proportionally fair policy. These bounds are heavy-traffic insensitive in the sense that their dominant terms do not depend on the specific forms of the file size distributions in the class. The state-space collapse result also implies the interchange of limits for the diffusion approximation result of \citet{VlaZhaZwa_14}. Our results in this paper demonstrated that the drift method can be applied to this sophisticated model, and can give explicit-form bounds even in settings where the diffusion approximation approach cannot.  An interesting direction that deserves further exploration is to derive higher moment bounds using the drift method. \citet{ErySri_12} have shown such a potential since they have derived higher moment bounds for routing and scheduling algorithms in a different setting.

\section*{Acknowledgments.}
\begin{sloppypar}
This work was supported in part by NSF Grants ECCS-1609202, ECCS-1739344, ECCS-1739189, ECCS-1609370, CMMI-1562276, CIF-1409106, CNS-2007733, the U.S.\ Army Research Office (ARO Grant No.\ W911NF-16-1-0259), and the U.S.\ Office of Naval Research (ONR Grant No.\ N00014-15-1-2169).
\end{sloppypar}


\bibliographystyle{informs2014} 
\bibliography{refs-weina} 


\iftoggle{complete}{%
\begin{APPENDICES}

\section{Tail and moment bounds for continuous-time Markov chains.}\label{app-tail}
\begin{lemma}\label{LEM-TAIL}
Let $(X(t)\colon t\ge 0)$ be a continuous-time Markov chain over a countable state space $\mathcal{X}$. Suppose that it is irreducible, nonexplosive and positive-recurrent, and it converges in distribution to a random variable $\overline{X}$. Consider a Lyapunov function $V\colon \mathcal{X}\rightarrow \mathbb{R}_+$ and define the drift of $V$ at a state $i\in\mathcal{X}$ as
\begin{equation*}
\Delta V(i)=\sum_{i'\in\mathcal{X}:i'\neq i}q_{ii'}(V(i')-V(i)),
\end{equation*}
where $q_{ii'}$ is the transition rate from $i$ to $i'$. Suppose that the drift satisfies the following conditions:
\begin{enumerate}[label=(\roman*),leftmargin=3.2em]
\item There exist constants $\gamma>0$ and $B>0$ such that $\Delta V(i)\le -\gamma$ for any $i\in\mathcal{X}$ with $V(i)>B$.

\item $\displaystyle
\nu_{\max}\triangleq\sup_{i,i'\in\mathcal{X}\colon q_{ii'}>0}|V(i')-V(i)|<+\infty.$

\item $\displaystyle \overline{q}\triangleq\sup_{i\in\mathcal{X}}(-q_{ii})<+\infty.$
\end{enumerate}
Then for any nonnegative integer $j$, we have
\begin{equation}\label{eq-tail}
\Pr(V(\overline{X})>B+2\nu_{\max}j)\le\biggl(\frac{q_{\max}\nu_{\max}}{q_{\max}\nu_{\max}+\gamma}\biggr)^{j+1},
\end{equation}
where
\begin{equation*}
q_{\max}=\sup_{i\in\mathcal{X}}\sum_{i'\in\mathcal{X}\colon V(i)<V(i')}q_{ii'}.
\end{equation*}
As a result, the $m$-th moment of $V(\overline{X})$ for any $m\in\mathbb{Z}_+$ can be bounded as follows:
\begin{equation}
\expect[V(\overline{X})^m]\le (2B)^m+(4\nu_{\max})^m\biggl(\frac{q_{\max}\nu_{\max}+\gamma}{\gamma}\biggr)^mm!.
\end{equation}
\end{lemma}

\proof{Proof.}
For the continuous-time Markov chain $(X(t)\colon t\ge 0)$, we consider the uniformized \cite{Kei_12} discrete-time Markov chain $(\hat{X}(t)\colon t=1,2,\dots)$ with $\overline{q}$ as the uniform rate parameter, i.e., its transition probability from $i$ to $i'$ is defined by
\begin{equation*}
p_{i i'}=
\begin{cases}
\frac{q_{i i'}}{\overline{q}}& i'\neq i,\\
1- \frac{\sum_{j\neq i}q_{ij}}{\overline{q}}& i'=i.
\end{cases}
\end{equation*}
Then $(\hat{X}(t)\colon t=1,2,\dots)$ has the same stationary distribution as $(X(t)\colon t\ge 0)$. So it suffices to prove that the same Lyapunov function $V$ for this discrete-time Markov chain $(\hat{X}(t)\colon t=1,2,\dots)$ has the tail bound \eqref{eq-tail}. We prove this by applying \citet{BerGamTsi_01}'s tail bound to this discrete-time Markov chain. Below we verify that $V$ satisfies the following three conditions required in \cite{BerGamTsi_01}, where $\Delta_d V(i)$ denotes its drift with respect to $(\hat{X}(t)\colon t=1,2,\dots)$:
\begin{enumerate}[leftmargin=3.4em,label=(\roman*')]
\item The drift $\Delta_d V(i)\le -\frac{\gamma}{\overline{q}}$ for any $i\in\mathcal{X}$ with $V(i)>B$.

\item $\displaystyle
\nu_{\max}\triangleq\sup_{i,i'\in\mathcal{X}\colon p_{ii'}>0}|V(i')-V(i)|<+\infty.$

\item $\displaystyle \expect[V(\overline{X})]<+\infty.$
\end{enumerate}

To prove (i'), we write the drift $\Delta_d V(i)$ as follows:
\begin{align*}
\Delta_d V(i)&=\sum_{i'\neq i}p_{i i'}(V(i')-V(i))\\
&=\frac{1}{\overline{q}}\sum_{i'\neq i}q_{i i'}(V(i')-V(i)).
\end{align*}
Then (i') follows from condition (i) since for any $i$ with $V(i)>B$,
\begin{equation*}
\Delta_d V(i)=\frac{1}{\overline{q}}\Delta V(i)\le -\frac{\gamma}{\overline{q}}.
\end{equation*}

Condition (ii') follows from the fact that $p_{i i'}>0$ either when $q_{i i'}>0$ or when $i'=i$.

Condition (iii') follows from the results in \cite{Haj_82}, since the required conditions there are satisfied due to (i') and (ii').

Now we have verified the conditions required in \cite{BerGamTsi_01}. Then by Theorem~1 in \cite{BerGamTsi_01}, for any nonnegative integer $j$, we have
\begin{equation*}
\Pr(V(\overline{X})>B+2\nu_{\max}j)\le\biggl(\frac{p_{\max}\nu_{\max}}{p_{\max}\nu_{\max}+\gamma/\overline{q}}\biggr)^{j+1},
\end{equation*}
where
\begin{equation*}
p_{\max}=\sup_{i\in\mathcal{X}}\sum_{i'\in\mathcal{X}\colon V(i)<V(i')}p_{ii'}.
\end{equation*}
By the definition of the transition probabilities, $p_{\max}=q_{\max}/\overline{q}$. Thus
\begin{equation*}
\Pr(V(\overline{X})>B+2\nu_{\max}j)\le\biggl(\frac{q_{\max}\nu_{\max}}{q_{\max}\nu_{\max}+\gamma}\biggr)^{j+1},
\end{equation*}
which is the tail bound in \eqref{eq-tail}.

With the tail bound, the moment bounds follow from the proof of Lemma~3 in \citep{MagSri_16}. \Halmos
\endproof

\section{Proof of Lemma~\ref{LEM-SUM-NR-EXP}.}\label{app-proof-lem-sum-nr-exp}
\proof{Proof.}
We first derive a bound on the drift $\Delta\|\bm{n}\|$. By the proof of Lemma~\ref{lem-nperp-drift-exp}, inserting \eqref{eq-bound-npart-exp} into \eqref{eq-bound-n2-drift-exp} gives
\begin{align*}
\Delta \|\bm{n}\|&\le \frac{\Delta \|\bm{n}\|^2}{2\|\bm{n}\|}\\
&\le -\epsilon\frac{\langle \bm{n},\bm{\lambda}^{(0)}\rangle}{\|\bm{n}\|}+\frac{A_1}{\|\bm{n}\|}&\\
&=-\epsilon\frac{\sum_r\weight_rn_r}{\|\bm{n}\|}+\frac{A_1}{\|\bm{n}\|}.
\end{align*}
Since norms are equivalent in $\mathbb{R}^{R}$, there exist positive constants $A_3$ and $A_4$ such that for any $\bm{n}$, $
A_3\|\bm{n}\| \le \sum_r\weight_rn_r\le A_4\|\bm{n}\|$.
Therefore,
\begin{equation*}
\Delta\|\bm{n}\|\le -\epsilon A_3+\frac{A_1}{\|\bm{n}\|}.
\end{equation*}
Then
\begin{equation*}
\Delta \|\bm{n}\|\le -\frac{\epsilon A_3}{2},\quad\text{when }\|\bm{n}\|\ge \frac{2A_1}{\epsilon A_3}.
\end{equation*}
Recall that
\begin{equation*}
\sup_{\bm{n},\bm{n}'\colon q_{\bm{n}\bm{n}'}>0} \bigl|\|\bm{n}'\|-\|\bm{n}\|\bigr|\le \nu,\quad
\sup_{\bm{n}}\sum_{\bm{n}'\colon \|\bm{n}\|<\|\bm{n}'\|}q_{\bm{n}\bm{n}'}\le\zeta,
\end{equation*}
and $\nu <+\infty$ and $\overline{q}<+\infty$. Then by Lemma~\ref{LEM-TAIL}, for any nonnegative integer~$j$,
\begin{equation*}
\Pr\biggl(\|\steady\|>\frac{2A_1}{\epsilon A_3 }+2\nu j\biggr)\le\Biggl(\frac{\zeta \nu }{\zeta \nu +\epsilon A_3/2}\Biggr)^{j+1}.
\end{equation*}
Let
\begin{equation*}
\beta=\frac{\zeta \nu }{\zeta \nu +\epsilon A_3/2}.
\end{equation*}
Then,
\begin{align*}
&\mspace{23mu}\Pr\biggl(\sum_r\weight_r\overline{N}_r>\frac{2A_1A_4}{\epsilon A_3}+2\nu A_4j\biggr)\\
&\le\Pr\biggl(\|\steady\|>\frac{2A_1}{\epsilon A_3 }+2\nu j\biggr)\\
&\le\beta^{j+1}.
\end{align*}
This completes the proof of Lemma~\ref{LEM-SUM-NR-EXP}. \Halmos
\endproof

\section{Validity of the constructed inner product.}\label{app-proof-lem-inner-well-defined}
\begin{lemma}\label{LEM-INNER-PRODUCT}
The inner product defined by the matrix $M$ is well-defined, i.e., the matrix $M$ is well-defined and positive definite.
\end{lemma}
\proof{Proof.}
It suffices to prove that for each route $r$, $M_r$ is well-defined and positive definite since $M$ is block-diagonal.

We first prove that $M_r$ is well-defined, i.e., the integral below is (entry-wise) finite:
\begin{equation*}
M_r=\frac{\weight_r}{\lambda_r^{(0)}}\int_0^{+\infty}\frac{\exp(\transition_r\sigma)\bm{1}\bm{1}^T\exp(\transition_r^T\sigma)}{\bm{\pi}_r(-\transition_r)^{-1}\exp(\transition_r\sigma)\bm{1}}d\sigma.
\end{equation*}
Let $\bm{P}(\sigma)=\exp(\transition_r\sigma)\bm{1}$ and $\chi=(-\transition_r)^{-1}$. Then from the properties of phase-type distributions we know that $\chi_{\phase_1\phase_2}$ is the expected time spent in phase $\phase_2$ given that phase $\phase_1$ is the initial state. Let $\overline{\bm{\chi}}=(\bm{\pi}_r\chi)^T$.
Then $\overline{{\chi}}_{\phase}$ is the expected time spent in phase $\phase$ for the initial distribution $\bm{\pi}_r$. Therefore, $\overline{{\chi}}_{\phase}>0$ for all $\phase\in[\Phase_r]$. With the above notation, the $(\phase_1,\phase_2)$th entry of $M_r$ can be written as
\begin{align*}
(M_r)_{\phase_1,\phase_2}&=\frac{\weight_r}{\lambda_r^{(0)}}\int_{0}^{+\infty}\frac{P_{\phase_1}(\sigma)P_{\phase_2}(\sigma)}{\sum_{\phase}\overline{\chi}_{\phase}P_{\phase}(\sigma)}d\sigma\\
&\le\frac{\weight_r}{\lambda_r^{(0)}}\int_{0}^{+\infty}\frac{P_{\phase_1}(\sigma)P_{\phase_2}(\sigma)}{\overline{\chi}_{\phase_2}P_{\phase_2}(\sigma)}d\sigma\\
&=\frac{\weight_r}{\lambda_r^{(0)}\overline{\chi}_{\phase_2}}\int_{0}^{+\infty}P_{\phase_1}(\sigma)d\sigma.
\end{align*}
By our assumptions, $\transition_r$ is an upper triangular matrix with all the main diagonal entries being negative. So $\transition_r$ is invertible and $\lim_{\sigma\rightarrow +\infty}\exp(\transition_r^T\sigma)$ is an all-zero matrix. Therefore, $P_{\phase_1}(\sigma)$ is integrable and thus $M_r$ is well-defined.

Next we prove that $M_r$ is positive definite.  Let $G(u)$ denote the complementary cumulative distribution function (CCDF) of the file size distribution on route $r$. Then
\begin{equation*}
G(u)=\bm{\pi}_r\exp(\transition_r u)\bm{1},\quad \int_{0}^{+\infty}G(u)du=\frac{1}{\servrate_r}.
\end{equation*}
The denominator inside the integral of $M_r$ can be written as
\begin{equation*}
\bm{\pi}_r(-\transition_r)^{-1}\exp(\transition_r\sigma)\bm{1}=\frac{1}{\servrate_r}-\int_{0}^{\sigma}G(u)du,
\end{equation*}
which is positive for all $\sigma\ge 0$. Therefore, it is obvious that $M_r$ is positive semi-definite. Further, $M_r$ is positive definite if and only if there exists no $\bm{y}\neq \bm{0}$ such that
\begin{equation}\label{eq-all-zero}
\bm{y}^T\exp(\transition_r\sigma)\bm{1}=\bm{0},\text{ for all }\sigma\ge 0,
\end{equation}
where $\bm{0}$ is an all-zero vector with dimension $\Phase_r\times 1$. If we view the pair $(\transition_r,\bm{1})$ as the $(A,B)$ matrix of a control system, \eqref{eq-all-zero} is equivalent to the controllability of the system $(A,B)$ \citep{Son_90}. By the Popov-Belevitch-Hautus (PBH) test (also referred to as Hautus Lemma) in control theory \citep{Son_90}, this is equivalent to that $\text{rank}[\lambda I-\transition_r,\bm{1}]=\Phase_r$ for each eigenvalue $\lambda$ of the matrix $\transition_r$, which is further equivalent to that $\transition_r^T$ has no eigenvector $\bm{v}$ such that $\bm{v}^T\bm{1}=0$. Now let us look at the eigenvectors of $\transition_r^T$. Recall that $\transition_r^T$ has a block-diagonal structure given in \eqref{eq-transition_r-block} with rate $\mu_r^{(b)}$ for each block $b$, where $\mu_r^{(b)}$'s are positive and distinct. Then the $-\mu_r^{(b)}$'s are the eigenvalues of $\transition_r^T$. Let $\bm{v}$ be an eigenvector associated with the eigenvalue $-\mu_r^{(b)}$. Then $\bm{v}$ satisfies that
\begin{align*}
&\begin{bmatrix}
(\transition_r^{(1)})^T+\mu_r^{(b)}I & \mspace{-18mu} 0 & \cdots & 0\\
0 & \mspace{-18mu}(\transition_r^{(2)})^T+\mu_r^{(b)}I & \cdots & 0\\
\vdots & \mspace{-18mu}\vdots & \ddots & \vdots\\
0 & \mspace{-18mu}0 & \cdots & (\transition_r^{(B_r)})^T+\mu_r^{(b)}I
\end{bmatrix}
\bm{v}\\
&=\bm{0}.
\end{align*}
Since $(\transition_r^{(b')})^T+\mu_r^{(b)}I$ is full rank for all $b'\neq b$ and
\begin{equation*}
(\transition_r^{(b)})^T+\mu_r^{(b)}I=
\begin{bmatrix}
0 & 0 & \cdots & 0\\
2\mu_r^{(b)} & 0 & \cdots & 0\\
\vdots & \ddots & \ddots & \vdots\\
0 & \cdots & 2\mu_r^{(b)} & 0
\end{bmatrix},
\end{equation*}
the eigenvector $\bm{v}$ has only one nonzero entry. Then $\bm{v}^T\bm{1}\neq 0$. This completes the proof that $M_r$ is positive definite. \Halmos
\endproof

\section{Properties of the constructed inner product.}\label{app-proof-lems-inner}
Below we first prove Lemmas~\ref{LEM-ROTATE} and \ref{LEM-LARGER-NORM} and then prove Properties \ref{cond-C1} and \ref{cond-C2}.

\proof{Proof of Lemma~\ref{LEM-ROTATE}.}
By the definitions of $\bm{\rho}_r^{(0)}$ and $M_r$ in \eqref{eq-load-vector} and \eqref{eq-Mr}, respectively,
\begin{align*}
&\mspace{23mu}\frac{1}{\weight_r}(\bm{\rho}_r^{(0)})^TM_r(-\transition_r)^T\\
&=\bm{\pi}_r(-\transition_r)^{-1}\int_0^{+\infty}\frac{\exp(\transition_r\sigma)\bm{1}\bm{1}^T\exp(\transition_r^T\sigma)(-\transition_r)^T}{\bm{\pi}_r(-\transition_r)^{-1}\exp(\transition_r\sigma)\bm{1}}d\sigma\\
&=\int_0^{+\infty}\bm{1}^T\exp(\transition_r^T\sigma)(-\transition_r)^Td\sigma\\
&=-\bm{1}^T\exp(\transition_r^T\sigma)\Bigm|_0^{+\infty}\\
&=\bm{1}^T.
\end{align*}
Here we have used that fact that $\lim_{\sigma\rightarrow +\infty}\exp(\transition_r^T\sigma)$ is an all-zero matrix since $\transition_r$ is an upper triangular matrix and its main diagonal entries are all negative. This completes the proof. \Halmos
\endproof

\proof{Proof of Lemma~\ref{LEM-LARGER-NORM}.}
We first derive another representation of $M_r(-\transition_r^T)+(-\transition_r)M_r$. Let
\begin{equation*}
M_r(t)=\frac{\weight_r}{\lambda_r^{(0)}}\int_0^{t}\frac{\exp(\transition_r\sigma)\bm{1}\bm{1}^T\exp(\transition_r^T\sigma)}{\bm{\pi}_r(-\transition_r)^{-1}\exp(\transition_r\sigma)\bm{1}}d\sigma.
\end{equation*}
Then
\begin{equation*}
M_r(-\transition_r^T)+(-\transition_r)M_r=\mspace{-12mu}\lim_{t\rightarrow +\infty}\Bigl(M_r(t)(-\transition_r^T)+(-\transition_r)M_r(t)\Bigr).
\end{equation*}
We can verify that
\begin{align}
&\mspace{23mu}M_r(t)(-\transition_r^T)+(-\transition_r)M_r(t)\nonumber\\
&=-\frac{\weight_r}{\lambda_r^{(0)}}\frac{\exp(\transition_rt)\bm{1}\bm{1}^T\exp(\transition_r^Tt)}{\bm{\pi}_r(-\transition_r)^{-1}\exp(\transition_rt)\bm{1}}+\frac{\weight_r}{\overline{\rho}_r}\bm{1}\bm{1}^T\nonumber\\
&\mspace{21mu}+\frac{\weight_r}{\lambda_r^{(0)}}\int_0^{t}\frac{\exp(\transition_r\sigma)\bm{1}\bm{1}^T\exp(\transition_r^T\sigma)}{\bm{\pi}_r(-\transition_r)^{-1}\exp(\transition_r\sigma)\bm{1}}\cdot\frac{\bm{\pi}_r\exp(\transition_r\sigma)\bm{1}}{\bm{\pi}_r(-\transition_r)^{-1}\exp(\transition_r\sigma)\bm{1}}d\sigma.\label{eq-Lyapunov-equation}
\end{align}
We have proved that $M_r$ is well-defined, so
\begin{equation*}
\lim_{t\rightarrow 0}\biggl(-\frac{\weight_r}{\lambda_r^{(0)}}\frac{\exp(\transition_rt)\bm{1}\bm{1}^T\exp(\transition_r^Tt)}{\bm{\pi}_r(-\transition_r)^{-1}\exp(\transition_rt)\bm{1}}\biggr)=(0)_{\Phase_r\times\Phase_r},
\end{equation*}
where $(0)_{\Phase_r\times\Phase_r}$ is the all-zero $\Phase_r\times\Phase_r$ matrix.

Now it suffices to prove that there exists a constant $\eta_r>0$ such that for any $\sigma\ge 0$,
\begin{equation}\label{eq-bound-hazard}
\frac{\bm{\pi}_r\exp(\transition_r\sigma)\bm{1}}{\bm{\pi}_r(-\transition_r)^{-1}\exp(\transition_r\sigma)\bm{1}}\ge 2\eta_r,
\end{equation}
since combining this with \eqref{eq-Lyapunov-equation} implies that for any $\bm{y}\in\mathbb{R}^{\Phase_r}$,
\begin{align*}
&\mspace{23mu}\bm{y}^T\biggl(\frac{1}{2}M_r(-\transition_r^T)+\frac{1}{2}(-\transition_r)M_r-\eta_rM_r\biggr)\bm{y}\\
&=\lim_{t\rightarrow +\infty}\bm{y}^T\biggl(\frac{1}{2}M_r(t)(-\transition_r^T)+\frac{1}{2}(-\transition_r)M_r(t)-\eta_rM_r(t)\biggr)\bm{y}\\
&\ge\lim_{t\rightarrow+\infty}\frac{\weight_r}{2\lambda_r^{(0)}}\int_0^{t}\frac{\bm{y}^T\exp(\transition_r\sigma)\bm{1}\bm{1}^T\exp(\transition_r^T\sigma)\bm{y}}{\bm{\pi}_r(-\transition_r)^{-1}\exp(\transition_r\sigma)\bm{1}}\cdot(2\eta_r)d\sigma-\eta_r\bm{y}^TM_r\bm{y}+\frac{\weight_r}{2\overline{\rho}_r}(\bm{y}^T\bm{1})^2\\
&=\frac{\weight_r}{2\overline{\rho}_r}(\bm{y}^T\bm{1})^2\\
&\ge 0.
\end{align*}
Let $g(u)$ and $G(u)$ denote the probability density function (PDF) and the complementary cumulative distribution function (CCDF) of the file size distribution on route $r$, respectively. Then
\begin{equation*}
G(\sigma)=\int_{\sigma}^{+\infty}g(u)du=\bm{\pi}_r\exp(\transition_r \sigma)\bm{1},
\end{equation*}
and
\begin{equation*}
\int_{\sigma}^{+\infty}G(u)du=\bm{\pi}_r(-\transition_r)^{-1}\exp(\transition_r\sigma)\bm{1}.
\end{equation*}
Thus \eqref{eq-bound-hazard} is equivalent to that there exists a constant $\eta_r>0$ such that for any $u\ge 0$,
$\frac{g(u)}{G(u)}\ge 2\eta_r,$
i.e., the hazard function is lower bounded by $2\eta_r$. Note that the eigenvalues of $\transition_r$ are the rates of the phases, denoted by $\mu_{r,\phase}$'s with $\phase\in[\Phase_r]$, which are all positive. Consider the Jordan canonical form $\transition_r=\Phi J\Phi^{-1}$. Then $\exp(\transition_ru)=\Phi\exp(Ju)\Phi^{-1}$,
where the $(i,j)$th entry of $\exp(Ju)$ is either $e^{-\mu_{r,i}u}u^{j-1}/(j-1)!$ or $0$. So $G(u)$ can be written as
\begin{equation*}
G(u)=\sum_{i,j\in[\Phase_r]}c_{ij}e^{-\mu_{r,i}u}\frac{u^{j-1}}{(j-1)!}
\end{equation*}
for some constants $c_{ij}$, and thus
\begin{equation*}
g(u)=-G'(u)=\sum_{i,j\in[\Phase_r]}c_{ij}\mu_{r,i}e^{-\mu_{r,i}u}\frac{u^{j-1}}{(j-1)!}.
\end{equation*}
Therefore, $\lim_{u\rightarrow +\infty}\frac{g(u)}{G(u)}\ge\min_{\phase\in[\Phase_r]}\mu_{r,\phase}>0$.
It can be verified that $g(u)>0$ and $G(u)>0$ for any $u\ge 0$. Thus there exists a constant $\eta_r>0$ such that for any $u\ge 0$, $\frac{g(u)}{G(u)}\ge 2\eta_r$. \Halmos
\endproof

\proof{Proofs of Properties \ref{cond-C1} and \ref{cond-C2}.}
Note that the constant $\eta_{\min}$ in property~\ref{cond-C2} is defined as $\eta_{\min}=\min_r\{\eta_r\}$ with the $\eta$'s given in Lemma~\ref{LEM-LARGER-NORM}.

We now prove \ref{cond-C1}:
\begin{align}
&\mspace{23mu}\langle \bm{b}^{(\ell)},(-\transition)^T(\bm{\rho}^{(0)}-\bm{nx})\rangle\nonumber\\
&= (\bm{b}^{(\ell)})^TM(-\transition)^T(\bm{\rho}^{(0)}-\bm{nx})\label{eq-condition-C1-proof-1}
\\
&=\sum_{r:\ell\in r}\frac{1}{\weight_r}(\bm{\rho}_r^{(0)})^TM_r(-\transition_r)^T(\bm{\rho}_r^{(0)}-\bm{n}_rx_r)\label{eq-condition-C1-proof-2}\\
&=\sum_{r:\ell\in r}\bm{1}^T(\bm{\rho}_r^{(0)}-\bm{n}_rx_r)\label{eq-condition-C1-proof-3}\\
&=C_{\ell}-\delta_{\ell}-\sum_{r:\ell\in r}\sum_{\phase}n_{r,\phase}x_r\label{eq-condition-C1-proof-4}\\
&=U_{\ell}-\delta_{\ell},\label{eq-condition-C1-proof-5}
\end{align}
where \eqref{eq-condition-C1-proof-2} follows from the definition of $\bm{b}^{(\ell)}$, \eqref{eq-condition-C1-proof-3} follows from Lemma~\ref{LEM-ROTATE}, and \eqref{eq-condition-C1-proof-4} follows from the heavy-traffic condition.

Next we prove condition \ref{cond-C2}. The inner product can be written in the following form:
\begin{align*}
&\mspace{23mu}\langle \bm{b}^{(\ell)}-\widehat{\bm{b}}^{(\ell)},(-\transition^T)(\bm{\rho}^{(0)}-\bm{nx})\rangle\\
&=(\bm{b}^{(\ell)}-\widehat{\bm{b}}^{(\ell)})^TM(-\transition^T)(\bm{\rho}^{(0)}-\bm{nx})\\
&=\sum_{r:\ell\in r}\frac{1}{\weight_r}(\bm{\rho}_r^{(0)}-\bm{n}_rx_r)^TM_r(-\transition_r^T)(\bm{\rho}_r^{(0)}-\bm{n}_rx_r)\\
&=\sum_{r:\ell\in r}\frac{1}{2\weight_r}(\bm{\rho}_r^{(0)}-\bm{n}_rx_r)^T\bigl(M_r(-\transition_r^T)+(-\transition_r)M_r\bigr)(\bm{\rho}_r^{(0)}-\bm{n}_rx_r).
\end{align*}
Then by Lemma~\ref{LEM-LARGER-NORM},
\begin{align*}
&\mspace{23mu}\sum_{r:\ell\in r}\frac{1}{2\weight_r}(\bm{\rho}_r^{(0)}-\bm{n}_rx_r)^T(M_r(-\transition_r^T)+(-\transition_r)M_r)(\bm{\rho}_r^{(0)}-\bm{n}_rx_r)\\
&\ge\sum_{r:\ell\in r}\frac{\eta_r}{\weight_r} (\bm{\rho}_r^{(0)}-\bm{n}_rx_r)^TM_r(\bm{\rho}_r^{(0)}-\bm{n}_rx_r)\\
&\ge \weight_{\min}\eta_{\min
}\|\widehat{\bm{b}}^{(\ell)}-\bm{b}^{(\ell)}\|^2,
\end{align*}
where $\eta_{\min}=\min_r\{\eta_r\}>0$ and $\weight_{\min}=\min_r\{\weight_r\}>0$. Therefore,
\begin{align*}
\langle \bm{b}^{(\ell)}-\widehat{\bm{b}}^{(\ell)},(-\transition^T)(\bm{\rho}^{(0)}-\bm{nx})\rangle\ge\weight_{\min}\eta_{\min
}\|\widehat{\bm{b}}^{(\ell)}-\bm{b}^{(\ell)}\|^2,
\end{align*}
which completes the proof. \Halmos
\endproof

\section{Proof of Lemma~\ref{LEM-SUM-NRK}.}\label{app-proof-lem-sum-nrk}
\proof{Proof.}
We first show the bounds on the drift $\Delta\|\bm{n}\|$.  By the proof of Lemma~\ref{lem-nperp-drift}, inserting \eqref{eq-bound-npart} into \eqref{eq-bound-n2-drift-3} gives
\begin{align*}
\Delta \|\bm{n}\|&\le \frac{\Delta \|\bm{n}\|^2}{2\|\bm{n}\|}\\
&\le -\epsilon\frac{\langle \bm{n},(-\transition^T)\bm{\rho}^{(0)}\rangle}{\|\bm{n}\|}+\frac{A_1}{\|\bm{n}\|}&\\
&=-\epsilon\frac{\bm{n}^TM(-\transition^T)\bm{\rho}^{(0)}}{\|\bm{n}\|}+\frac{A_1}{\|\bm{n}\|}.
\end{align*}
We bound the term $\bm{n}^TM(-\transition^T)\bm{\rho}^{(0)}$ below. First, it can be written as follows
\begin{equation}\label{eq-linear-n}
\bm{n}^TM(-\transition^T)\bm{\rho}^{(0)} = \sum_r\lambda_r^{(0)}\bm{n}_r^TM_r\bm{\pi}_r^T.
\end{equation}
Then for each term for route $r$, we have
\begin{equation*}
\lambda_r^{(0)}M_r\bm{\pi}_r^T = \weight_r\int_0^{+\infty}\frac{\exp(\transition_r\sigma)\bm{1}\bigl(\bm{\pi}_r\exp(\transition_r\sigma)\bm{1}\bigr)^T}{\bm{\pi}_r(-\transition_r)^{-1}\exp(\transition_r\sigma)\bm{1}}d\sigma.
\end{equation*}
Recall that we have proved that there exists a constant $\eta_r>0$ such that for any $\sigma\ge 0$,
\begin{equation*}
\frac{\bm{\pi}_r\exp(\transition_r\sigma)\bm{1}}{\bm{\pi}_r(-\transition_r)^{-1}\exp(\transition_r\sigma)\bm{1}}\ge 2\eta_r,
\end{equation*}
in \eqref{eq-bound-hazard} of the proof of Lemma~\ref{LEM-LARGER-NORM}. Since each entry of $\exp(\transition_r\sigma)\bm{1}$ is nonnegative, we have
\begin{align*}
\lambda_r^{(0)}M_r\bm{\pi}_r^T &\ge 2\weight_r\eta_r\int_{0}^{+\infty}\exp(\transition_r\sigma)\bm{1}d\sigma\\
&=2\weight_r\eta_r(-\transition_r)^{-1}\bm{1},
\end{align*}
where the inequality is in an entry-wise sense. Recall that $\transition_r$ is a block-diagonal matrix in the form of \eqref{eq-transition_r-block}. Then each entry of $(-\transition_r)^{-1}\bm{1}$ is greater than or equal to the corresponding $\frac{1}{\mu_r^{(b)}}$, where $\mu_r^{(b)}$'s are the rates in the phase-type distributions. Applying this bound to \eqref{eq-linear-n} and letting
\begin{equation*}
 \mu_{\max}=\max_{r,b}\mu_r^{(b)},
\end{equation*}
we have
\begin{align*}
\bm{n}^TM(-\transition^T)\bm{\rho}^{(0)} & \ge \frac{2\eta_{\min}}{\mu_{\max}}\sum_{r,\phase}\weight_rn_{r,\phase},
\end{align*}
Since norms are equivalent in $\mathbb{R}^{\Phase}$, there exist positive constants $A_3$ and $A_4$ such that for any $\bm{n}$, $A_3\|\bm{n}\| \le \sum_{r,\phase}\weight_rn_{r,\phase}\le A_4\|\bm{n}\|$.
Therefore,
\begin{equation*}
\Delta\|\bm{n}\|\le -\frac{2\eta_{\min}A_3\epsilon}{\mu_{\max}}+\frac{A_1}{\|\bm{n}\|}.
\end{equation*}
Then
\begin{equation*}
\Delta \|\bm{n}\|\le -\frac{\eta_{\min}A_3\epsilon}{\mu_{\max}},\quad\text{when }\|\bm{n}\|\ge \frac{\mu_{\max}A_1}{\eta_{\min}A_3\epsilon},
\end{equation*}
which is the first upper bound on the drift in Lemma~\ref{LEM-SUM-NRK}.

Let
\begin{gather*}
\overline{q}= \sup_{\bm{n}}(-q_{\bm{nn}}),\\
\nu=2\max_{r,\phase}\|\bm{e}^{(r,\phase)}\|,\quad\zeta=2\left(\max_{\ell}C_{\ell}\right)\cdot\left(\max_{r,\phase_1,\phase_2:\phase_1\neq\phase_2}(\transition_{r})_{\phase_1,\phase_2}\right).
\end{gather*}
It can be verified that $\overline{q}<+\infty$, and that $\nu$ and $\zeta$ are positive constants such that
\begin{gather*}
\sup_{\bm{n},\bm{n}'\colon q_{\bm{n}\bm{n}'}>0} \bigl|\|\bm{n}'\|-\|\bm{n}\|\bigr|\le \nu,\quad 
\sup_{\bm{n}}\sum_{\bm{n}'\colon \|\bm{n}\|<\|\bm{n}'\|}q_{\bm{n}\bm{n}'}\le\zeta.
\end{gather*} 
Therefore, for all $\bm{n}$,
\begin{equation*}
\Delta\|\bm{n}\|\le\nu\zeta,
\end{equation*}
which is the second upper bound on the drift in Lemma~\ref{LEM-SUM-NRK}. Then by the Foster-Lyapunov theorem, the flow count process is positive recurrent.

By Lemma~\ref{LEM-TAIL}, for any nonnegative integer $j$,
\begin{equation*}
\Pr\biggl(\|\steady\|>\frac{\mu_{\max}A_1}{\eta_{\min}A_3\epsilon }+2\nu j\biggr)\le\Biggl(\frac{\zeta \nu }{\zeta \nu +\epsilon\eta_{\min}A_3/\mu_{\max}}\Biggr)^{j+1}.
\end{equation*}
Let
\begin{equation*}
\beta=\frac{\zeta \nu }{\zeta \nu +\epsilon\eta_{\min}A_3/\mu_{\max}}.
\end{equation*}
Then,
\begin{align*}
&\mspace{23mu}\Pr\biggl(\sum_{r,\phase}\weight_r\overline{N}_{r,\phase}>\frac{\mu_{\max}A_1A_4}{\eta_{\min}A_3\epsilon }+2\nu A_4j\biggr)\\
&\le\Pr\biggl(\|\steady\|>\frac{\mu_{\max}A_1}{\eta_{\min}A_3\epsilon }+2\nu j\biggr)\\
&\le\beta^{j+1}.
\end{align*}
This completes the proof of Lemma~\ref{LEM-SUM-NRK}. \Halmos
\endproof

\end{APPENDICES}
}{%
}%
\end{document}